\pdfoutput=1
\documentclass[12pt, a4paper]{article}
\usepackage[T1]{fontenc}

\usepackage[vmargin = 1.3in, hmargin =1.3in]{geometry}
\usepackage{setspace}
\usepackage[kerning=true,protrusion=true,expansion]{microtype}
\microtypecontext{spacing=nonfrench}

\usepackage{titlesec}

\let\originalparagraph\paragraph
\renewcommand{\paragraph}[2][.]{\originalparagraph{#2#1}} 

\usepackage{mathtools, amsthm, amssymb, bm}
\mathtoolsset{centercolon}
\allowdisplaybreaks

\usepackage{mathrsfs}

\usepackage[caption=false]{subfig}
\usepackage{caption}

\DeclareCaptionLabelFormat{panel}{Panel #2}
\makeatletter
\renewcommand{\p@subfigure}{\thefigure.}
\makeatother

\usepackage{graphicx,tikz,xcolor}
\usetikzlibrary{arrows.meta,positioning,patterns,decorations.pathreplacing,calligraphy}
\usepackage{pgfplots, pgfplotstable}
\pgfplotsset{compat=newest}
\pgfplotsset{
	label style={font=\scriptsize},  
	legend style={
		font=\footnotesize,
		nodes={font=\footnotesize}
	},
	ticklabel style={font=\scriptsize},
	axis lines=middle,
	axis line style={-latex},
	ylabel style={
		above,
		xshift=5pt,
		font=\scriptsize
	},
	xlabel style={
		below,
		xshift=-14pt,
		yshift=-10pt,
		font=\scriptsize
	},
	legend pos=north west,
	width=0.392\textwidth, scale only axis,
	samples = 100,
}
\usepgfplotslibrary{fillbetween}

\definecolor{blue}{RGB}{43,147,206}
\definecolor{GmailBlue}{RGB}{42, 93, 176}

\definecolor{red}{RGB}{221,126,0}
\definecolor{orange}{RGB}{208,106,11} 

\definecolor{green}{RGB}{0,158,115} 
\definecolor{gray}{RGB}{73, 73, 73}
\definecolor{yellow}{RGB}{230,168,0}
\definecolor{pink}{RGB}{211,0,214}

\usepackage[authoryear]{natbib}
\usepackage[
backref=page,
bookmarks=true,
bookmarksnumbered=true,
bookmarksopen=true,
pdfborder={0 0 0},
breaklinks=true,
colorlinks=true,
linkcolor=GmailBlue,
citecolor=GmailBlue,
filecolor=GmailBlue,
urlcolor=GmailBlue,
hypertexnames=false,
plainpages=false,
]{hyperref}

\hypersetup{backref=true, 
	pagebackref=true, 
}

\makeatletter
\newcommand\org@hypertarget{}
\let\org@hypertarget\hypertarget
\renewcommand\hypertarget[2]{%
	\Hy@raisedlink{\org@hypertarget{#1}{}}#2%
} \makeatother

\makeatletter
\renewcommand*{\NAT@spacechar}{\ }
\AtBeginDocument{ \hypersetup{ 
		pdftitle = {Persuading an inattentive and privately informed receiver}, 
		pdfauthor = {Pietro Dall'Ara} 
} }
\makeatother

\usepackage{apptools}
\AtAppendix{\counterwithin{lemma}{section}}
\AtAppendix{\counterwithin{proposition}{section}}
\AtAppendix{\counterwithin{theorem}{section}}
\AtAppendix{\counterwithin{corollary}{section}}
\AtAppendix{\counterwithin{remark}{section}}
\AtAppendix{\counterwithin{definition}{section}}

\theoremstyle{plain}
\newtheorem{theorem}{Theorem}
\newtheorem{proposition}{Proposition}
\newtheorem{lemma}{Lemma}
\newtheorem{corollary}{Corollary}

\newtheorem{definition}{Definition}
\newtheorem{remark}{Remark}
\newtheorem{assumption}{Assumption}
\newtheorem{example}{Example}

\let\originalleft\left
\let\originalright\right
\renewcommand{\left}{\mathopen{}\mathclose\bgroup\originalleft}
\renewcommand{\right}{\aftergroup\egroup\originalright}
\newcommand*\diff{\mathop{}\!\mathrm{d}}

\DeclareMathOperator*{\argmax}{argmax}

\DeclareMathOperator{\Int}{int}

\newcommand{\Gmarg}{G_\lambda}
\newcommand{\gmarg}{g_\lambda}
\newcommand{\Gcond}{G_{c|\lambda}}
\newcommand{\gcond}{g_{c|\lambda}}

\title{\textbf{Persuading an inattentive and privately informed receiver}\thanks{I am grateful for comments from Jacopo Bizzotto, Giacomo Calzolari, Elias Carroni, Ryan Chahrour, Vincenzo Denicol\`o, Mehmet Ekmekci, Marco Errico, Jan Knoepfle, Hideo Konishi, Stephan Lauermann, Niccol\`o Lomys, Chiara Margaria, Laurent Mathevet, Teddy Mekonnen, Marco Pagnozzi, Stefano Piasenti, Guillaume Pommey, Antonio Rosato, Giacomo Rubbini, Fernando Stragliotto, Junze Sun, Utku \"Unver, Dong Wei, and Bumin Yenmez. This paper was previously circulated with the title ``The extensive margin of Bayesian persuasion.''
		}}
\author{\href{https://sites.google.com/view/pietrodallara12}{\color{black}{Pietro Dall'Ara}}\thanks{CSEF and University of Naples Federico II; contact: \href{mailto:pietro.dallara@unina.it}{pietro.dallara@unina.it}.}}

\usepackage{datetime}
\newdateformat{specialdate}{\THEDAY~\monthname[\THEMONTH] \THEYEAR}
\date{\specialdate\today}

\begin{document}	
	\maketitle
	\begin{abstract}
		This paper studies the persuasion of a receiver who accesses information only if she exerts costly attention effort. A sender designs an experiment to persuade the receiver to take a specific action. The experiment affects the receiver's attention effort, that is, the probability that she updates her beliefs. Persuasion has two margins: an extensive margin (effort) and an intensive margin (action). The receiver's utility exhibits a complementarity property in information and effort. By leveraging this property, we establish an equivalence between experiments and persuasion mechanisms \`a la Kolotilin et al.~(2017). In applications, the sender's optimal strategy involves censoring favorable states.
	\end{abstract}
	\noindent {Keywords}: Persuasion; Inattention; Information Acquisition; Information Design.\\
	{JEL codes}: D82, D83, D91.
	
	\newpage  \tableofcontents \newpage \normalsize \onehalfspacing
	\section{Introduction}\label{sec:intro}
	In the ``information age,'' consumers evaluate whether information sources are worth their attention because learning takes effort and time \citep{simon_knowledge_1996, floridi_fourth_2014}. The persuasion literature studies how a sender, such as an advertiser or media outlet, provides information to persuade a receiver to take a specific action \citep{kamenica_bayesian_2019}. When attention is costly, the sender faces an additional problem: the receiver can be persuaded only if she pays attention. This paper studies a persuasion model in which the sender's information affects the attention effort of the receiver, who privately knows the costs and benefits of information.

	The \emph{intensive} margin of persuasion refers to the sender's influence on the receiver's action, given that the receiver is attentive, whereas the \emph{extensive} margin refers to whether the receiver pays attention to the information. The study of the extensive margin is important to understand how consumers allocate attention to product advertisements and news content. The choice of attention ultimately determines the success of marketing campaigns, and the spread of information across heterogeneous audiences.

	To study the extensive and intensive margins of persuasion, we introduce an attention decision into a persuasion game with two players: Sender (he) and Receiver (she). In the first stage of the game, Sender designs a signal, a random variable that is jointly distributed with an unknown state. Receiver then chooses her attention effort after observing the signal’s distribution, but before its realization. By exerting costly effort, Receiver increases the probability of observing the realization of the signal. In the last stage of the game, Receiver takes a binary action: 1 or 0. The interests of the players conflict because Receiver chooses action 1 only if she expects the state to exceed her outside option, whereas Sender wants her to choose 1 regardless of the state. Both the outside option and the effort cost constitute a privately known {type} of Receiver. The outside option reflects the benefits of information, because it is unlikely that a piece of information is useful if the available outside option is very attractive. Similar games are applied to study the persuasion of voters, electoral manipulation, and credit-rating agencies \citep{alonso_persuading_2016AER, gehlbach_electoral_2015, bizzotto_fees_2021}.

	The correlation between the state and the signal affects both Receiver's attention effort $e$ (the extensive margin) and her action after observing the signal realization (the intensive margin). Specifically, Receiver updates her beliefs with probability $e$, and does not update with the remaining probability. The effort represents the acquisition of information, and its costs can be monetary, such as subscription fees, or cognitive, such as mental exertion. This attention model is less general than those with flexible information acquisition \citep*{caplin_rationally_2022, pomatto_cost_2023}, as Receiver only chooses the probability with which she uniformly observes every signal realization. This parsimonious model accommodates asymmetric information between Sender and Receiver, and a general functional form of the cost of effort.\footnote{Typical applications of flexible information acquisition define costs over belief distributions. The information cost in this model is ``experimental'' because Receiver effectively chooses mixtures of full and null information about the signal \citep*{denti_experimental_2022}.}

	In the model, the utility of Receiver is supermodular in information and effort (\autoref{cor:superm}). In particular, the return from effort increases in a type-specific informativeness order. This property is a complementarity between information and attention effort. Complementarity is a feature of information acquisition that is likely to arise from sources like news outlets and advertisements. For instance, this feature emerges when voters' willingness to subscribe to a newspaper increases as space dedicated to election coverage expands, and when TV audiences pay more attention to increasingly informative advertisements.
	There is empirical evidence that product awareness increases in the informative content of ads \citep*{honka_advertising_2017, tsai_informational_2021}. This paper analyzes the scope of persuasion in these settings.

	We establish the equivalence between persuasion mechanisms and signals (Theorem \ref{thm:mechanisms}). A persuasion mechanism is a menu of signals, one for every Receiver's report of her type. Under a persuasion mechanism, Receiver makes a report and chooses her effort. Specifically, Receiver chooses the probability with which she observes the signal that corresponds to her report. For every persuasion mechanism, there is a signal that induces the same action and effort of all types. The key step in the proof is to construct a signal that ``allocates'' to each type the same type-specific informativeness as the mechanism. This step establishes the equivalence with respect to effort choices. The constructed signal also replicates Receiver's optimal action; thus, the equivalence in \citet*{kolotilin_persuasion_2017} arises in the special case of costless effort. As an implication, an information provider need not offer a fine collection of targeted experiments, and the analysis of the extensive margin can focus on single signals without loss of generality.

	We characterize the optimal signal in applications and show that it censors high states. An upper censorship is a signal that reveals low states and pools high states, illustrated in Figure \ref{fig:cens:a}.
		\begin{figure}[t]
		\centering 
		\subfloat[An upper censorship is a signal that reveals the state $\theta$ if $\theta$ is below a threshold $\overline \theta$ and sends a single message, \textsf{m}, if the state is above $\overline \theta$.\protect\label{fig:cens:a}]{
			\begin{tikzpicture}[scale=0.86]
				\draw [ultra thick, densely dotted, darkgray]  (0,0) -- (4,0);
				\draw [ultra thick, blue]  (4,0) -- (6,0);
				\draw [thick] (0,-.2) -- (0, .2);
				\draw [thick] (4,-.2) -- (4, .2);
				\draw [thick] (6,-.2) -- (6, .2);
				\node[align=left, below] at (2.07, .7)%
				{\footnotesize $\theta$ is revealed};
				\node[align=right, below] at (5,.7)%
				{\footnotesize \textsf{m}};
				\node [above] at (0,-.8) {\footnotesize 0};
				\node [above] at (6,-.8) {\footnotesize 1};
				\draw [very thick , decorate, decoration = {calligraphic brace,raise=1.5pt}] (4,0) --  (6,0);
				\node [above] at (4,-0.8) {\footnotesize $\overline \theta$};
			\end{tikzpicture}
		}
		\quad
		\subfloat[A bi-upper censorship is a signal that reveals the state $\theta$ if $\theta$ is below a lower threshold $\theta_1$, sends a message \textsf{m\textsubscript{1}} if the state is between $\theta_1$ and an upper threshold $\theta_2$, and sends a different message \textsf{m\textsubscript{2}} if the state is above $\theta_2$.\protect\label{fig:cens:b}]{
			\begin{tikzpicture}[scale=0.86] 
				\draw [ultra thick, densely dotted, darkgray]  (0,0) -- (3.5,0);
				\draw [ultra thick, green]  (3.5,0) -- (5,0);
				\draw [ultra thick, red]  (5,0) -- (6,0);
				\draw [thick] (0,-.2) -- (0, .2);
				\draw [thick] (3.5,-.2) -- (3.5, .2);
				\draw  (5,-.1) -- (5, .1);
				\draw [thick] (6,-.2) -- (6, .2);
				\node[align=center, below] at (1.77,.7)
				{ \footnotesize $\theta$ is revealed};
				\node[align=right, below] at (4.33,.7)
				{ \footnotesize \textsf{m\textsubscript{1}}};
				\node[align=right, below] at (5.56,.7)
				{ \footnotesize \textsf{m\textsubscript{2}}};
				\node [above] at (0,-.8) {\footnotesize 0};
				\node [above] at (6,-.8) {\footnotesize 1};
				\draw [very thick, decorate, decoration = {calligraphic brace,raise=1.5pt}] (3.5,0) --  (5,0);
				\draw [very thick, decorate, decoration = {calligraphic brace,raise=1.5pt}] (5,0) --  (6,0);
				\node [above] at (3.5,-0.8) {\footnotesize $\theta_1$};
				\node [above] at (5,-0.8) {\footnotesize $\theta_2$};
			\end{tikzpicture}
		}
		\caption{An upper censorship (a) and a bi-upper censorship (b), for a state $\theta$ supported on $[0,1]$.}
		\protect\label{fig:cens}
	\end{figure}
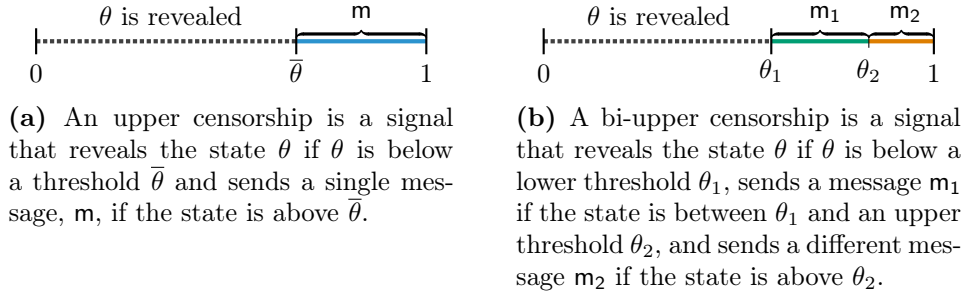
	Upper censorships are optimal if the outside option follows a single-peaked distribution (Theorem \ref{thm:existence}). In the costless-attention case, the result follows directly from the shape of the noise in Receiver's action given her posterior belief. The noise --- perceived by Sender --- is exogenous and due to asymmetric information about the outside option. Our result accounts for the endogenous noise that is due to the choice of effort.  The equilibrium upper censorship is more informative if effort costs a small amount known to Sender than if effort is costless (Proposition \ref{prop:cost}). Moreover, the equilibrium upper censorship provides more information as the effort cost increases in the reverse-hazard-rate order, if the outside option is known to Sender (Proposition \ref{prop:reverse}). We also consider an extension inspired by models of media capture. In this model, Sender values Receiver's effort directly, not only because effort ultimately affects Receiver's action.  \citet{gehlbach_government_2014} microfound these preferences by modeling advertising revenues.\footnote{See \citet{alonso_competitive_2025} for a recent contribution and literature review.} We show that ``bi-upper censorships'' are optimal signals (Proposition \ref{prop:unique}, Figure \ref{fig:cens:b}). In the proof, the additional censorship region allows Sender to separately control the extensive and intensive margins. These results suggest that attention constraints can incentivize information providers to increase the quality of their signals.

	The model captures a form of costly attention based on overall exposure to a communication environment, rather than a granular fine-tuning of specific learning activities. In contrast to models in which attention is allocated to specific signals or states, here the key friction is the overall burden of exposure to a stream of messages. We interpret the effort cost as related to the time and mental energy spent engaging with Sender’s messages, which can be high even for uninformative content. For example, consider a user who scrolls through a news feed containing product advertisements, friends' personal updates, and news about current affairs. The platform controlling the news-feed algorithm wants both to generate high attention and to steer users towards advertised products and specific news outlets. The user can read the entire feed, or adopt a brief cognitive practice, such as ``snacking'' or skimming, which processes less information and is less mentally taxing.\footnote{Cognitive load, measured via pupil responses, is higher for some headlines than others, which suggests that processing content in news feeds consumes scarce cognitive resources \citep*{shi_true_2023}.} This exposure-based cost is natural if the cognitive frictions of the user are bundled with other costs --- such as platform fees, foregone outside opportunities, and fatigue --- that are incurred regardless of content in the feed. Our results are informative for these settings, in which we study equilibrium attention, persuasion mechanisms, and how higher attention costs induce more information, in Section \ref{sec:receiver}, \ref{sec:mechanisms}, and \ref{sec:application}, respectively.\footnote{The results in Section \ref{sec:receiver} and \ref{sec:mechanisms} do not depend on the preferences of Sender, and the preferences in the hypothesis of Proposition \ref{prop:unique} capture an attention-influence tradeoff.} Other ways of incorporating rational inattention in communication admit different interpretations and capture other economic environments; see Section \ref{sec:model:discussion}.

	
	\paragraph{Related literature} 
	Existing work considers persuasion without costly information acquisition \citep[e.g.]{kamenica_bayesian_2011}. The optimality properties of upper censorships are known, and the equivalence between persuasion mechanisms and signals is shown by \citet*{kolotilin_persuasion_2017}. We generalize these results to the case of Receiver's costly effort and privately known effort cost. The key distinction is about attention costs: this model allows for costly attention, which allows for new comparative statics (Proposition \ref{prop:cost} and \ref{prop:reverse}), and new applications (Proposition \ref{prop:unique}) of the model. This model is not nested in the fruitful ``mean-measurable'' paradigm because, in equilibrium, effort is a function of the entire posterior-mean distribution, not of a single posterior mean; this observation is implied by Lemma \ref{lem:receiver} and Sender's maximand in Lemma \ref{app:lem:uniqueness}. Hence, known techniques cannot be applied off the shelf \citep*{kolotilin_optimal_2018, dworczak_simple_2019, kleiner_extreme_2021}.
	
	The persuasion of an inattentive Receiver has been studied without private information. In \citet{wei_persuasion_2021}, the attention cost is posterior separable. As a result of costly attention and symmetric information, the optimal signal is binary, and Receiver pays full attention in equilibrium. In the main model of \citet{bloedel_persuading_2021}, the attention cost is proportional to the entropy reduction of the belief, and upper censorships are optimal. The connection with these approaches is described in Section \ref{sec:model:discussion}. Certain dynamic models of persuasion include costly attention \citep*{liao_2021_bayesian, che_keeping_2023}, although the focus of these binary-state models is on the intertemporal flow of information; see also: \citet*{branco_too_2016, knoepfle_dynamic_2020, hebert_engagement_2024}. The belief of a psychological Receiver arises from an optimization problem that typically occurs after the signal is realized, rather than before the signal realization as in the choice of effort in this model \citep*{lipnowski_disclosure_2018, galperti_persuasion_2019, declippel_nonbayesian_2022, augias_persuading_2023}. In models of salience-based attention, an agent responds to salient states through belief distortions that bias attention; e.g., \citet*{ba_over_2025}. This behavioral mechanism is excluded from this model.
	
	Existing work studies information acquisition with different incentives for Sender and additional information sources. The literature on attention management considers Receiver's attention given a benevolent Sender, who maximizes Receiver's material payoff ignoring attention costs \citep*{lipnowski_attention_2020}. The literature on persuasion with acquisition of ``outside information'' studies extra information beyond what Sender provides \citep*{brocas_influence_2007, bizzotto_testing_2020, matyskova_bayesian_2023}. 
	
	
	\paragraph{Outline}
	Section \ref{sec:model} describes the model, Section \ref{sec:receiver} analyzes the equilibrium attention and action of Receiver, Section \ref{sec:mechanisms} describes the equivalence between persuasion mechanisms and signals, and Section \ref{sec:application} considers upper censorships. Section \ref{sec:model:discussion} discusses alternative approaches to inattention in information design. All proofs are in Appendix \ref{app:sec:proofs}.
	
	\section{Model} \label{sec:model}
	\subsection{Players, actions, and payoffs}
	Two players, Sender (he) and Receiver (she), play the following persuasion game. Receiver chooses action $a\in\{0,1\}$ and effort $e\in[0,1]$, knowing her type $(c, \lambda)\in[0,1]^2$. The material payoff of action $a$, given state $\theta \in \Theta \coloneqq [0,1]$, is $a(\theta -c)$, and the cost of effort $e$ is $\lambda k(e)$, for a continuous function $k\colon [0,1]\to\mathbb R$ and given a type $(c, \lambda)$. The \emph{outside-option type} $c$ represents the opportunity cost of taking the risky action, 1, and the \emph{attention type} $\lambda$ scales the effort cost. The utility of Receiver is her material payoff net of effort cost, and is given by
	\begin{align*}
		U_R(\theta, a, e,c, \lambda) := a(\theta -c) -\lambda k(e).
	\end{align*}
	Let $T \subseteq [0,1]^2$ denote the support of Receiver's type.

	Sender chooses a persuasion mechanism: a collection $(\pi_r)_{r \in T}$ of signals about the state; that is, for every report $r \in T$, $\pi_r \colon \Theta \to \bm \Delta(M)$ is measurable, letting $\bm \Delta(M)$ be the set of Borel probability distributions over the rich message space $M$.\footnote{In this model, letting $M=[0,1]$ is sufficient.} The representation of signals as convex functions used in the analysis is in Section \ref{sec:model:policies}. The utility of Sender is given by $U_S(a):=a$. The results in Section \ref{sec:mechanisms} do not depend on Sender's utility, and in Proposition  \ref{prop:unique} we consider a linear function of $a$ and $e$ as Sender's utility for comparative statics and applications.
	
	\begin{remark}
			The continuity of  $k$ ensures a well-defined equilibrium effort, and the main results on the extensive margin, persuasion mechanisms, and upper censorship --- Corollary \ref{cor:superm}, Lemma \ref{lem:interval}, Theorem \ref{thm:mechanisms}, and Theorem \ref{thm:existence} --- hold in this setting. Essentially, Corollary \ref{cor:superm} establishes enough structure to have these results, and does not depend on additional properties of the cost function, such as monotonicity. The characterization of the equilibrium signals and the comparative statics (Proposition \ref{prop:cost}, \ref{prop:reverse}, and \ref{prop:unique}) use the linear cost given by $k(e)=e$. With linear costs, the attention choice is effectively binary, that is, full or no attention. This choice captures costs that need not be of a cognitive nature, such as a market price or a fixed access fee for information. 
	\end{remark}

	\subsection{Information and timing} \label{sec:model:model}
	\paragraph{Information}
	Let $\mathcal D$ be the set of distributions over $[0,1]$ identified by their distribution functions. The state $\theta$ is distributed according to a full-support distribution $F_0	\in \mathcal D$, the \emph{prior belief}, with mean $x_0$ and a positive density. Receiver's type is independent of $\theta$ and admits a marginal distribution of the attention type $\lambda$, $\Gmarg \in \mathcal D$, and a conditional distribution of the outside option $c$ given $\lambda$, $\Gcond(\cdot|\lambda)\in \mathcal D$.

	\paragraph{Timing}
	First, Sender chooses a persuasion mechanism $(\pi_r)_{r \in T}$, without knowing either the state or Receiver's type. Second, Receiver observes the mechanism and, knowing her type $(c, \lambda)$, chooses a report $r \in T$ and an effort $e \in [0,1]$. Third, Nature draws the state $\theta$ according to $F_0$, and the signal realization from $\pi_r(\theta)$. Then, with probability $e$, Receiver observes the signal realization, updates her belief about the state using Bayes' rule, and chooses an action given her posterior belief; with probability $1-e$, Receiver does not observe the signal realization and chooses an action given the prior belief. The equilibrium notion is Perfect Bayesian Equilibrium (Appendix \ref{app:sec:eq:eq}). Figure \ref{fig:timing} illustrates the timing.		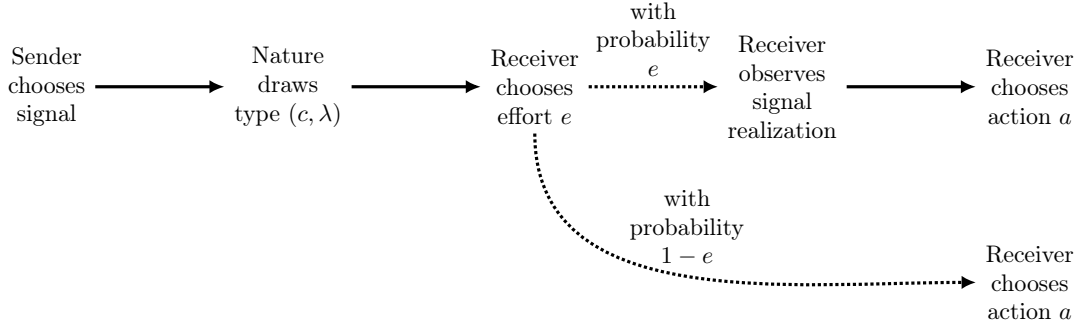
\begin{figure}[t!]
		\centering
		\begin{tikzpicture}[
			node distance=1.55cm,
			every node/.style={align=center, font=\footnotesize},
			dashedarrow/.style={-{Latex[width=1.8mm,length=1.55mm]}, very thick, densely dotted},
			arrow/.style={-{Latex[width=1.8mm,length=1.55mm]}, very thick},
			dot/.style={circle, fill, minimum size=3pt, inner sep=0pt},
			scale=0.93,
			transform shape
			]
			\node (sender) {Sender\\chooses \\ persuasion \\ mechanism};
			\node (nature1) [right=of sender] {Nature\\draws \\ type $(c, \lambda)$\\ and state $\theta$};
			\node (receiver) [right=of nature1] {Receiver\\chooses \\ report $r$ \\ and effort $e$};
			\node (obs) [right= of receiver] {Receiver\\observes\\signal \\ realization};
			\node (action) [right=of obs] {Receiver\\chooses \\ action $a$};
			\node (noobs) [right=1.9 of action, below=1.5cm of action]{Receiver\\chooses \\ action $a$};
			\draw[arrow] (sender) -- (nature1);
			\draw[arrow] (nature1) -- (receiver);
			\draw[arrow] (obs) -- (action);
			\draw[dashedarrow] (receiver) to[out=0, in=-180]  node[above, draw=none] {with \\ probability \\ $e$} (obs);
			\draw[dashedarrow] (receiver) to[out=-90, in=-180] node[above, draw=none] {with \\ probability \\ $1-e$} (noobs);
		\end{tikzpicture}
		\caption{The timeline of the game.}
		\label{fig:timing}
	\end{figure}

	\subsection{Information policies} \label{sec:model:policies}
	Without loss, signals can be represented by the distributions of the posterior belief's mean induced on a Bayesian player who observes the signal realization.\footnote{Signals can be represented by posterior-mean distributions in persuasion games that are ``mean-measurable'' (as this model) and  in which Receiver pays full attention (unlike this model). Appendix \ref{app:sec:eq} shows that this equivalence holds in this model.} Given the presence of Receiver's effort, it pays off to represent signals by the integrals of such distributions, called ``information policies''.

	Define the \emph{information policy of $F\in \mathcal D$} as the function $I_F \colon  \mathbb R_+ \to \mathbb R_+ $ such that
	\begin{align*}
		I_F \colon x \mapsto \int_0^xF(y)\diff y,
	\end{align*}
	the set of feasible distributions $\mathcal F= \{ F \in \mathcal D  \mid  I_F(1) = I_{F_0}(1),\ I_F(x) \leq I_{F_0}(x) \ \text{for all} \ x\in \mathbb R_+ \}$,
	and the set of information policies $\mathcal I = \{ I \colon \mathbb R_+ \to \mathbb R_+ \mid   I \ \text{is convex,} \ I_{\overline F} (x)\leq I(x) \leq I_{ F_0}(x)  \ \text{for all} \ x\in \mathbb R_+ \} $,
	in which $\overline F $ is the distribution putting full mass at the prior mean. Figure \ref{fig:I} illustrates the set $\mathcal I$ and the pointwise order on $\mathcal I$.
	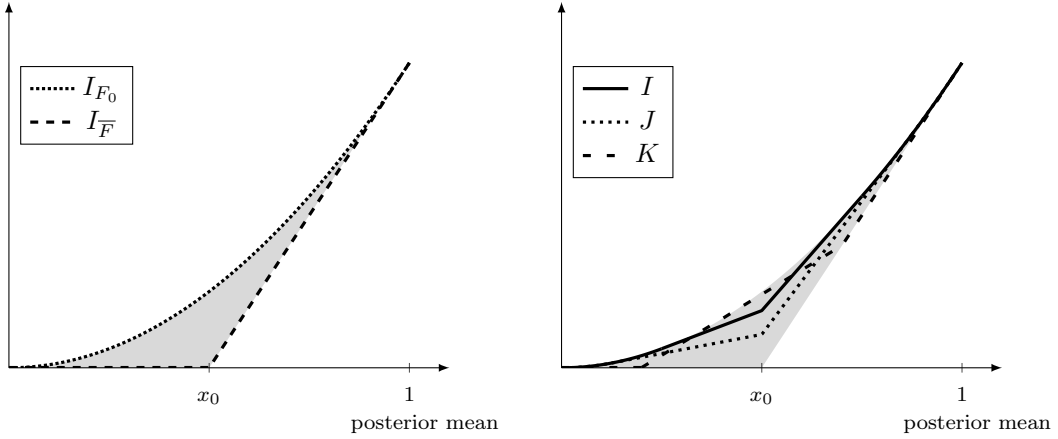
\begin{figure}[t!]
		\centering
		\subfloat[An information policy takes values only in the shaded region, which illustrates $\mathcal I$,  the set of convex functions that lie between $I_{F_0}$, corresponding to a fully informative signal, and $I_{\overline F}$, corresponding to an uninformative signal.]{
			\begin{tikzpicture}
				\begin{axis}[
					xlabel={posterior mean},
					xticklabels={$0$, $x_0$, $1$}, 
					xtick={0, 0.5, 1},
					ytick = \empty,
					yticklabels = \empty,
					xmax=1.03,
					ymax=0.5
					]
					\addplot[domain=0:1, name path = F, very thick, densely dotted] {0.5*x^2}; \addlegendentry{$I_{F_0}$};
					\addplot[domain=0:1, very thick, dashed, name path = N] {max(0.0012,x-0.5)}; \addlegendentry{$I_{\overline F}$};
					\addplot[color=black, fill=black, fill opacity=0.15] fill between [of=F and N,];
				\end{axis}
			\end{tikzpicture}
		}
		\quad
		\subfloat[Information policy $I$ is more informative than information policy $J$ because $I(x)\ge J(x)$ for all $x$. Information policies $K$ and $I$ are not comparable.\protect\label{fig:I:a}]{
			\begin{tikzpicture}
				\begin{axis}[
					xlabel={posterior mean},
					xticklabels={$0$, $x_0$, $1$}, 
					xtick={0, 0.5, 1},
					ytick = \empty,
					yticklabels = \empty,
					xmax=1.03,
					ymax=0.5
					]
					\addplot[domain=0:1/4, very thick, black] {0.5*x^2}; \addlegendentry{$I$};
					\addplot[domain=0:1/8, very thick, dotted, black] {0.5*x^2}; \addlegendentry{$J$};
					\addplot[domain=0:0.2, very thick, loosely dashed, black] {0.0012}; \addlegendentry{$K$};	
					\addplot[domain=1/4:3/4, very thick, black] {max(1/32 + (1/4)*(x-1/4), 9/32 + (3/4)*(x-3/4))};
					\addplot[domain=3/4:1, very thick, black] {0.5*x^2};
					\addplot[domain=1/8:7/8, very thick, dotted, black] {max(1/128 + (1/8)*(x-1/8), 49/128 + (7/8)*(x-7/8))};
					\addplot[domain=7/8:1, very thick, dotted, black] {0.5*x^2};  
					\addplot[domain=0.2:1, very thick, loosely dashed, black] {min(max(0.0012+0.4*(x-0.2) , max(0.0012,x-0.5)), 0.5*x^2) };
					\addplot[domain=0:1, name path = F, very thick, dashed, opacity=0] {0.5*x^2}; 
					\addplot[domain=0:1, name path = N, very thick, densely dotted, opacity=0] {max(0.0012,x-0.5)}; 
					\addplot[color=black, fill=black, fill opacity=0.15] fill between [of=F and N,];
				\end{axis}
			\end{tikzpicture}
		}
		\caption{Panel (a) illustrates the set of information policies, a panel (b) illustrates the pointwise order of information policies; the prior is uniform for all figures.\protect\label{fig:I:b}}
		\protect\label{fig:I}
	\end{figure}

	We identify signals with information policies using the results of \citet{gentzkow_rothschild-stiglitz_2016} and \citet{kolotilin_optimal_2018} (Appendix \ref{app:sec:eq:withoutloss}), and we identify persuasion mechanisms with sequences of information policies. We treat the right-derivative  $I'$ of an information policy $I$ as the distribution function of the posterior mean if Receiver exerts effort 1 and chooses $I$ from the persuasion mechanism. Specifically, Sender chooses $(I_r)_{r \in T} $ in the first stage of the game, Receiver chooses a report $r$, and the posterior mean of Receiver is drawn from $I_r'$ with a probability corresponding to her effort, and is equal to $x_0$ with the remaining probability. 
	
	\begin{definition} An \emph{equilibrium} is a tuple $\langle (I^*_r)_{r \in T}, \rho(\cdot), \eta(\cdot), \alpha \rangle$, in which $(I^*_r)_{r \in T} $ is Sender's persuasion mechanism, $\rho(c, \lambda, (I_r)_{r \in T}) \in T$ is Receiver's report given type $(c, \lambda)$ and mechanism $(I_r)_{r \in T}$, $\eta(c, \lambda, (I_r)_{r \in T})\in[0,1]$ is Receiver's effort given type $(c, \lambda)$ and mechanism $(I_r)_{r \in T}$, and $\alpha(c, \lambda, x)\in[0,1]$ is the probability that Receiver chooses action 1 given type $(c, \lambda)$ and posterior mean $x$, in a Perfect Bayesian Equilibrium (Appendix \ref{app:sec:eq:eq}).
	\end{definition}

	\paragraph{Notation}
	We let $\partial I(x)$ denote the subdifferential of $I\in \mathcal I$ at $x>0$, and use $I'(x^-)=\lim _{\varepsilon \to 0^+}I'(x-\varepsilon)$. By definition, it holds that $\partial I(x) = \{m\in\mathbb R \mid  \text{for all} \ x' \in\mathbb R_+, \ I(x') \ge I(x) + m(x'-x)\}=[I'(x^-), I'(x)]$ \citep{hiriart-urruty_fundamentals_2004}. The function $g\colon \mathbb R^2\to \mathbb R$ exhibits \emph{strictly increasing differences} if, for all $s',\,s \in \mathbb R$ with $s< s'$, the function $t\mapsto g(s', t)-g(s, t)$ is increasing.

	\section{Receiver's equilibrium action and effort} \label{sec:receiver}

	\subsection{Receiver's action and effort}
	This section studies Receiver's equilibrium action and effort for a given type $(c, \lambda)$ and information policy.
	
	Given the posterior mean $x$, Receiver chooses action 1 if $x>c$ and action 0 if $x<c$. Because $\theta \mapsto U_R(\theta, a, e, c, \lambda) $ is affine, the expected utility from choosing the action optimally given posterior mean $x$ is $	U_R(x, e, c, \lambda) = \max _{a \in \{0,1\}} U_R(x, a, e, c, \lambda)$, and we have $U_R(x, e, c, \lambda)=(x-c)_+-\lambda k (e)$.

	To characterize the equilibrium effort, we define the \emph{marginal benefit of effort given information policy $I$} as the difference in expected material payoff with and without the information in $I$; i.e., $B(I, c)=\int_{[0,1]} (x-c)_+\diff I'(x) - (x_0-c)_+$. The marginal benefit of effort given $I$ is also referred to as the value of information in the literature. The \emph{net informativeness of information policy $I$} is the difference between $I$ and the uninformative-signal information policy, $I_{\overline F}$ (Figure \ref{fig:net}). 
	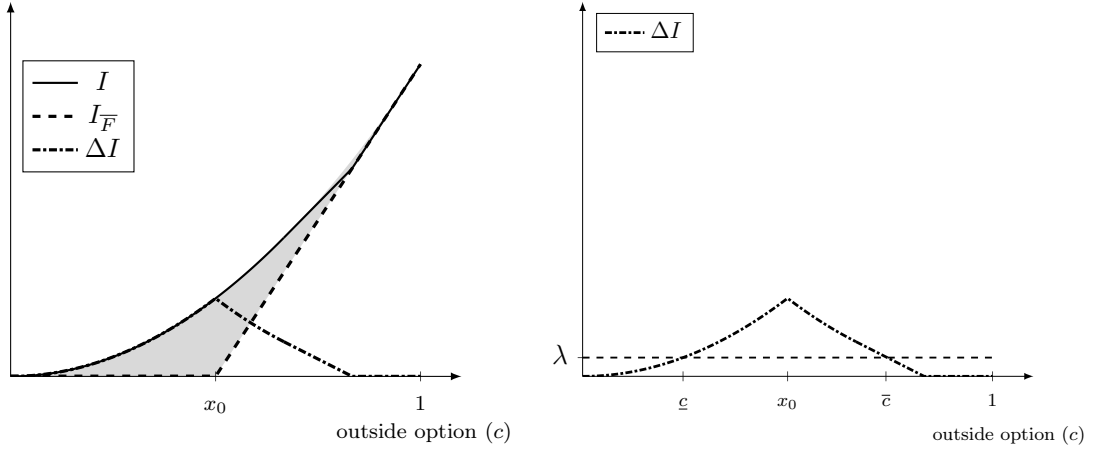
\begin{figure}[t!]
		\centering 
		\subfloat[The net informativeness of $I$ at outside option $c$, $\Delta I(c)$, is found by subtracting the value of the uninformative-signal information policy at $c$, $I_{\overline F}(c)$, from $I(c)$. The function $ \Delta I$ is single-peaked, and the prior mean $x_0$ is the peak.\protect\label{fig:net}]{
		\begin{tikzpicture}
		\begin{axis}[
				xlabel={outside option},
				xticklabels={$0$, $x_0$,  $1$}, 
				xtick={0, 0.5,  1},
				ytick = \empty,
				yticklabels = \empty,
				xmax=1.03,
				ymax=0.5
				]
				\addplot[domain=0:0.5, very thick, solid] {max(0.0012, 0.5*(x-0.25))}; \addlegendentry{$I$};
				\addplot[domain=0:1, name path = N, very thick, dashed] {max(0.0012,x-0.5)}; \addlegendentry{$I_{\overline F}$};
				\addplot[domain=0:0.5, ultra thick, densely dashdotted] {max(0.0012, 0.5*(x-0.25))- max(0.0012,x-0.5)}; \addlegendentry{$\Delta I$};
				\addplot[domain=0.5:1, very thick, solid] {0.5*x^2};
				\addplot[domain=0.5:1, ultra thick, densely dashdotted] {0.5*x^2-max(0.0012,x-0.5)};
				\addplot[domain=0:1, name path = F, very thick, dashed, opacity=0] {0.5*x^2}; 
				\addplot[color=black, fill=black, fill opacity=0.15] fill between [of=F and N,];
			\end{axis}
		\end{tikzpicture}
		}\quad 
		\subfloat[Given $k(e)=e$ and information policy $I$, the marginal benefit of effort equals the marginal cost $\lambda$ for outside-option types $\underline c$ and $\overline c$ (Lemma \protect\ref{lem:receiver}). Receiver chooses effort 1 if $c\in(\underline c, \overline c)$, and does not exert effort if \protect$c\in {[0,1]}\setminus { [\underline c, \overline c]}$.\protect\label{fig:ext}]{
		\begin{tikzpicture}
			\begin{axis}[
				xlabel={outside option},
				xticklabels={$0$, $\underline c$, $x_0$, $\overline c$, $1$}, 
				xtick={0,0.245, 0.5, 0.74, 1},
				ytick = {0.03},
				yticklabels={$\lambda$},
				xmax=1.03,
				ymax=0.5
				]
				\addplot[domain=0:0.6667, ultra thick, densely dashdotted, black] {0.5*x^2 - max(0.0012,x-0.5)}; \addlegendentry{$\Delta I$};
				\addplot[domain=0.6667:1, ultra thick, densely dashdotted, black] {max((2/9)+(2/3)*(x-0.6667), max(0.0012,x-0.5)) - max(0.0012,x-0.5)};
				\addplot [domain=0:1, dashed, color = black, thick] {0.03};
			\end{axis}
		\end{tikzpicture}
	}
		\caption{Panel (a) illustrates the construction of the net informativeness, and panel (b) illustrates the subset of outside-option types that exert positive effort, given the information policy $I$ and linear $k$.\protect\label{fig:netinfo}}
	\end{figure}
	The following result shows that the marginal benefit of effort is given by the net informativeness evaluated at $c$, using the operator $\Delta\colon I \mapsto I-I_{\overline F}$ to express the net informativeness succinctly.
	\begin{lemma}[Net informativeness]\label{lem:value}
		Given an information policy $I$,  the marginal benefit of effort is equal to the net informativeness evaluated at $c$, that is,
		\begin{align*}
		B(I, c)= \Delta I(c).
		\end{align*}
	\end{lemma}
	The net informativeness $\Delta I$ is single-peaked, with a peak at the prior mean $x_0$, by construction, as in Figure \ref{fig:ext}. Intuitively, extreme outside-option types benefit the least from the signal because they are the most certain about the optimal action at the prior belief.
	
	The following result characterizes Receiver's equilibrium action and effort.
	\begin{lemma}[Receiver's rationality]\label{lem:receiver}
		If $\langle I^*_\bullet, \rho, \eta, \alpha \rangle$ is an equilibrium, then, for every type $(c, \lambda)$ and mechanism $I_\bullet$:
		\begin{enumerate}
			\item The probability of action 0 given the induced policy $I = I_{\rho(c, \lambda, I_\bullet)}$ is above $I'(c^-)$ and below $I'(c)$, that is, 
			\begin{align*}
				1-\int _{[0,1]} \alpha(c, \lambda, x) \diff I'(x) \in \partial I(c);
			\end{align*}
			\item The effort $\eta(c, \lambda, I_\bullet)$ maximizes the expected payoff of Receiver, that is, 
			\begin{align*}
				\eta(c, \lambda, I_\bullet) \in \argmax_{e \in [0,1]}e\Delta I_{\rho(c, \lambda, I_\bullet)} (c) -\lambda k(e).
			\end{align*}
		\end{enumerate}
	\end{lemma}
	
	Part 1 simply states the equilibrium conditions that the probability of action 0 satisfies. Let $x$ be the mean of the belief of Receiver. Receiver takes action 1 with probability 1 if $x> c$, does not take action 1 if $x<c$, and may mix if $x=c$. It follows that $1-I'(c^-)$ is the probability that the action is 1 assuming Receiver chooses 1 if indifferent, that $1-I'(c)$ is the probability of 1 assuming Receiver chooses 0 if indifferent, and any element in the subdifferential of $I$ at $c$ corresponds to a tie-breaking rule.\footnote{The effort of indifferent types is not relevant in equilibrium with an atomless outside option (Lemma \ref{app:lem:uniqueness}).}
	
	The key implication of Lemma \ref{lem:receiver} is part 2. Although Receiver's effort strategy $\eta$ is formally a function of the entire mechanism $I_\bullet$, the joint maximization over reports and efforts sequentially collapses. Because the unchosen policies in the menu are payoff-irrelevant, the equilibrium effort effectively depends only on the information policy $I$ induced by the optimal report. Consequently, the net informativeness of $I$ at the outside option is a sufficient statistic for the effort decision. As an implication, the two dimensions of Receiver's type, $c$ and $\lambda$, represent her private information about, respectively, the benefit and cost of attention.

	\subsection{Complementarity between effort and net informativeness}
	This section studies the complementarity between effort and net informativeness in Receiver's optimal-effort problem.

	Receiver's \emph{value of information policy $I$}, given type $(c, \lambda)$ and effort $e$, is $V_\lambda(e, \Delta I(c)):=e\Delta I(c) -\lambda k(e)$. The value of $I$ exhibits strictly increasing differences in net informativeness and effort by Lemma \ref{lem:receiver}.
	\begin{corollary}[Complementarity]\label{cor:superm}
		Receiver's value of information policy $I$, $V_\lambda(e, \Delta I(c))$, exhibits strictly increasing differences in net informativeness $\Delta I(c)$ and effort $e$.
	\end{corollary}	
	By Lemma \ref{lem:receiver}, the equilibrium effort maximizes the value of Sender's information policy, given type $(c, \lambda)$. As an implication, a more informative Sender's information policy makes Receiver better off. In particular, we say that $I$ is more informative than $J$ if: $I(x)\ge J(x)$ for all $x\in \mathbb R_+$. If $I$ is more informative than $J$, then $I$ allocates more net informativeness to every type than $J$. Finally, by the increasing-differences property and the envelope theorem (Lemma \ref{app:lem:ET}), Receiver is better off facing $I$ than $J$.\footnote{This observation also arises as an implication of Blackwell's theorem; Corollary \ref{cor:superm} is a stronger result that we use in Section \ref{sec:mechanisms}.}

	The following result characterizes the set of types that exert positive effort given a single signal. In particular, for this result, Sender chooses a degenerate mechanism that associates the same information policy with every report.
	\begin{lemma}[Interval structure]\label{lem:interval}
		Let $\langle (I^*_r)_{r \in T}, \rho, \eta, \alpha \rangle$ be an equilibrium. Given an information policy $J$ and attention type $\lambda$, and the mechanism $J_\bullet$ such that $J_r = J$ for all $r \in T$, define the function $\eta_\lambda \colon c \mapsto \eta(c, \lambda, J_\bullet)$. The set $\eta^{-1}_\lambda((0, 1])$ is a nonempty interval if type $(x_0, \lambda)$ chooses positive effort, i.e., $\eta_\lambda (x_0)>0$, and is empty otherwise.
	\end{lemma}

	For intuition, assume linear effort cost, i.e., $k(e) = e$, and fix information policy $I$. Receiver compares the marginal cost and marginal benefit of effort. As shown in Figure \ref{fig:ext}, exerting effort 1 is optimal only if $\Delta I(c) \ge \lambda$, and no effort is optimal only if $\Delta I(c) \le \lambda$; moreover, the net informativeness of $I$ at an outside option is single-peaked as a function of the outside option. Because $\Delta I$ is single-peaked and the effort is nondecreasing in net informativeness, the set of outside-option types that exert positive effort is an interval. The proof of Lemma \ref{lem:interval} generalizes the first part of the argument: the optimal effort is nondecreasing in net informativeness by Corollary \ref{cor:superm}.

	\section{Equivalence between persuasion mechanisms and signals} \label{sec:mechanisms}
	This section studies the equivalence between information policies and persuasion mechanisms.

	We focus on equilibria in which Receiver truthfully reports her type. 
	\begin{definition}\label{def:mechanisms}
	A persuasion mechanism $I_\bullet$ is \emph{incentive compatible} (IC) if, 		for every type $(c, \lambda)$ and report $r$:
		\begin{align*}
			\max _{e\in [0,1]} V_\lambda(e, \Delta I_{(c, \lambda)}(c)) \ge \max _{e\in [0,1]} V_\lambda(e, \Delta I_{r}(c)).
		\end{align*}
	\end{definition}
	Restricting attention to IC persuasion mechanisms is without loss of generality by a revelation-principle argument (Lemma \ref{app:lem:revelation_direct}). Specifically, in the Appendix we show that for every equilibrium $E$ with persuasion mechanism  $I_\bullet$, there exist an IC persuasion mechanism $\hat I_\bullet$ and an equilibrium $\hat E$ in which: Sender obtains the same expected payoff as in $E$, Sender chooses $\hat I_\bullet$, and Receiver reports truthfully.
	
	\begin{remark}
		
		In Appendix \ref{app:sec:rev}, we consider a more general model than the present one, in order to study more general screening devices than persuasion mechanisms. In particular, we consider ``coordination mechanisms'' \`a la \citet{myerson_optimal_1982}, with which Sender assigns an information policy and an effort recommendation to every report.  Importantly, a coordination mechanism allows Sender to mix by assigning a distribution over effort recommendations and information policies to every report. Persuasion mechanisms can be viewed as a particular kind of coordination mechanism with a constant effort recommendation. In fact, because we do not restrict the equilibrium behavior of Receiver after her report, the equilibrium outcomes induced by a coordination mechanism with a constant effort recommendation are those induced by a persuasion mechanism that assigns, to every report, the ``expected'' information policy induced by the coordination mechanism and the report.

		Persuasion mechanisms are without loss of generality within deterministic coordination mechanisms, but with loss otherwise.\footnote{By Proposition \ref{app:prop:revelation_product}, persuasion mechanisms are without loss within the set of coordination mechanisms that require the effort and information policy to be independent given any report, and by Example \ref{app:example:knownoutside} persuasion mechanisms are with loss.} Persuasion mechanisms are without loss also under certain conditions on $k$ and the type distribution.\footnote{Proposition \ref{app:prop:linear} shows that persuasion mechanisms are without loss under Assumption \ref{ass:prop:reverse}; i.e., the same hypothesis as Proposition \ref{prop:reverse}.} Appendix \ref{app:sec:rev} provides these results, discusses the economic implications, and argues that the possibility of randomizing and correlating the information policy and the effort recommendation is of limited interest for the applications of this model. Specifically, exploiting the additional value of coordination mechanisms requires a communication protocol in which Sender reveals the realized quality of the information before Receiver commits to her exposure, in a way that is not consistent with the natural structure of news media consumption.
	\end{remark}

	We say that a persuasion mechanism $I_\bullet$ is equivalent to an information policy $J$ if: all types choose the same action and effort under truthful reporting given $I_\bullet$ as in some equilibrium in the subgame in which Sender assigns the information policy $J$ to every report.
	\begin{definition} \label{def:mechanisms:eq}
		An IC persuasion mechanism $I_\bullet$ is \emph{equivalent to information policy $J$} if, for every type $(c, \lambda)$:
		\begin{align*}
			\text{1.}& \ \argmax_{e \in [0,1]} V_\lambda (e, \Delta I_{(c, \lambda)}(c)) \subseteq \argmax_{e \in [0,1]} V_\lambda (e, \Delta J(c)),
			\\
			\text{2.}& \  \text{if} \ (0, 1] \cap	\argmax_{e \in [0,1]} V_\lambda (e, \Delta I_{(c, \lambda)}(c))  \ne \emptyset, \ \text{then} \ \partial I_{(c, \lambda)}(c) \subseteq \partial J(c).
		\end{align*}
	\end{definition}
	Part 1 requires that any effort that is optimal under truthful reporting in the mechanism $I_\bullet$ is also optimal if Sender chooses the single information policy $J$. Part 2 is a requirement only for types that are attentive with positive probability under the truthful equilibrium: if a type chooses action 1 with some probability under truthful reporting, then she chooses action 1 with the same probability in some equilibrium of the subgame starting with $J$. If effort is costless, then it is without loss to consider that all types choose effort 1. In this case, the nontrivial part of Definition \ref{def:mechanisms:eq} is the second one, all types choose action 1 with the same probability, which is the definition adopted in \citet*[p.\ 1954]{kolotilin_persuasion_2017}.

	 The novelty of the definition is the first part, which requires type $(c, \lambda)$ to choose the same effort under the mechanism as under the equivalent signal. Part 2 of Definition \ref{def:mechanisms:eq} does not deal with types who exert effort 0 under truthful reporting given $I_\bullet$. The reason is that the equilibrium action given the prior belief does not depend on the available information.\footnote{Formally, the reason is that the equivalence of the action decision holds as a consequence of part 1 ``for this type.'' Specifically, $\argmax_{e \in [0,1]} V_\lambda (e, \Delta I_{(c, \lambda)}(c))\allowbreak=\allowbreak \{0\}$ implies that $0\allowbreak \in \allowbreak \argmax_{e \in [0,1]} V_\lambda (e, \Delta J(c))$ by part 1, and the optimal action at the prior belief given $I_\bullet$ is the same as given $J$, possibly via equilibrium selection.}
	
	The following result shows that every IC persuasion mechanism is equivalent to a signal.
	\begin{theorem}\label{thm:mechanisms}
		For every IC persuasion mechanism $I_\bullet$, there exists an information policy $J$ such that  $I_\bullet$ is equivalent to $J$.
	\end{theorem}

	We sketch the intuition for Theorem \ref{thm:mechanisms}, which relies on the property of increasing differences of Receiver's payoff, see Corollary \ref{cor:superm}.\footnote{This sketch leverages Corollary \ref{cor:superm}, and Appendix \ref{app:sec:mechanisms} verifies that supermodularity is key by establishing the result for general payoff functions, without using the linearity of $e\cdot \Delta I(c)$.} We claim that an IC mechanism $I_\bullet$ is equivalent to its upper envelope $J$, defined by $J\colon x\mapsto \sup _{r\in T}I_r(x)$ (Figure \ref{fig:sup:A}).  Fix a type $(c, \lambda)$ that exerts positive effort under truthful reporting in the mechanism.
	\begin{figure}[t!]
		\centering
		\subfloat[The upper envelope $J$ of the information policies in the persuasion mechanism $I_\bullet=(I, L, K)$.\protect\label{fig:sup:A}]{
		\begin{tikzpicture}
			\begin{axis}[
				xlabel={posterior mean},
				xticklabels={$0$, $x_0$, $1$}, 
				xtick={0, 0.5, 1},
				ytick = \empty,
				yticklabels = \empty,
				xmax=1.03,
				ymax=0.5
				]
				\addplot[domain=0:1, thick, black] {max(0.0006,x-0.5, 0+ (1/4)*(x-1/7), 0+ (1/2)*(x-1/4), 0+ (4/5)*(x-0.42))}; \addlegendentry{$J$};
				\addplot[domain=0:1, very thick, loosely dotted, black] {max(0.0006,x-0.5, 0+ (1/4)*(x-1/7))}; \addlegendentry{$I$};
				\addplot[domain=0:1, very thick, loosely dashed, black] {max(0.0006,x-0.5, 0+ (1/2)*(x-1/4))}; \addlegendentry{$L$};
				\addplot[domain=0:1, very thick, densely dotted, black] {max(0.0006,x-0.5, 0+ (4/5)*(x-0.42))};\addlegendentry{$K$};
				\addplot[domain=0:1, name path = F, very thick, dashed, opacity=0] {0.5*x^2}; 
				\addplot[domain=0:1, name path = N, very thick, densely dotted, opacity=0] {max(0.0012,x-0.5)}; 
				\addplot[color=black, fill=black, fill opacity=0.15] fill between [of=F and N,];
			\end{axis}
		\end{tikzpicture}
		}
		\quad
			\subfloat[The $\overline \theta$ upper censorship equals the information policy $I_{F_0}$ for all $x\le \overline \theta$, has constant slope $F_0(\overline \theta)$ at $x\in(\overline \theta,  x_{\overline \theta})$, and equals $I_{\overline F}$ for all $x\ge  x_{\overline \theta}$.\protect\label{fig:sup:B}]{
		\begin{tikzpicture}
			\begin{axis}[
				xlabel={posterior mean},
				xticklabels={$0$,  $\overline \theta$, $x_0$,  $x_{\overline \theta}$, $1$}, 
				xtick={0, 0.4, 0.5, 7/10, 1},
				ytick = \empty,
				yticklabels = \empty,
				xmax=1.03,
				ymax=0.5
				]
				\addplot[domain=0:0.4, thick, solid] {0.5*x^2}; \addlegendentry{$I_{\overline \theta}$};
				\addplot[domain=0:1, name path = N, very thick, dashed] {max(0.0012,x-0.5)}; \addlegendentry{$I_{\overline F}$};
				\addplot[domain=0:1, very thick, densely dotted, black] {0.5*x^2}; \addlegendentry{$ I_{F_0}$};
				\addplot[domain=0.4:1, name path = C, thick, solid] {max((2/25)+(2/5)*(x-0.4), max(0.0012,x-0.5))};
				\addplot[domain=0:1, name path = F, very thick, dashed, opacity=0] {0.5*x^2}; 
				\addplot[color=black, fill=black, fill opacity=0.15] fill between [of=F and N,];
			\end{axis}
		\end{tikzpicture}
		}
		\caption{Panel (a) illustrates the upper envelope $J$ of the information policies in a persuasion mechanism $I_\bullet$; panel (b) illustrates the $\overline \theta$ upper censorship, letting $ x_{\overline \theta}$ denote the conditional expectation of $\theta$ given that $\theta\ge\overline \theta$.
		}
		\label{fig:sup}
	\end{figure}
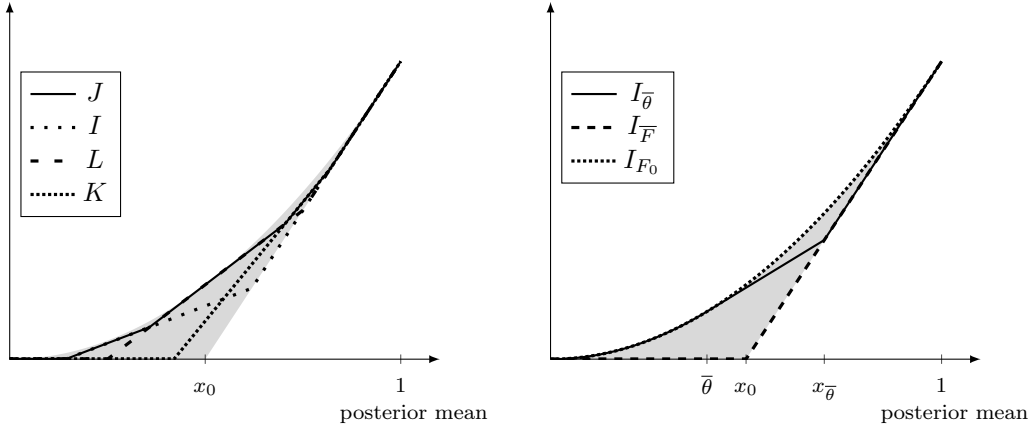
	A report $r$ is \emph{active} at $x$ if: $I_r(x)\ge I_{r'}(x)$ for all $r'\in T$. The first step is to show that any active report at $c$ maximizes Receiver's expected utility. By Lemma \ref{lem:receiver}, a report affects utility only through the net informativeness, i.e., $\Delta I_r(c)$. By increasing differences, an active report at $c$ makes Receiver weakly better off than any other report (Corollary \ref{cor:superm}, using the envelope theorem applied to the maximization that determines effort; see Lemma \ref{app:lem:ET}). Hence, an active report at $c$ maximizes expected utility at the reporting stage.

	Towards the equivalence with respect to effort, we strengthen the observation: Receiver is \emph{strictly} better off with an active report than with an inactive one. This conclusion uses both the fact that Corollary \ref{cor:superm} establishes \emph{strictly} increasing differences and the hypothesis that type $(c, \lambda)$ exerts positive effort (Lemma \ref{app:lem:ET}). Hence, by incentive compatibility, a truthful report is active. Therefore, we have that $J(c)=I_{(c, \lambda)}(c)$, which implies that $\Delta J(c)=\Delta I_{(c, \lambda)}(c)$. To complete the argument, we apply Lemma \ref{lem:receiver}: the net informativeness is the only component of the information policy that affects the effort decision. We conclude that Receiver has the same set of optimal efforts under the truthful equilibrium of mechanism $I_\bullet$ as under the signal $J$. The proof uses a continuity argument to cover the case of zero effort.

	The equivalence with respect to action decisions follows from a simple observation. To show that a given type $(c, \lambda )$ chooses the same action under $J$ as under $I_r$, it suffices to show that $\partial I_r(c)\subseteq \partial J(c)$, using Definition \ref{def:mechanisms:eq}. The inclusion holds if report $r$ is active at $c$, because the upper envelope expands the subdifferential of any active information policy, see Figure \ref{fig:sup:A}. Finally, the preceding discussion shows that the truthful report $(c, \lambda)$ is active.

	Theorem  \ref{thm:mechanisms} guarantees that the characterization of the extensive margin of persuasion in Section \ref{sec:receiver} holds in more general environments, including applications in which multiple information structures are available to decision-makers. The supplementary appendix (Section \ref{app:sec:supp}) provides a numerical example that illustrates this result.
	
	
	\paragraph{On the connection with \citet*{kolotilin_persuasion_2017}}
	\citeauthor*{kolotilin_persuasion_2017} (KMZL) establish Theorem \ref{thm:mechanisms} in the case of costless attention; i.e.,  \(\lambda\) takes only value \(0\) and Receiver chooses effort \(e=1\). In this case, the only relevant part of Definition \ref{def:mechanisms:eq} is the second: for every outside option \(c\), the probability of action \(0\) induced by the information policy \(I_{(c, 0)}\) lies in the subdifferential of the information policy \(J\). Theorem \ref{thm:mechanisms} extends this equivalence to a setting with costly effort and a privately known attention cost.

	The main step in the proof of KMZL is a characterization of the interim utilities that are implementable by a mechanism --- functions of the form $\overline U \colon c \mapsto \int _{[0,1]}(x-c)_+\diff I'(x)$, in which \(I \) is the information policy $I_{(c, 0)}$ for some IC mechanism \(I_\bullet \). Implementable utilities are closely related to the marginal benefit of effort in IC mechanisms, as formalized by Lemma \ref{lem:value}. Specifically, if \(\overline U\) is implementable and \(\underline  U(c)\) is the utility of type \((c,0)\) from the uninformative information policy, then \(\overline U(c)-\underline  U(c) = \Delta I_{(c, 0)} (c)\), which is the marginal benefit of effort generated by \(I_{(c, 0)}\) at outside option \(c\), for the ``inducing'' IC mechanism $I_\bullet$. Thus, the implementable utility $\overline U(c)$ corresponds to $\Delta I_{(c, 0)}(c)$ up to the addition of $I_{\overline F}(c) +x_0-c$, a term that is constant in $I$. In fact, this observation proves Lemma \ref{lem:value}.

	The key step in the proof of Theorem \ref{thm:mechanisms}  is to show that the upper envelope of a mechanism induces the same effort choice for all types, rather than a characterization of implementable utility functions. Implementable utility functions in our setup are functions of the form  $(c, \lambda)\mapsto \max _{e\in[0,1]}V_{\lambda}(e, \Delta I_{(c, \lambda)}(c))+(x_0-c)_+$. The primary role of information policies in the proof of Theorem \ref{thm:mechanisms} is to discipline effort via Lemma \ref{lem:value} and Corollary \ref{cor:superm}: the complementarity between information and attention implies that replacing a mechanism with its upper envelope increases the marginal benefit of effort (Figure \ref{fig:sup:A}). Hence, this change preserves the optimal effort of truthful revelation in an IC mechanism, as described in the preceding paragraphs. This step is specific to environments with costly attention, and is absent in KMZL.

	Theorem~\ref{thm:mechanisms} is stated for a model with two important features: (i) the state-dependent component of Receiver's payoff, the material payoff, depends on a \emph{binary action}; (ii) every type updates beliefs in the same way, i.e., $(c, \lambda)$ and $\theta$ are mutually independent, referred to as \emph{type-state independence}. These restrictions are essential for the result, and common in the applications cited in Section \ref{sec:application}. KMZL show that persuasion mechanisms can implement a strictly larger set of utilities than signals if Receiver chooses among multiple actions, even with type-state independence  and costless attention (see their Appendix~A). Our proof uses the same binary-action setup as KMZL. Thus, an extension of our model with richer action choices does not admit Theorem \ref{thm:mechanisms} in its present form. Example 2 in \citet{guo_interval_2019} shows that, in general, the equivalence does not hold without type-state independence, with binary action and costless attention.  Proposition 5.1 in \citet{guo_interval_2019} provides a sufficient condition for the equivalence that allows for correlation between the type and the state. Correlation is important because it allows us to capture settings with private information of Receiver about the state. The study of conditions for the equivalence in more general models with costly attention, multiple actions, and stochastic dependence between the state and the type of Receiver is an open research avenue.

	\section{Optimality properties of upper censorships} \label{sec:application}
	
	This section discusses the properties of the following class of information policies, which are illustrated in Figure \ref{fig:cens:a} and \ref{fig:sup:B}.
	
	\begin{definition}
		Given a cutoff $\overline \theta \in [0, 1]$, the $\overline \theta$ \emph{upper censorship} is the unique information policy $I_{\overline \theta }\in \mathcal I$ such that, for all $x\in\mathbb R_+$,
		\begin{align*}
			I_{\overline \theta}(x) = \begin{cases}
				I_{F_0}(x), & \text{if} \ x \in [0,\overline  \theta], \\
				\max\{ I_{F_0}(\overline \theta) +F_0(\overline \theta) (x-\overline \theta) , I_{\overline F}(x)\}, & \text{if} \ x \in (\overline \theta, \infty ).
			\end{cases}
		\end{align*}
	\end{definition}
	
		Under the following assumption, we restrict attention to type distributions that are single-peaked in the outside option and with the same peak for all attention costs.

	\begin{assumption}\label{ass:1}
		For all $\lambda$, the conditional distribution of the outside-option type given attention type $\lambda$ admits a density function $\gcond(\cdot|\lambda)$ that is absolutely continuous and satisfies $
		\int_{[0,1]}\int_{[0,1]}
		\left|
		\frac{\partial \gcond}{\partial c}(c\mid\lambda)
		\right|
		\diff c\diff \Gmarg(\lambda)<\infty$. Moreover, there exists an outside option $p\in(0,1)$ such that: for all $\lambda$, $\gcond(\cdot|\lambda) $ is nondecreasing on $[0,p]$ and nonincreasing on $[p,1]$.
	\end{assumption}

		Single-peaked distributions of the outside option are relevant for applications \citep*{romanyuk_cream_2019, kolotilin_censorship_2022, gitmez_informational_2023, augias_persuading_2023, shishkin_evidence_2023, sun_publicly_2024}. The class of single-peaked distributions includes the standard uniform and the truncated normal. The continuity restriction rules out the symmetric-information benchmark, i.e., the case in which Sender knows the type, treated in Appendix \ref{app:sec:KG}. The restriction of a constant peak in the attention cost helps with tractability to obtain the comparative statics in Proposition \ref{prop:cost} and \ref{prop:unique}. If the density has a peak that varies in $\lambda$, then the construction that we outline in what follows does not identify an upper censorship that improves upon any information policy, which is the key step for Theorem \ref{thm:existence}.
	
	We first establish that an equilibrium exists in which Sender uses an upper censorship, and that the expected utility of Sender in equilibrium is unique.
	\begin{theorem} \label{thm:existence}
		Let Assumption \ref{ass:1} hold. There exists an equilibrium in which Sender's persuasion mechanism assigns the same upper censorship to every report. Moreover, Sender's expected utility is the same in every equilibrium.
	\end{theorem}

	First, we describe the intuition for the first part of the result: existence of an equilibrium with an upper censorship. In the case of costless attention and Sender-optimal equilibria, the argument for Theorem \ref{thm:existence} relies on the shape of the exogenous noise in the action given a posterior belief, from the point of view of Sender. In this case, single-peakedness of the outside option makes Sender's expected utility S-shaped in the posterior mean, so upper censorships are naturally optimal because Sender is risk-loving conditional on low means and risk-averse conditional on high means. Second-order dominance is related to the informativeness of the signal, because $F\in \mathcal F$ is a mean-preserving spread of $\hat F\in \mathcal F$ iff $I_F\ge I_{\hat F}$, and an upper censorship induces either full information conditional on the state being lower than the threshold state, or no information except that the state exceeds the threshold. Therefore, intuitively, upper censorships induce mean distributions that align with the interests of Sender. Rational inattention introduces a novel difficulty for Sender: he can no longer take Receiver's attention as given. Because effort is endogenous, Sender must actively incentivize it. In this case, the relevant information policy from the point of view of Sender includes the effort $e$, and is given  by $x\mapsto e I(x) + (1-e) I_{\overline F}(x)$. We claim that, in equilibrium, effort is affected by the informativeness of Sender's information policy in a way that aligns with the interests of Sender. Suppose that Sender increases the net informativeness of posterior mean $x$: $\Delta I (x)$. This change induces outside option $x$ to pay more attention, by Lemma \ref{lem:receiver}.  If a type with outside option $x$ increases her effort, then she gathers more information, because the policy $x\mapsto e I(x) + (1-e) I_{\overline F}(x)$ increases along the informativeness order in $e$. Therefore, Sender spreads out Receiver's posterior-mean distribution around $x$ by increasing the net informativeness, even in the case of endogenous attention. 

	In the Appendix, we establish that continuity of the outside option distribution ensures that Sender is indifferent among all best responses of Receiver. Intuitively, identifying an equilibrium amounts to solving a maximization problem faced by Sender. Multiplicity arises only if a positive mass of types break ties in ways that lead to multiple objective functions. Ties between the two actions occur with probability 0, as in models with costless attention \citep*{lipnowski2024perfect}. We describe why ties among effort levels are irrelevant, considering a given type that is indifferent among effort levels. There are two cases: (1) the net informativeness at this type's outside option, $\Delta I(c)$ is positive; (2) her net informativeness is 0, which implies that her equilibrium action is the same as without information. Consider case (1). Observe that net informativeness is constant only if it takes value 0, by construction; see Figure \ref{fig:ext}.  Hence, a small change in the outside option changes the marginal benefit of effort, via Lemma \ref{lem:value}, so this case has zero probability by Assumption \ref{ass:1}. Consider case (2): all information policies lead to the same action, which implies that every ``equilibrium'' objective function has the same value among these information policies. Therefore, there are not multiple equilibrium expected payoffs for Sender.

	Theorem \ref{thm:existence} implies that the extensive margin of an optimal persuasion mechanism can be studied via an upper censorship, and reduces Sender's optimization to the uni-dimensional problem of censoring the state space. Because a convenient property of upper censorships is that the informativeness of the underlying signal is fully determined by the censorship point, any two upper censorships can be ordered in the Blackwell ranking. Hence, these results allow for a natural comparative static: does the equilibrium information quality increase in the attention cost? We study this question for the case of a linear cost function, i.e., for all $e$, $k(e)=e$.

	We say that \emph{Sender knows the attention cost $\lambda^*$} if the distribution of the attention cost puts full mass at some value $\lambda^*\in[0,1]$, and \emph{Sender knows the outside option $c^*$} if the distribution of the outside option is an atom at some $c^*\in[0,1]$. We say that \emph{independence and strict single-peakedness} hold if Assumption \ref{ass:1} holds, and $\gcond(\cdot|\lambda) $ is increasing on $[0,p]$, decreasing on $[p,1]$, and does not depend on $\lambda$. An information policy $I$ is \emph{optimal} if there exists an equilibrium in which Sender's persuasion mechanism assigns $I$ to every report.

	The following result shows that Sender provides more information as Receiver's attention cost increases, for small attention costs. Recall that, under strict single-peakedness and costless attention, there exists a unique cutoff $\theta_{0}\in\Theta $ such that the $\theta_0$ upper censorship is optimal \citep*{kolotilin_censorship_2022}.
		
	\begin{proposition}\label{prop:cost}
		Let Sender know the attention cost $\lambda^*$, independence and strict single-peakedness hold, $k$ be linear, and $\theta_0>0$. For all sufficiently small $\lambda^*>0$, there exists a unique $\theta_{\lambda^*}$ such that the $\theta_{\lambda ^*}$ upper censorship is optimal; moreover, it holds that $\theta_{\lambda ^*}>  \theta_0$.
	\end{proposition}
	

	We describe the intuition in the symmetric-information benchmark, for $c>x_0$ and $\lambda>0$. Sender maximizes the probability that Receiver chooses action 1. Due to symmetric information, the action is deterministic, unless the posterior mean is precisely $c$. Due to costly effort, Sender solves the problem with the additional constraint that Receiver exerts effort 1 --- if Receiver finds it optimal to exert a positive effort, then she finds it optimal to exert effort 1, given that $k$ is linear.\footnote{Formally, focusing on Sender-preferred equilibria, we consider that Sender solves
		\begin{align*}
		\max _{I\in \mathcal I} 1- I'(c^-) \	\text{subject to:} \  \Delta I (c) \geq \lambda;
	\end{align*}
	which is analyzed in the Appendix. The feasible set is empty if $\lambda$ is sufficiently high, in which case any signal is optimal. The argument in the text shows that: if the constraint does not bind at the solution, then $\Delta I (c) =0$, which contradicts $\lambda>0$ and equilibrium effort choice (Lemma \ref{lem:value}). This result is proved in the Appendix, see Lemma \ref{app:lem:knownzeta}; the problem without the constraint appears in \citet{gentzkow_costly_2014}.  For intuition, it is sufficient to consider effort 1. Consider an equilibrium with interior effort given information policy $I$ with $I(c)< I_{F_0}(c)$. There exists an information policy $J$ that ``increases'' $I(c)$ by a small amount and changes $I'(c)$ by a small amount; such $J$ can be constructed to be an upper censorship using Lemma \ref{app:lem:knownzeta}. If, instead, $I(c)\ge I_{F_0}(c)$, then there exists an equilibrium in which Sender chooses the fully informative $I_{F_0}$.} This additional constraint is not present when effort is costless, and we refer to it as the ``participation constraint''. We claim that the participation constraint binds. If the constraint does not bind, then Sender increases the probability of a posterior mean $x$ with $x\ge c$ as much as possible, subject to Bayes' rule. Specifically, he induces the mean $x=c$. Hence, given Sender's optimal signal, Receiver faces two contingencies: either she is indifferent between the actions or she finds it optimal to go for the riskless action. Therefore, information brings no value, and Receiver pays the positive cost $\lambda$ for a useless signal, a contradiction. So, the constraint binds. As an implication, Sender provides more information as $\lambda$ increases.

	Proposition \ref{prop:cost} shows that the insight generalizes to private information about the outside option, for a sufficiently small attention cost that is known to Sender. A change in the censorship state affects the extensive margin because of the outside-option distribution. However, because a nontrivial upper censorship is uniquely optimal with costless effort, only the upper bound of the extensive margin---denoted by $\overline c$ in Figure \ref{fig:ext}---is affected by small changes in the censorship state around $\theta_0$, not the lower bound $\underline c$. Specifically, for every $\overline \theta$, the net informativeness of $I_{\overline \theta}$ is 0 at any outside option greater than the conditional expectation of $\theta$ given $\theta \ge \overline \theta$. Hence, by increasing the censorship state, Sender is offsetting the decrease in the upper bound of the extensive margin that occurs as the attention cost increases. This argument leads to Proposition \ref{prop:cost}. Qualitatively, this comparative static holds also in \citet[Proposition 7]{wei_persuasion_2021}.

	The following result shows that Sender provides more information as the distribution of the attention cost increases in the reverse hazard rate order. In particular, in the class of distributions ($G$) with a well-defined and nonincreasing reverse hazard rate ($\lambda \mapsto G '(\lambda)/G(\lambda)$), we say that distribution $G_1$ is higher than distribution $G_2$ in the RHR order if, for all $\lambda$, it holds that $G_1'(\lambda)/G_1(\lambda)\ge G_2'(\lambda)/G_2(\lambda)$. Dominance in the RHR order implies first-order dominance, so an increase along the RHR order is consistent with the idea that the attention cost stochastically increases.
	
	\begin{assumption}\label{ass:prop:reverse}
		Sender knows the outside option $c^*$, with $c^*>x_0$, $k$ is linear, and the distribution of the attention cost admits a positive density on $(0,1)$ and a nonincreasing reverse hazard rate.
	\end{assumption}
	
	\begin{proposition}\label{prop:reverse}
	Let Assumption \ref{ass:prop:reverse}. There exists a unique state $\overline \theta$ such that the $\overline \theta$ upper censorship is optimal. Moreover, the state $\overline \theta$ is weakly higher if the attention-cost distribution increases in the RHR order.
	\end{proposition}

	Proposition \ref{prop:cost} and \ref{prop:reverse} show that Sender is more willing to inform Receiver as the attention cost increases, in each of the two cases in which one dimension of private information is absent. Consequently, if Receiver could actively conceal her true attention cost or mislead Sender into believing she faces a high cognitive burden, she would have an incentive to do so.

	In applications to media capture, Sender cares directly about attention. In this case, Sender is a dictator and owns the media of a state, so he collects advertising revenues. The utility function $U_G$ captures this application, defined by $U_G(\theta, a, e, c, \lambda) = a+\gamma e$ for $\gamma \ge0$.	This model is introduced by \citet{gehlbach_government_2014}, who assume a binary state and a binary signal. The case of $\gamma=0$ has been studied by \citet*{kolotilin_censorship_2022}, who show that upper censorships are optimal with costless attention.

	The following result shows that a generalization of the class of upper censorships contains an optimal information policy under certain conditions.\footnote{Equilibrium existence is not established for this extension. The difficulty lies in establishing continuity of the extensive margin --- i.e., continuity of $F\mapsto \underline c _\lambda(\Delta I_F)$ and $F\mapsto \overline  c_\lambda (\Delta I_F)$, defined in the proof of Lemma \ref{lem:interval} --- equipping $\mathcal F$ with the $L^1$-norm topology.} These information policies are defined by two threshold states, as in Figure \ref{fig:cens}. Define the induced posterior means of the pooled intervals corresponding to thresholds $\theta_1,\, \theta_2 \in \Theta$, with $\theta_1 \le \theta_2$, by
		\begin{align*}
			x_1 &\coloneqq \begin{cases}
				\frac{1}{F_0(\theta_2) - F_0(\theta_1)} \int_{\theta_1}^{\theta_2} \theta \diff F_0(\theta), & \text{if } \theta_1 < \theta_2, \\
				\theta_1, & \text{if } \theta_1 = \theta_2,
			\end{cases} \\
			x_2 &\coloneqq \begin{cases}
				\frac{1}{1 - F_0(\theta_2)} \int_{\theta_2}^1 \theta \diff F_0(\theta), & \text{if } \theta_2 < 1, \\
				1, & \text{if } \theta_2 = 1.
			\end{cases}
		\end{align*}
 	The \emph{bi-upper censorship} corresponding to $\theta_1$ and $\theta_2$ is the unique information policy $I$ given by
		\begin{align*}
			I(x) =
			\begin{cases}
				I_{F_0}(x), & \text{if } x \in [0, \theta_1], \\
				I_{F_0}(\theta_1) + F_0(\theta_1)(x - \theta_1), & \text{if } x \in (\theta_1, x_1], \\
				I_{F_0}(\theta_2) + F_0(\theta_2)(x - \theta_2), & \text{if } x \in (x_1, x_2], \\
				I_{\overline{F}}(x), & \text{if } x \in (x_2, \infty),
			\end{cases}
		\end{align*}
		 and a {bi-upper censorship} is an information policy that satisfies the above equality for thresholds  $\theta_1,\, \theta_2 \in \Theta$ with $\theta_1\le\theta_2$.

	\begin{proposition}\label{prop:unique}
		Let Sender know the attention cost $\lambda^*$, Assumption \ref{ass:1} hold with $p\ge x_0$, $k$ be linear, and the utility of Sender be given by $U_G$. For every equilibrium $\langle I^*_\bullet, \rho, \eta, \alpha \rangle$, there exists a bi-upper censorship such that the mechanism assigning this bi-upper censorship to every report yields a weakly higher Sender's expected utility than the equilibrium mechanism.
	\end{proposition}
	
	The proof constructs a bi-upper censorship that improves upon an arbitrary information policy, both in terms of higher expected action and attention. The requirement that the peak satisfies $p\ge x_0$ represents sufficient ex-ante disagreement between Sender and Receiver, as in \citet{shishkin_evidence_2023}.

	The proof clarifies the role of the highest threshold state in a bi-upper censorship. Specifically, the highest threshold state increases the marginal benefit of effort of certain types in case an upper censorship induces fewer outside-option types to exert effort than the optimal extensive margin.  First, we construct an upper censorship $I_{\overline \theta}$ that improves upon a given $I$ for $\gamma=0$, thanks to the same intuition as for Theorem \ref{thm:existence}. Second, we take into account the preferences of Sender over the extensive margin: we modify $I_{\overline \theta}$ in a way that replicates the extensive margin of $I$ by censoring extreme states on either side of the state space. Separating ``very high'' states from ``moderately high'' states allows Sender to increase the upper bound of the extensive margin of $I_{\overline \theta}$, if needed; additionally censoring low states increases the lower bound of the extensive margin of $I_{\overline \theta}$. Hence, given the choice of these additional two censoring regions, we obtain a candidate ``improving'' policy that is not a bi-upper censorship, because it involves three censoring regions. Lastly, we leverage single-peakedness to note that decreasing the lower bound of the extensive margin is beneficial for Sender, as described following Proposition \ref{prop:cost}. The lower bound is minimized by choosing to fully reveal low states, so the argument returns a bi-upper censorship that improves upon an arbitrary information policy.

	\section{Discussion and interpretation} \label{sec:model:discussion}
	
	In this section, we interpret $e$ as representing attention effort, and consider the effort-choice stage of the game for a nondecreasing $k$. An increase in attention effort results in a more informed Receiver and higher costs. The functional form of effort cost is unrestricted in the model to capture a range of attention- and non-attention-related phenomena. Examples of costly attention include cognitive difficulties and memory limits. In contrast, the opportunity cost of being attentive is relevant when evaluating media subscription or exposure.
	
%

	In \citet*{lipnowski_attention_2020} (LMW), the attention cost is proportional to the reduction in the uncertainty about the state. Receiver incurs a cost for what she learns about the state. LMW is a model of \emph{delegated} learning \citep{bloedel_persuading_2021}, and fits applications in which a separate agent conducts research on $\theta$, and Receiver learns through that research. As an illustration, LMW captures the problem of a firm (Receiver) that processes data provided by a research agency (Sender). \citet{wei_persuasion_2021} applies this paradigm to study Sender's state-independent  preferences.

	In the main model of \citet{bloedel_persuading_2021} (BS), the attention cost is proportional to the entropy reduction of the belief about Sender's message.\footnote{We describe the main model of BS with Sender's state-independent preferences and entropy-based cost, even though the paper includes other preferences and costs. \citeauthor{bloedel_persuading_2021} also consider our symmetric-information benchmark, as an alternative to their model. Under symmetric information, Sender chooses a binary signal without loss of optimality by standard arguments; see ``Lemma 3'' in that article.} Receiver incurs a cost for what she learns about Sender's talk. BS is a model of learning \emph{from communication}, fitting applications in which communication is costly to process. As an illustration, BS captures the problem of a social-media user (Receiver) who learns from the advertisement of an influencer (Sender) at a cost, which, importantly, is related to deciphering words and situations portrayed in the advertisement. The Sender's optimal signal is an upper censorship in BS, although for a different reason than in this model. In particular, Sender perceives Receiver's action as random given a realized message in BS, because of the attention strategy, whereas randomness arises due to both effort and asymmetric information in this model (as discussed in Section \ref{sec:application}).

	In this model, the attention cost is independent of the information provided by Sender, and fully flexible in this class. Receiver incurs a cost that is related to her exposure to information. We model learning \emph{via exposure}, fitting applications in which Receiver's strategy has a cost irrespective of the information provided by Sender. Paying full attention to a communication that turns out to be uninformative is allowed to have any cost here, whereas this strategy is costless in LMW. Additionally, paying full attention to the uninformative communication of a single possible message, constant in the state, is allowed to have any cost in this model, whereas this strategy is costless in BS. As an illustration, this model captures the problem of a platform user (Receiver) who devotes a share of her mental energy to learning about current affairs from her news feed, which is engineered by a platform (Sender). If the feed contains only friends' updates and product advertisements, then searching for news is both costly and fruitless. This separation between costs and information is well suited for applications in which cognitive costs are thought of less granularly than in BS-LMW. As cognitive costs are aggregated with costs of a different nature, such as platform fees, or outside opportunities that accrue with the time spent learning, the total costs become less dependent on Sender's communication strategy.

	Our model builds upon BS and LMW by constraining the attention strategies to mixtures of full and null information about Sender's message. In the rational-inattention tradition \citep{sims_implications_2003}, Receiver in BS-LMW flexibly allocates her cognitive resources, because she can learn in any conceivable way. However, in these papers, allowing such flexibility comes with tracking multiple signal structures, and using extensions of entropy-based costs. By insisting on a single attention variable, effort, our framework preserves the essential tradeoff of rational inattention. We also analyze questions related to screening and the shape of the extensive margin, complementing the current literature. Moreover, departures from flexibility reflect real-life psychological and technological constraints. For instance, a consumer may only choose the time and mental energy to spend in front of the TV, and a voter may only choose how many articles to sample randomly, and from which to learn fully, in a newspaper. Lastly, rational inattention is not the only explanation for costly effort, which could refer to the opportunity cost of learning, or a transfer paid to ``infomediaries'' in applications --- including Sender, as in the hypothesis of Proposition \ref{prop:unique}.

	In \cite{matyskova_bayesian_2023}, Receiver pays a cost to access additional information beyond what Sender provides. In \citet{dworczak_preparing_2022}, Receiver may have access to information sources other Sender's one. These models target a fundamentally different strategic context than ours and fit applications in which Sender's communication is costless to understand. Sender's tradeoff involves (i) inducing favorable actions and (ii) preventing access to external information that may hinder (i). In the model with binary state and action of \citeauthor{matyskova_bayesian_2023}, Sender provides more information as the  cost of extra information increases, similar to Proposition \ref{prop:cost}. However, the channel is different: in  \cite{matyskova_bayesian_2023}, Sender provides more information so as to disincentivize (extra) attention, whereas here Sender does so to incentivize attention.

	\citet{bilnancini_signaling_2018} study Receiver's attention cost in a model with a different communication protocol: Sender chooses a payoff-relevant message after learning about the state. Their model differs from a persuasion game \`a la \citet{kamenica_bayesian_2011}, in which Sender chooses the distribution of a payoff-irrelevant message without incentive-compatibility constraints. 
	The authors identify a strategic complementarity: if the different types of Sender pool on the same message, then Receiver does not acquire information about this message, even if she pays attention in a separating equilibrium. According to this logic, the lower the information in the strategy of Sender, the lower the incentives of Receiver to acquire information. This strategic complementarity of the signaling game reveals certain properties of the pooling equilibrium in which the message that is sent is the least costly for Sender in all states, and resonates with the supermodularity of the objective function in Receiver's choice of effort (Corollary \ref{cor:superm}). In fact, Receiver chooses her attention taking an experiment about the state as given in both games. However, the experiment is observed by Receiver in this model, whereas Sender takes the attention of Receiver as given in \citet{bilnancini_signaling_2018}. Therefore, the models capture different economic environments. The ``covert communication'' of \citeauthor{bilnancini_signaling_2018} captures a communication game between two parties, in which communication and attention are chosen simultaneously, as in social interactions. The ``overt communication'' of this model captures settings in which the communication strategy is known prior to the attention decisions, such as when a user searches for information on a social media platform. When the user decides to spend time and cognitive resources on the associated app, she naturally has a belief about the usefulness of the available information, based on past use of the app and her network.

	\section{Conclusion}
	
	This paper proposes a model of inattention within a persuasion game that underscores the complementarity between information and attention effort. This complementarity leads to the equivalence of persuasion mechanisms and experiments. The optimization problem of the sender is solved by censoring favorable states. The same logic applies in contexts in which the sender directly values attention, such as in applications to media capture.
	
	In general, complementarity may not hold for all audiences and informational environments, for instance due to information overload and psychological constraints. A study of the extensive margin of persuasion that incorporates these distinctions offers an open avenue for future research.


%

	\newpage 
	\appendix
			
		\label{app:appendix}

		\section{Preliminaries, equilibrium, and revelation principle} \label{app:sec:eq}
		
		In these appendices, we let $P\in\bm\Delta (T)$ denote the distribution of Receiver's type.
		\subsection{Preliminaries} \label{app:sec:eq:withoutloss}
		
		\paragraph{Mean measurability}
		We show that Sender's signal affects the decisions and payoffs of both Sender and Receiver only through the distribution of the posterior mean of a Bayesian agent who always observes the signal realization.
		
		The optimal action of type $t=(c, \lambda)$, given posterior belief $\mu\in \mathcal D$, depends on the belief only through its mean $\overline x_\mu:=\int _{[0,1]} \theta\diff \mu (\theta)$. Receiver's expected material payoff given belief $\mu$ is given by
		\begin{align*}
			v_t(\mu) &:= \begin{cases} \int_{[0,1]} (\theta-c) \diff \mu (\theta), & \ \text{if} \ \overline x_\mu \geq c, \\
				0, & \ \text{if} \  \overline x_\mu < c.
			\end{cases}
		\end{align*}
		We note that $v_t(\mu) $ depends on the belief $\mu$ only through $\overline x_\mu$. If Sender's signal induces the Bayes-plausible  \citep{kamenica_bayesian_2011} distribution over posterior beliefs $p$, then type $t$ chooses $e\in[0,1]$ to maximize $e\mapsto 
			e\int_{\mathcal D} v_t(\mu) \diff p(\mu) + (1-e) v_t(F_0) -\lambda k(e)$.
		Hence, Receiver's action, effort, and her payoff depend on Sender's signal only via the distribution of the posterior mean; i.e., the distribution of $x_\mu$ implied by $p$. This conclusion follows from Sender's payoff function, which depends on the signal only via Receiver's choice of action. The same conclusion holds under the hypothesis of Proposition \ref{prop:unique}.

		\paragraph{Equivalence between signals and information policies}		
		\begin{lemma}\label{lem:policies}
			The following hold:
			\begin{enumerate}
				\item If $F\in \mathcal F$, then $I_F\in \mathcal I$;
				\item If $I\in\mathcal I$, then $I'\in \mathcal F$, extending $I$ to take value $0$ at every $x<0$.
			\end{enumerate}
		\end{lemma}
		\begin{proof}
			See \citet{gentzkow_rothschild-stiglitz_2016} and \citet{kolotilin_optimal_2018}.
		\end{proof}
	
		In the rest of the appendices, we identify each element of $\mathcal I$ with its restriction on $[0,1]$.
		\paragraph{Properties of information policies} We collect certain properties of information policies. Specifically, by Lemma \ref{app:lem:level} we can use results of \citet{augias_economics_2025}, and Proposition \ref{app:prop:isoelastic} uses Lemma \ref{app:lem:V}, Lemma \ref{app:lem:C}, and Proposition \ref{app:prop:uppercensorships}.

		We endow the set $\mathcal C(X)$ of real-valued and continuous functions on $X\coloneqq [0,1]$ with the supremum norm $|\cdot |_\infty$. We define the set $\mathcal K \coloneqq \{u\in \mathcal C(X) \mid  u \ \text{is convex}\}$.
		We endow $\mathcal K$ with the pointwise order $\le $ so that: $v \le u$ if, for all $x$, $u(x) \ge v(x)$. We adopt the convention of saying, for all $u,\,v\in \mathcal K$, for all $Y\subseteq X$, $u\rvert_Y \ge v\rvert _Y$ if, for all $x\in Y$, we have $u(x) \ge v(x)$. Moreover, for all $u\in\mathcal K$, we define $\partial u(X)=\bigcup_{x\in \Int X} \partial u(x)$, and consider, for all $x<1$, $u'(x)$ as the right-derivative $\partial _+u(x) $ of $u$ at $x$, and, for all $x>0$, $u'(x^-)$ as the left-derivative $\partial _-u(x) $ of $u$ at $x$. We denote by $I_0$ the information policy of $F_0$, that is  $x\mapsto \int _0^xF_0(t)\diff t$, and  by $\underline I$ the information policy of $\overline F$, that is $x\mapsto (x-x_0)_+$.

		\begin{lemma}[No atoms at 1]\label{app:lem:noatomsatone}
			If $I\in\mathcal K$, and $I_0\ge I\ge \underline I$, then $I'(1^-)\ge 1$.
		\end{lemma}
		\begin{proof}
			Fix $I\in\mathcal I$ and $\delta>0$ such that: for all $\varepsilon>0$, $I'(1-\varepsilon ^-)\le 1-\delta $. Fix $\delta'>0$ such that $\delta'<\delta$. By continuity of $F_0$, there exists $\varepsilon'>0$ such that: for all $\varepsilon<\varepsilon'$, we have $F_0(1-\varepsilon)>1-\delta'$. Fix $\varepsilon<\varepsilon'$ with $\varepsilon>0$. By convexity of $I_0$, we have
			\begin{align*}
				I_0(1)\ge I_0(1-\varepsilon)+F_0(1-\varepsilon)\varepsilon.
			\end{align*}
			Hence, we have, using $I\le I_0$, $\delta>\delta'$, $0<\varepsilon<\varepsilon'$,
			\begin{align*}
				I_0(1)> I(1-\varepsilon)+(1-\delta)\varepsilon.
			\end{align*}
			However, by convexity of $I$ and $I_0\ge I\ge\underline I$, we have
			\begin{align*}
				I(1-\varepsilon)\ge I_0(1)-\varepsilon (1-\delta).
			\end{align*}
		\end{proof}

		\begin{lemma}[No atoms at 0]\label{app:lem:noatomsatzero}
			If $I\in\mathcal K$, and $I_0\ge I\ge \underline I$, then $I'(0)\le 0$.
		\end{lemma}
		\begin{proof}
			Fix $I\in\mathcal K$, with $I_0\ge I\ge \underline I$ and $\delta>0$ such that $I'(0)=\delta$. Then, for all $\varepsilon>0$, we have $I(\varepsilon)-I(0)\ge \delta\varepsilon$. 
			
			By continuity of $F_0$, there exists $\varepsilon'>0$ such that, for all $\varepsilon<\varepsilon'$, we have $F_0(\varepsilon)<\delta$. Hence,  for all $\varepsilon<\varepsilon'$, with $\varepsilon>0$, we have $I_0(\varepsilon)-I_0(0)<\delta\varepsilon$. Therefore, for every $\varepsilon<\varepsilon'$ with $\varepsilon>0$, we have that $\delta\varepsilon>I_0(\varepsilon)$ and that $I_0(\varepsilon)\ge \delta\varepsilon$.
		\end{proof}
		
		We study the subset $ \mathcal U\coloneqq \{ I\in\mathcal K \mid \underline I\le I\le I_0 , \ \partial I(X)\subseteq [0,1]\}$ of $\mathcal K$ defined by a level and slope constraint. By construction, it holds that $\mathcal U\subseteq \mathcal I$. By the following result,  $\mathcal U = \mathcal  I$.

		\begin{lemma}[Level-constraint characterization]\label{app:lem:level}
			It holds that
			\begin{align*}
				\{I \in\mathcal K \mid  I_0\rvert_{\Int X}\ge I\rvert_{\Int X}\ge \underline I\rvert_{\Int X}\}\subseteq \mathcal U.
			\end{align*}
		\end{lemma}
		
		\begin{proof}
			
			Suppose $I\colon X\to[0,1] $ is convex and such that $ I_0\rvert_{\Int X}\ge I\rvert_{\Int X}\ge \underline I\rvert_{\Int X}$.

			\emph{Part 1: Level.}
			By left continuity of $I_0$ and $\underline I$ at 1, right continuity of $I_0$ and $\underline I$ at 0, $I_0(1)=\underline I(1)$, and $I_0(0)=\underline I(0)$, we have $I_0\ge I\ge\underline I$.  			The distribution on $[0,1]$ represented by the CDF $F_0$ has full support, so $1>F_0(x)>0$ for all $x\in\Int X$. As a consequence, $I_0\rvert _{\Int X}>\underline I\rvert _{\Int X}$.
			
			\emph{Part 2: Continuity.}
			Suppose $I\colon X\to[0,1] $ is convex and such that $ I_0\rvert_{\Int X}\ge I\rvert_{\Int X}\ge \underline I\rvert_{\Int X}$. We want to show that $I$ is right continuous at 0. Towards a contradiction, suppose that there exists $\delta>0$ such that, for all $\varepsilon>0$, we have $I(\varepsilon)>\delta$. By right continuity of $I_0$, we have that: there exists $\varepsilon'>0$ such that, for all $\varepsilon<\varepsilon'$ with $\varepsilon>0$, we have $I_0(\varepsilon)<\delta$. Hence, we have, for all $\varepsilon<\varepsilon'$ with $\varepsilon>0$, we have $I_0(\varepsilon)<\delta$ and $I(\varepsilon)>\delta$, which is a contradiction.  We want to show that $I$ is left continuous at 1. Suppose that there exists $\delta>0,\,\varepsilon'>0$ such that, for all $\varepsilon<\varepsilon'$  with $\varepsilon>0$, we have $I(1-\varepsilon)<1-x_0-\delta$. By left continuity of $\underline I$, there exists $\varepsilon''$ such that, for all $\varepsilon<\varepsilon''$, we have $\underline I(1-\varepsilon)>1-x_0-\delta$. Hence, we have, for all $\varepsilon<\min\{\varepsilon', \varepsilon''\}$, $I(1-\varepsilon)<\underline I(1-\varepsilon)$, which is a contradiction.

			\emph{Part 3: Subdifferential.} Let $x<1$ be such that $I'(x)>1$ and $x_0\le x$. Then $I(1)-I(x)>1-x$, but $I(1)-\underline I(x)=1-x$, so: for all $x\in\Int X $, $\partial_+ I(x)\le 1$. Let $x>0$ be such that $x<x_0$ and $I'(x^-)<0$. Then, $I(x)-I(0)<0$, and, by hypothesis, we have $ I(x)\ge \underline I(x)=0$, so: for all $x\in\Int X $, $\partial_- I(x)\ge 0$. Hence, using Lemmata \ref{app:lem:noatomsatone} and \ref{app:lem:noatomsatzero}, we conclude that $\partial I(X) \subseteq  [0,1]$.
		\end{proof}

				 We note that $I\in \mathcal I$ is absolutely continuous by \citet[Theorem 14.13]{yeh_real_2006}.
		
		\begin{lemma}[Tangent characterization]\label{app:lem:tangent_integration}
			For all $I\in \mathcal I$, $c,\,x\in X$, we have
			\begin{align*}
				I(x)+I'(x)(c-x) = \int _0^xc-t\diff I'(t).
			\end{align*}
		\end{lemma}
		
		\begin{proof}
			Fix $c,\,x\in X$. The result follows from Riemann-Stieltjes integration by parts, given that $I$ is absolutely continuous on $[0, x]$.
		\end{proof}
		
		We refer to the $\theta$ upper censorship as $I^\theta$. For $c\in\Int X$, we define $\underline \theta(c)\in X$ as the largest $\theta$ such that $I^\theta(c)=\underline I(c)$ as follows. By Lemma \ref{app:lem:tangent_integration}, the function $x\mapsto I_0(x)+F_0(x)(c-x)$ is increasing on $[0,c]$. This function is continuous, takes value 0 at 0, and value $I_0(c)$ at $c$. By the intermediate value theorem, there exists a unique $x\in[0,c) $ such that
		\begin{align*}
			I_0(x)+F_0(x)(c-x) =(c-x_0)_+,
		\end{align*}
		which we denote by $\underline \theta(c)$.

		\begin{lemma}[Image of upper censorships]\label{app:lem:image}
			For all $c\in \Int X$,  it holds that:
			\begin{enumerate}
				\item For all $\theta\ge c$, $I^\theta(c)=I_0(c)$;
				\item There exists $\underline \theta\in [0, c)$ such that: for all $\theta\le  \underline \theta$, $I^\theta(c)= \underline I(c)$, and $\theta\mapsto I^\theta(c)$ is increasing on $(\underline \theta, c)$, and $\underline \theta=0$ iff $c\le x_0$;
				\item For all $v\in((c-x_0)_+, I_0(c))$, there exists a unique $\theta\in (\underline \theta,c)$ such that: 
				\begin{align*}
					I_0(\theta)+F_0(\theta)(c-\theta)= v.
				\end{align*}
			\end{enumerate}
		\end{lemma}
		\begin{proof}
			Part 1 holds by definition of upper censorships.
			
			Fix $c\in \Int X,\,v\in((c-x_0)_+, I_0(c))$. First, we obtain certain properties of the tangent of $I_0$ at $x$, evaluated at $c$, as a function of $x$, $x\mapsto	I_0(x)+F_0(x)(c-x)$. Fix $x, \,y\in X$. We have, by Lemma \ref{app:lem:tangent_integration},
			\begin{align*}
				I_0(x)+F_0(x)(c-x) - I_0(y)-F_0(y)(c-y) &= \int_y^x c-t\diff F_0(t).
			\end{align*}
			Hence, if $c\ge x>y$, then $I_0(x)+F_0(x)(c-x) - I_0(y)-F_0(y)(c-y) >0$, so the function $x\mapsto I_0(x)+F_0(x)(c-x)$ is increasing on $[0,c]$. This function is continuous, and takes value 0 at 0 and $I_0(c)$ at $c$. By the intermediate value theorem, there exists a unique $\underline x\in[0, c)$ such that:
			\begin{align*}
				I_0(x)+F_0(x)(c-x) =(c-x_0)_+,
			\end{align*}
			Hence part (2) holds.

			Moreover, there exists a unique $\overline x\in(\underline x, c)$ such that 
			\begin{align*}
				I_0(x)+F_0(x)(c-x) =v.
			\end{align*}
			Therefore, there exists a unique $\theta\in (0, c)$ such that: 
			\begin{align*}
				I_0(\theta)+F_0(\theta)(c-\theta) = v,
			\end{align*}
			which, in fact, equals $\overline x$. Hence, part (3) holds.
		\end{proof}

		\begin{lemma}\label{app:lem:V}
			Fix $c\in\Int X$. The function $V\colon [\underline \theta(c) ,c]\to [(c-x_0)_+, I_0(c)] $  given by $\theta\mapsto	I_0(\theta)+F_0(\theta)(c-\theta)$ is increasing, surjective, and continuous. Moreover, if $F_0$ is continuously differentiable, then $V$ is continuously differentiable.
		\end{lemma}
		\begin{proof}
			$V$ is increasing by Lemma \ref{app:lem:image}, is continuous by continuity of $F_0$, and surjective by Lemma \ref{app:lem:image}.
		\end{proof}

		\begin{lemma}\label{app:lem:C}
			Fix $c\in\Int X$. The function $C\colon [(c-x_0)_+, I_0(c)]\to [\underline \theta (c),c]$ given by, for every $v\in [(c-x_0)_+, I_0(c)]$, the unique value of $\theta\in[0,c]$ such that
			\begin{align*}
				I_0(\theta)+F_0(\theta)(c-\theta)=v,
			\end{align*}
			is increasing, surjective, continuous, and the inverse function of $V$. 
		\end{lemma}
		\begin{proof}
			By construction, $C$ is the inverse function of $V$. Hence, by Lemma \ref{app:lem:V}, $C$ is increasing, continuous, and surjective on its stated domain.
		\end{proof}

	\begin{lemma}\label{app:lem:concavity}
		Let $c\in(x_0, 1)$. The function $w\mapsto 1-F_0(C(w+c-x_0))$ is strictly concave on $[0, I_0(c)-\underline I(c)]$.
	\end{lemma}
	\begin{proof}
		Let $v_1, v_2 \in [(c-x_0)_+, V(c)]$ with $v_1 < v_2$. Using Lemmata \ref{app:lem:V} and \ref{app:lem:C} and a change of variables,
		\begin{align*}
			\int_{v_1}^{v_2} \frac{1}{c-C(v)} \diff v &= \int_{C(v_1)}^{C(v_2)} \frac{1}{c-\theta} \diff V(\theta)\\
			&= \int_{C(v_1)}^{C(v_2)} \diff F_0(\theta) \\
			&= F_0(C(v_2)) - F_0(C(v_1)).
		\end{align*}
		Therefore, for any $w \in [0, I_0(c)-\underline I(c)]$, 
		\begin{align*}
			1 - F_0(C(w+c-x_0)) &= 1 - F_0(C(c-x_0)) - \int_{0}^{w} \frac{1}{c-C(v+c-x_0)} \diff v.
		\end{align*}
		Because $C$ is increasing, the integrand $v \mapsto \frac{1}{c-C(v+c-x_0)}$ is increasing. 
	\end{proof}

		\begin{lemma}\label{app:lem:boundrightderivative}
			For all $c\in\Int X$, $v\in[(c-x_0)_+, I_0(c)]$, $I\in\mathcal I$, we have that: if $I(c)=v$, then
			\begin{align*}
				I'(c^-)\ge F_0(C(v)).
			\end{align*}
		\end{lemma}
		\begin{proof}
			Fix $c\in\Int X$, $v\in[(c-x_0)_+, I_0(c)]$, $I\in\mathcal I$ with $I(c)=v$ and the unique $\theta\in[0,c]$ such that $
				I^\theta(c) = v$.
			If $v<I_0(c)$, then $\theta<c$, so  $I'(c^-)<F_0(\theta)$ implies $I(c)-I(\theta)< (c-\theta)F_0(\theta)$. Hence, it is not possible that $I^\theta(c)=I(c)$.

			Therefore, we  have that: for all $c\in\Int X$, $v\in[(c-x_0)_+, I_0(c))$, $I\in\mathcal I$ with $I(c)=v$, it holds that $	I'(c^-)\ge F_0(C(v))$. 
			
			To conclude the proof, fix $c\in\Int X$, $I\in\mathcal I$ with $I(c)=I_0(c)$ and $I'(c^-)<F_0(c)$. By continuity of $F_0$, we have that: there exists $x<c$ such that: for all $y\in(x, c)$, we have $F_0(y)>I'(c)$. Hence, we have $I_0(c)-I_0(x)>I(c)-I(x)$, which contradicts $I_0(c)=I(c)$.
		\end{proof}

		By the following result, we have that $1-F_0(C(v+c-x_0))$ is the largest probability of a mean in excess of $c$ given the constraint that $I(c)-c+x_0\ge v$. 
		\begin{proposition} \label{app:prop:uppercensorships}
			Let $c \in \Int X$. For all $v\in[(c-x_0)_+, I_0(c)]$, there exists a unique $\theta\in[\underline \theta(c) , c]$ such that, for all $I\in\mathcal I$, if $I(c)=v$, then
			\begin{align*}
				I^\theta(x) - I(x) \geq 0 &, \ \text{for all} \  x\in [0,c],\\
				I^\theta (x) - I(x) \leq 0 &, \ \text{for all} \  x\in [c,\infty),\\
				I'(c^-) \ge  F_0(C(v)).&
			\end{align*}
			Moreover, for all $v\in((c-x_0)_+, I_0(c)]$, there exists a unique $\theta\in(0, c]$ such that, for all $I\in\mathcal I$, if $I(c)=v$, then the preceding condition holds.
		\end{proposition}
		\begin{proof}
			Let $c \in \Int X$ and $v\in[(c-x_0)_+, I_0(c)]$. By Lemma \ref{app:lem:knownzeta}, there exists $\theta \in[0,c]$ such that: $I^\theta (c) = v$ and, for all $I\in\mathcal I$ with $I(c)=v$, we have
			\begin{align*}
				I^\theta(x) - I(x) \geq 0 &, \ \text{for all} \  x\in [0,c],\\
				I^\theta (x) - I(x) \leq 0 &, \ \text{for all} \  x\in [c,\infty).
			\end{align*}
			However, by Lemma \ref{app:lem:boundrightderivative}, there exists a unique $\theta\in[\underline \theta(c), c]$ such that $I^\theta (c) = I(c)$, and is given by $C(v)$. 
		\end{proof}

	\subsection{Equilibrium definition} \label{app:sec:eq:eq}
	We define a Perfect Bayesian Equilibrium in which Sender directly chooses a list of information policies. This approach is without loss due to the observations in Section \ref{app:sec:eq:withoutloss}. An equilibrium is a tuple $\langle I^*_\bullet, \rho, \eta, \alpha \rangle$, in which $I^*_\bullet \in \mathcal I^{|T|}$, $\rho(\cdot, \hat I_\bullet) \colon T\to T$ is measurable for all $\hat I_\bullet \in \mathcal I^{|T|}$, $\eta(\cdot, \hat I_\bullet) \colon T\to [0,1]$ is measurable for all $\hat I_\bullet \in \mathcal I^{|T|}$, and $\alpha(c, \lambda, \cdot) \colon [0,1]\to [0,1]$ is measurable for all $(c, \lambda)\in T$, such that:
	\begin{enumerate}
		\item {Action optimality:} for all $x\in[0,1]$ and $(c, \lambda) \in T$,
		\begin{align*}
			\alpha(c, \lambda, x) \in  \argmax _{b \in [0,1]} b(x-c);
		\end{align*}
		
		\item {Reporting and effort optimality:} for all $(c, \lambda) \in T$ and $\hat I_\bullet \in \mathcal I^{|T|}$,
		\begin{align*}
			(\rho(c, \lambda, \hat I_\bullet), \eta(c, \lambda, \hat I_\bullet)) \in \argmax_{(r, \tilde e) \in T \times [0,1]} \int _{[0,1]}\max_{a\in\{0,1\}}U_R(x,a, \tilde e, c, \lambda) \diff \left( \tilde e \hat I'_r(x) + (1-\tilde e) I'_{\overline F}(x) \right);
		\end{align*}
		
		\item Sender rationality given $\alpha, \, \rho$, and $\eta$: $I^*_\bullet$ maximizes the expected utility $\hat W(\cdot, \rho, \eta, \alpha)$ on $\mathcal I^{|T|}$, in which $\hat W(\hat I_\bullet, \rho, \eta, \alpha)$ is defined by:
		\begin{align*}
			\int_T \int_{[0,1]} \alpha(c, \lambda, x) \diff \left( \eta(c, \lambda, \hat I_\bullet) \hat I'_{\rho(c, \lambda, \hat I_\bullet)}(x) + (1 - \eta(c, \lambda, \hat I_\bullet)) I'_{\overline F}(x) \right) \diff P((c, \lambda)).
		\end{align*}
	\end{enumerate}

		\subsection{Coordination mechanisms}\label{app:sec:rev}
		
		In this section, we study more general mechanisms than in the main text. In Section \ref{app:sec:rev:rev}, we introduce the model, and establish that persuasion mechanisms are without loss of generality under the restriction to deterministic mechanisms in the framework of \citet{myerson_optimal_1982} (Proposition \ref{app:prop:revelation_product}). In Section \ref{app:sec:disintegration}, we identify certain mechanisms that are without loss of generality (Proposition \ref{app:prop:disintegration_activation}), and we provide conditions on the parameters of the model under which persuasion mechanisms are without loss (Proposition \ref{app:prop:isoelastic} and \ref{app:prop:linear}). In Section \ref{app:sec:examples}, we show that persuasion mechanisms are with loss. In Section \ref{app:sec:rev:discussion}, we motivate the focus on persuasion mechanisms.

		\subsubsection{Revelation principle for coordination mechanisms}\label{app:sec:rev:rev}
		
		We consider a generalized principal-agent problem \`a la \citet{myerson_optimal_1982}, with a single agent, and in which: $T$ is the set of types, $E\coloneqq [0,1]$ is the set of decisions of the agent, $\mathcal I$ is the set of decisions of the principal, and $P\in\bm \Delta ( T)$ is the type distribution. We denote by $\mathscr{B} (U)$ the Borel sigma algebra on a set $U$, and by $\mathscr B(E\times \mathcal I)$ the product sigma algebra $\mathscr B(E) \times \mathscr B(\mathcal I)$. We note that $\mathcal I$ is a compact subset of a metric space.\footnote{This observation follows from Lemma \ref{app:lem:level} and Proposition 1 in \citet{augias_economics_2025}.} Given type $t=(c, \lambda )\in T$ and decisions $(e, I)\in E\times \mathcal I$, the designer's payoff is
		\begin{align*}
			u_0(I, e, t) \coloneqq 	e(1-I'(c^-)) + (1-e)[x_0\ge c],
		\end{align*}
		and the agent's payoff is
		\begin{align*}
			u_1(I, e, t)\coloneqq e\Delta I(c)-\lambda k(e)+I_{\overline  F}(c)+x_0-c.
		\end{align*}

		A \emph{coordination mechanism} is a tuple $(R, \mathscr R, M, \mathscr M, \pi)$ in which: $R$ is a set of reports, $\mathscr R$ is a sigma algebra on $R$,   $M$ is a set of messages, $\mathscr M$ is a sigma algebra on $M$, and $\pi(\cdot |r)$ is a distribution on $M\times \mathcal I$, that is $\pi$ is a function $\pi \colon R \times \mathscr B(M\times \mathcal I) \to [0,1]$ such that, for all measurable $N\times \mathcal J\subseteq M\times \mathcal I$, $y\mapsto \pi(N\times\mathcal J |y)$ is measurable; a typical mechanism is denoted by the triple $(R, M, \pi)$, omitting the notation for sigma algebras. Given a coordination mechanism $G=(R, M, \pi)$, a \emph{participation strategy for $G$} is a pair $(\rho, \delta)$ of measurable functions with $\rho \colon T\to R$ and $\delta \colon M\times T\to E$. We define the utility of $i\in\{0,1\}$ from the mechanism $G$ and the participation strategy $( \rho, \delta )$  for $G$ given type $t$ as
		\begin{align*}
			U_i(\rho, \delta, G, t) \coloneqq \int_{M\times \mathcal I} u_i(I, \delta (m, t), t)\diff \pi(m, I |\rho (t)).
		\end{align*}
		We define the utility of $i\in\{0,1\}$ from the mechanism $G$ and the participation strategy $( \rho, \delta )$  for $G$ as
		\begin{align*}
			V_i(\rho, \delta, G) \coloneqq \int_{T} U_i(\rho, \delta, G, t)\diff P(t).
		\end{align*}
		Given a coordination mechanism $G$, a participation strategy $(\rho, \delta)$ for $G$ is an \emph{equilibrium in $G$} if, for all $t \in T$, for all $r \in R$, and for all measurable $\eta \colon M \to E$, we have
		\begin{align*}
			U_1(\rho, \delta, G, t) \ge \int_{M\times \mathcal I} u_1(I, \eta(m), t)\diff \pi(m, I |r).
		\end{align*}

		\citet[Proposition 2]{myerson_optimal_1982} establishes a revelation principle that applies to this setup. A \emph{direct coordination mechanism} is a coordination mechanism $(R, M, \pi)$ with $R=T$, $\mathscr R =\mathscr B(T)$, $M=E$, and $\mathscr M=\mathscr B(E)$. Let $(\rho^*, \delta^*)$ be the pair of functions such that: for all $(t, e)\in T\times E$, $\rho^*(t)\coloneqq t$ and $\delta^*(e, t)\coloneqq e$. The pair of functions $(\rho^*, \delta^*)$ is a valid participation strategy for every direct coordination mechanism. A direct coordination mechanism $G$ is \emph{incentive compatible} if $(\rho^*, \delta^*)$ is an equilibrium in $G$.
		
		\begin{proposition}[\citealt{myerson_optimal_1982}]\label{app:prop:revelation_myerson}
			For every coordination mechanism $G$, for every equilibrium $(\rho, \delta)$ in $G$, there exists a direct and incentive compatible coordination mechanism $F$ such that
			\begin{align*}
				V_0(\rho, \delta, G)	= V_0(\rho^*, \delta^*, F).
			\end{align*}
		\end{proposition}
		\begin{proof}
			See Proposition 2 in \citet{myerson_optimal_1982}.
		\end{proof}

		The following result establishes that persuasion mechanisms are without loss of generality under the restriction to mechanisms that do not feature correlation between the message and the information policy conditional on a report. As an implication, persuasion mechanisms are without loss under deterministic mechanisms. Given a mechanism $(R, M, \pi)$ and a report $r$, we let $\pi_\mathcal I(\cdot |r)$ and $\pi_M(\cdot|r)$ denote the marginal probability distribution over $\mathcal I$ and over $M$, respectively, of the joint distribution $\pi(\cdot |r)$ on $M\times\mathcal I$. A \emph{product mechanism} is a coordination mechanism $(R, M, \pi)$ such that, for all $r\in R$, for all measurable $\mathcal  J\subseteq \mathcal I ,\, N\subseteq M $, we have $\pi (\mathcal J \times N|r ) = \pi_{\mathcal I} (\mathcal J|r )  \pi_{M} (N|r ) $.
		An IC persuasion mechanism $(I_r)$ induces the payoff $v$ if there exists a measurable selection $\eta \colon T\to E$ from $t\mapsto \argmax_{e \in [0,1]} u_1(I_t, e, t)$ such that:
		\begin{align*}
			v= \int_T u_0(I_t, \eta(t), t)\diff P(t).
		\end{align*}

		\begin{proposition}\label{app:prop:revelation_product}
			For every product mechanism $G$, for every equilibrium $(\rho, \delta)$ in $G$, there exists an IC persuasion mechanism $I_\bullet$ and a payoff $v\ge V_ 0(\rho, \delta, G)$ such that: $I_\bullet$ induces $v$.
		\end{proposition}
		The key step in the argument is the observation that any effort decision of an incentive-compatible product mechanism can be replicated by the equilibrium of an IC persuasion mechanism.
		\begin{proof}
		Let's fix a product mechanism $G$, and an equilibrium $(\rho, \delta)$ in $G$.

		\emph{Part 1: Construction of the persuasion mechanism.}
		For every report $r \in R$, define the expected information policy $\bar{I}_r \coloneqq \int_{\mathcal{I}} I \diff \pi_{\mathcal{I}}(I | r)$, and note that $\bar{I}_r \in \mathcal{I}$. We construct a persuasion mechanism, $I_\bullet$, by assigning the deterministic policy $I_t \coloneqq \bar{I}_{\rho(t)}$ to every report $t \in T$.
		
		\emph{Part 2: A property of $u_0$ and $u_1$.} We claim that, for all $i \in \{0,1\}$, $e \in E$, $t \in T$, and $r \in R$, it holds that $\int_{\mathcal{I}} u_i(I, e, t) \diff \pi_{\mathcal{I}}(I | r) = u_i(\bar{I}_r, e, t)$. 
		
		For $i=1$,  for all $e \in E$, $t \in T$, and $I$, $u_1(I, e, t)$ depends affinely on the evaluation $I(c)$. Hence, we have $\int_{\mathcal{I}} I(c) \diff \pi_{\mathcal{I}}(I | r) = \bar{I}_r(c)$, which establishes the claim for $i=1$. For $i=0$, $u_0(I, e, t)$ is affine in the left-derivative $I'(c^-)$. It suffices to show that $\int_{\mathcal{I}} I'(c^-) \diff \pi_{\mathcal{I}}(I | r) = \bar{I}_r'(c^-)$. For any $I \in \mathcal{I}$ and $\varepsilon \in (0, c)$, define the difference quotient $D_\varepsilon(I) \coloneqq \frac{I(c) - I(c - \varepsilon)}{\varepsilon}$. Because $I$ is convex and $\partial I(X) \subseteq [0,1]$, we have $D_\varepsilon(I) \in [0,1]$. As $\varepsilon \to 0^+$, we have $D_\varepsilon(I) \to I'(c^-)$. The dominated convergence theorem implies
		\begin{align*}
			\lim_{\varepsilon \to 0^+} \int_{\mathcal{I}} D_\varepsilon(I) \diff \pi_{\mathcal{I}}(I | r) = \int_{\mathcal{I}} I'(c^-) \diff \pi_{\mathcal{I}}(I | r).
		\end{align*}
		Furthermore, we have
		\begin{align*}
			\int_{\mathcal{I}} D_\varepsilon(I) \diff \pi_{\mathcal{I}}(I | r) = \frac{\bar{I}_r(c) - \bar{I}_r(c - \varepsilon)}{\varepsilon}.
		\end{align*}
		Taking the limit as $\varepsilon \to 0^+$ yields $\bar{I}_r'(c^-)$. Hence, the claim holds.
		
		\emph{Part 3: Incentive compatibility and optimal effort.}
		Fix type $t$. Because $G$ is a product mechanism, by Part 2 the interim expected utility of type $t$ satisfies
		\begin{align*}
			U_1(\rho, \delta, G, t) &= \int_M \left( \int_{\mathcal{I}} u_1(I, \delta(m, t), t) \diff \pi_{\mathcal{I}}(I | \rho(t)) \right) \diff \pi_M(m | \rho(t)) \\
			&= \int_M u_1(I_t, \delta(m, t), t) \diff \pi_M(m | \rho(t)).
		\end{align*}
		Consider a deviation in $G$ to report $r \in T$ and to choose a constant effort strategy $\eta(m) = e$ for all $m \in M$. The  equilibrium condition for $(\rho, \delta)$ implies
		\begin{align*}
			U_1(\rho, \delta, G, t) \ge \int_M \left( \int_{\mathcal{I}} u_1(I, e, t) \diff \pi_{\mathcal{I}}(I | \rho(r)) \right) \diff \pi_M(m | \rho(r)) = u_1(\bar{I}_{\rho(r)}, e, t).
		\end{align*}
		Setting $r = t$, we obtain $U_1(\rho, \delta, G, t) \ge u_1(\bar{I}_{\rho(t)}, e, t) = u_1(I_t, e, t)$ for all $e \in E$. Hence, we have $
			U_1(\rho, \delta, G, t) \ge \max_{e \in E} u_1(I_t, e, t)$.
		Moreover, for all $m \in M$, it holds that $u_1(I_t, \delta(m, t), t) \le \max_{e \in E} u_1(I_t, e, t)$. Hence, we have
		\begin{align*}
			U_1(\rho, \delta, G, t) \le \max_{e \in E} u_1(I_t, e, t).
		\end{align*}
		Therefore, $U_1(\rho, \delta, G, t) = \max_{e \in E} u_1(I_t, e, t)$. This equality requires $u_1(I_t, \delta(m, t), t)$ to equal its maximum almost everywhere. We conclude that $\delta(m, t) \in \argmax_{e \in E} u_1(I_t, e, t)$ for $\pi_M(\cdot | \rho(t))$-almost every $m \in M$. 
		
		Finally, we verify that the persuasion mechanism $(I_t)_{t \in T}$ is incentive compatible. Using the equality $U_1(\rho, \delta, G, t) = \max_{e \in E} u_1(I_t, e, t)$, we have
		\begin{align*}
			\max_{e \in E} u_1(I_t, e, t) \ge u_1(\bar{I}_{\rho(r)}, e, t)
		\end{align*}
		for all $r \in T$ and $e \in E$. Hence, we have that
		\begin{align*}
			\max_{e \in E} u_1(I_t, e, t) \ge \max_{e \in E} u_1(I_r, e, t)
		\end{align*}
		for all $r \in T$. Thus, the persuasion mechanism $(I_t)_{t \in T}$ is incentive compatible.
		
		\emph{Part 4: Measurable selection and designer's payoff.}
		For each $t \in T$, define the set of optimal efforts $\Phi(t) \coloneqq \argmax_{e \in E} u_1(I_t, e, t)$. Because $u_1$ is continuous and $E$ is compact, $\Phi(t)$ is non-empty and compact. Define the correspondence $\Psi(t) \coloneqq \argmax_{e \in \Phi(t)} u_0(I_t, e, t)$. By the Measurable Maximum Theorem \citep[Theorem 18.19]{aliprantis_infinite_2006}, $\Psi$ admits a measurable selection $\eta \colon T \to E$. Because $\eta(t) \in \Phi(t)$ for all $t \in T$, $\eta$ is a valid effort selection induced by the IC persuasion mechanism $(I_t)_{t \in T}$. In Part 3, we established that $\delta(m, t) \in \Phi(t)$ for $\pi_M(\cdot | \rho(t))$-almost every $m \in M$. By the definition of $\Psi$, it holds that $u_0(I_t, \delta(m, t), t) \le u_0(I_t, \eta(t), t)$ for $\pi_M(\cdot | \rho(t))$-almost every $m \in M$. Taking the expectation with respect to $\pi_M(\cdot | \rho(t))$, we have
		\begin{align*}
			\int_M u_0(I_t, \delta(m, t), t) \diff \pi_M(m | \rho(t)) \le u_0(I_t, \eta(t), t).
		\end{align*}
		Taking the expectation with respect to $P$ over $T$, we have
		\begin{align*}
			V_0(\rho, \delta, G) 
			&= \int_T \int_M u_0(I_t, \delta(m, t), t) \diff \pi_M(m | \rho(t)) \diff P(t) \\
			&\le \int_T u_0(I_t, \eta(t), t) \diff P(t).
		\end{align*}
		Therefore, the IC persuasion mechanism $(I_t)_{t \in T}$ induces a payoff $v \ge V_0(\rho, \delta, G)$.
		\end{proof}

			An \emph{activation mechanism} is a direct coordination mechanism $(T, \mathcal I \times E, \pi)$ in which: for all $r\in T$, there exist a probability measure $\mu_r$ on $E$ and a measurable function $f_r\colon E\to \mathcal I$ such that, for all measurable $\mathcal J \subseteq \mathcal I,\,F\subseteq E$, we have
		\begin{align*}
			\pi(\mathcal J \times F|r) = \mu_r(f_r^{-1}(\mathcal J)\cap F);
		\end{align*}
		equivalently, letting $\operatorname{id}_E$ be the identity $e\mapsto e$ on $E$, and $(f_r, \operatorname{id}_E)_\# \mu_r $ the pushforward measure of $\mu_r$ by the function $(f_r, \operatorname{id}_E)$, we have $ \pi(\mathcal J \times F|r) =(f_r, \operatorname{id}_E)_\# \mu_r (\mathcal J \times F)$. If the effort recommendation is deterministic or the information policy does not vary with the recommendation, then an activation mechanism is a product mechanism, but in general, an activation mechanism allows for correlation between the effort and the information policy following a report. 
		
		The following result shows that activation mechanisms are without loss of generality.
		\begin{proposition} \label{app:prop:disintegration_activation}
			For every incentive compatible direct coordination mechanism $G$, there exists an activation mechanism $F$ that is: incentive compatible, and such that
			\begin{align*}
				V_0(\rho^*, \delta^*, G) = V_0(\rho^*, \delta^*, F).
			\end{align*}
		\end{proposition}
		The proof has two steps. First, we apply the disintegration theorem to the distribution $\pi(\cdot | r)$, for every report $r$, to conclude that the random information policy is allowed to depend on the report only through the realized effort. Second, we observe that the conditional expectation of the payoff of the designer and the agent given a type and an effort level in an incentive compatible direct coordination mechanism corresponds to the payoff given the expected information policy.
		
		\begin{proof}
			Fix an incentive compatible direct coordination mechanism $G=(T, E,\pi)$.
			
			\emph{Part 1: Construction of the activation mechanism.}
			By the disintegration theorem \citep[Theorem 2.18]{cinlar_statistics}, there exist a probability measure $\mu_r$ on $E$ and a measurable function $\nu (\cdot | r, \cdot )\colon \mathscr B_{\mathcal I}\times E\to[0,1]$ such that, for all measurable $\mathcal J \subseteq \mathcal I,\,F\subseteq E$, we have
			\begin{align*}
				\pi(\mathcal J \times F|r) =	\int_F \nu(\mathcal J|r,e)\diff\mu_r(e).
			\end{align*}
			For every $(r,e)\in T\times E$, define $
			f_r(e)\coloneqq\int_{\mathcal I}I\diff\nu(I|r,e)$, and note that $f_r(e)\in\mathcal I$. Moreover, $(r,e)\mapsto f_r(e)$ is measurable. Define the direct mechanism $H=(T,E,\pi^H)$ by
			\begin{align*}
				\pi^H(A|r)
				\coloneqq
				\int_E \mathbf 1_A(f_r(e),e)\diff\mu_r(e),
				\qquad A\in\mathscr B(\mathcal I \times E);
			\end{align*}
			note that measurability of $(r,e)\mapsto f_r(e)$ implies that $H$ is a valid coordination mechanism. For every measurable $\mathcal J\subseteq\mathcal I$ and $F\subseteq E$,
			\begin{align*}
				\pi^H(\mathcal J \times F|r)
				=
				\mu_r(f_r^{-1}(\mathcal J) \cap F).
			\end{align*}
			Hence $H$ is an activation mechanism.
			
			\emph{Part 2: A property of $u_0$ and $u_1$.}
			By Part 2 of the proof of Proposition \ref{app:prop:revelation_product}, for every $i\in\{0,1\}$, $r\in T$, $e,a\in E$, and $t\in T$,
			\begin{align*}
				u_i(f_r(e),a,t)
				=
				\int_{\mathcal I}u_i(I,a,t)\diff\nu(I|r,e).
			\end{align*}
			
			\emph{Part 3: Incentive compatibility and the designer's payoff.}
			Fix $t,r\in T$ and a measurable function $\eta\colon E\to E$. By the construction of $H$, and Part 2, for every $i\in\{0,1\}$,
			\begin{equation}\label{app:eq:activation_payoffs}
				\begin{aligned}
					&\int_{\mathcal I \times E}
					u_i(I,\eta(e),t)\diff\pi^H(I,e|r)\\
					&\quad=
					\int_E u_i(f_r(e),\eta(e),t)\diff\mu_r(e)\\
					&\quad=
					\int_E\left(
					\int_{\mathcal I}u_i(I,\eta(e),t)\diff\nu(I|r,e)
					\right)\diff\mu_r(e)\\
					&\quad=
					\int_{\mathcal I \times E}
					u_i(I,\eta(e),t)\diff\pi(I,e|r).
				\end{aligned}
			\end{equation}
			Taking $r=t$ and $\eta(e)=e$, we obtain
			\begin{align*}
				U_i(\rho^*,\delta^*,H,t)
				=
				U_i(\rho^*,\delta^*,G,t).
			\end{align*}
			Since $G$ is incentive compatible, \eqref{app:eq:activation_payoffs} implies
			\begin{align*}
				U_1(\rho^*,\delta^*,H,t)
				&=U_1(\rho^*,\delta^*,G,t)\\
				&\ge
				\int_{\mathcal I \times E}
				u_1(I,\eta(e),t)\diff\pi(I,e|r)\\
				&=
				\int_{\mathcal I \times E}
				u_1(I,\eta(e),t)\diff\pi^H(I,e|r).
			\end{align*}
			Thus $H$ is incentive compatible. Following the same steps for Sender's payoff, we conclude that
			\begin{align*}
				V_0(\rho^*,\delta^*,H)
				=
				V_0(\rho^*,\delta^*,G).
			\end{align*}
		\end{proof}

		\subsubsection{Sufficient conditions for persuasion mechanisms to be without loss of generality} \label{app:sec:disintegration}

	We consider the case in which the cost function is isoelastic, the designer knows the outside option, and Receiver never finds effort 1 optimal.
	\begin{assumption}\label{app:ass:isoelastic}
		There exist $\gamma>1$, $c^*\in(x_0,1)$, and
		$\underline\lambda>0$ such that: for all $e$, $k(e)=e^\gamma$, $T\subseteq\{c^*\}\times[\underline\lambda,1]$, and $\gamma\underline\lambda\ge\Delta I_0(c^*)$.
	\end{assumption}

	Under this assumption, persuasion mechanisms are without loss of generality. A coordination mechanism is \emph{deterministic} if every report is associated with a deterministict pair of an information policy and a message. A deterministic coordination mechanism is a product mechanism.
	\begin{proposition}\label{app:prop:isoelastic}
		If Assumption \ref{app:ass:isoelastic} holds, then, for every incentive compatible activation mechanism $G$, there exists an activation mechanism $F$ that is: deterministic, incentive compatible, and such that
		\begin{align*}
			V_0(\rho^*, \delta^*, G) \le  V_0(\rho^*, \delta^*, F).
		\end{align*}
		\end{proposition}
	
	\begin{proof}
		Fix an incentive compatible activation mechanism $G=(T,E,\pi)$. Let $\overline v\coloneqq\Delta I_0(c^*)$. For every $r\in T$, let $\mu_r$ and $f_r$ represent the activation mechanism.
		
		\emph{Part 1: Upper censorships.}
		Define $H\colon[0,\overline v]\to[0,1]$ by
		\begin{align*}
			H(v)\coloneqq1-F_0(C(v+\underline I(c^*))).
		\end{align*}
		By Lemmata \ref{app:lem:image}--\ref{app:lem:C}, the upper censorship $I^{C(v+\underline I(c^*))}$ has net informativeness $v$ at $c^*$ and left derivative $1-H(v)$ at $c^*$. Moreover, by Lemma \ref{app:lem:boundrightderivative}, for every $I\in\mathcal I$,
		\begin{equation}\label{app:eq:isoelastic_bound}
			1-I'(c^{*-})\le H(\Delta I(c^*)).
		\end{equation}
		The function $H$ is continuous, nonnegative, and nonincreasing by the properties of $F_0$ and $C$, and concave by Lemma \ref{app:lem:concavity}.
		
		\emph{Part 2: A consequence of incentive compatibility.}
		For every $t=(c^*,\lambda)\in T$ and $v\in[0,\overline v]$, the unique maximizer of $a v-\lambda a^\gamma$ over $a\in E$ is $(v/(\gamma\lambda))^{1/(\gamma-1)}$, which belongs to $E$ because $\overline v\le\gamma\lambda$. Hence,
		\begin{equation}\label{app:eq:isoelastic_value}
			\max_{a\in E}\{a v-\lambda a^\gamma\}
			=
			\frac{\gamma-1}{\gamma}
			(\gamma\lambda)^{-1/(\gamma-1)}
			v^{\gamma/(\gamma-1)}.
		\end{equation}
		Fix $t\in T$. The function $e\mapsto(\Delta f_t(e)(c^*)/(\gamma\lambda))^{1/(\gamma-1)}$ is measurable. Choosing this effort following report $t$ weakly improves Receiver's payoff at every recommendation. Incentive compatibility and uniqueness of the maximizer therefore imply
		\begin{equation}\label{app:eq:isoelastic_effort}
			e=
			\left(\frac{\Delta f_t(e)(c^*)}{\gamma\lambda}\right)^{1/(\gamma-1)}
			\quad\text{for $\mu_t$-almost every $e$.}
		\end{equation}
		
		Because $c^*>x_0$, the policy-independent term in Receiver's payoff is zero. By \eqref{app:eq:isoelastic_value} and \eqref{app:eq:isoelastic_effort}, incentive compatibility against report $r$ and an optimal effort choice implies
		\begin{align*}
			\int_E \bigl(\Delta f_t(e)(c^*)\bigr)^{\gamma/(\gamma-1)}
			\diff\mu_t(e)
			\ge
			\int_E \bigl(\Delta f_r(e)(c^*)\bigr)^{\gamma/(\gamma-1)}
			\diff\mu_r(e).
		\end{align*}
		Because $t,r\in T$ are arbitrary, there exists a number $m\in[0,\overline v^{\gamma/(\gamma-1)}]$ such that
		\begin{equation}\label{app:eq:isoelastic_moment}
			\int_E \bigl(\Delta f_r(e)(c^*)\bigr)^{\gamma/(\gamma-1)}
			\diff\mu_r(e)=m
			\quad\text{for every $r\in T$.}
		\end{equation}
		
		\emph{Part 3: Concavity.}
		Define $\Phi\colon[0,\overline v^{\gamma/(\gamma-1)}]\to\mathbb R$ by $\Phi(z)\coloneqq z^{1/\gamma}H(z^{(\gamma-1)/\gamma})$. We claim that $\Phi$ is concave. For $z$ in the interior of its domain, letting $v=z^{(\gamma-1)/\gamma}$, its right derivative satisfies
		\begin{align*}
			\Phi'_+(z)
			=
			\frac1\gamma\frac{H(v)}v
			+
			\frac{\gamma-1}{\gamma}H'_+(v).
		\end{align*}
		The function $v\mapsto H(v)/v$ is nonincreasing because $H$ is nonnegative and nonincreasing, and $H'_+$ is nonincreasing because $H$ is concave. Thus $\Phi'_+$ is nonincreasing. Continuity of $\Phi$ establishes concavity on the closed interval.
		
		\emph{Part 4: Construction of $F$ and Sender's payoff.}
		Let $v^*\coloneqq m^{(\gamma-1)/\gamma}$ and choose the information policy $I^*\coloneqq I^{C(v^*+\underline I(c^*))}$. For every $t=(c^*,\lambda)\in T$, let $e_t\coloneqq(v^*/(\gamma\lambda))^{1/(\gamma-1)}$. Define $F=(T,E,\pi^F)$ by $\pi^F(\cdot|t)\coloneqq\delta_{(e_t,I^*)}$. The map $t\mapsto e_t$ is measurable, so $\pi^F$ is a Markov kernel. Hence $F$ is a direct activation mechanism assigning a single policy--effort pair to every report.
		
		Every report yields the same policy $I^*$. The assigned effort $e_t$ maximizes type $t$'s payoff given $I^*$, so $F$ is incentive compatible.
		
		Finally, by \eqref{app:eq:isoelastic_bound}, \eqref{app:eq:isoelastic_effort}, Jensen's inequality, and \eqref{app:eq:isoelastic_moment},
		\begin{align*}
			U_0(\rho^*,\delta^*,G,t)
			&\le
			(\gamma\lambda)^{-1/(\gamma-1)}
			\int_E
			\Phi\left(\bigl(\Delta f_t(e)(c^*)\bigr)^{\gamma/(\gamma-1)}\right)
			\diff\mu_t(e)\\
			&\le
			(\gamma\lambda)^{-1/(\gamma-1)}\Phi(m)
			=
			e_t H(v^*)
			=
			U_0(\rho^*,\delta^*,F,t).
		\end{align*}
		Hence Sender's interim payoff weakly increases for every type, which implies $V_0(\rho^*,\delta^*,G)\le V_0(\rho^*,\delta^*,F)$.
	\end{proof}

	We consider the case of linear $k$, known outside option, and nonincreasing reverse hazard rate of the attention cost; i.e., the same hypothesis as Proposition \ref{prop:reverse}. We show that deterministic mechanisms are without loss. As an implication of this result and Proposition \ref{app:prop:revelation_product}, persuasion mechanisms are without loss under Assumption \ref{ass:prop:reverse}.
	
	\begin{proposition}\label{app:prop:linear}
			If Assumption \ref{ass:prop:reverse} holds, then, for every incentive compatible direct coordination mechanism $G$, there exists a direct coordination mechanism $F$ that is: deterministic, incentive compatible, and such that
			\begin{align*}
				V_0(\rho^*, \delta^*, G) \le  V_0(\rho^*, \delta^*, F).
			\end{align*}
	\end{proposition}
\begin{proof}
	
	\emph{Part 1: preliminary observations.}
	Let $\bar v \coloneqq \Delta I_0(c^*)$. The condition
	$\Gmarg(\bar v)>0$ implies $c^*<1$. Recall that, for
	$v\in[0,\bar v]$,
	\[
	H(v)\coloneqq 1-F_0\bigl(C(v+\underline I(c^*))\bigr).
	\]
	By Lemma \ref{app:lem:concavity}, $H$ is continuous, positive, and
	strictly concave. Moreover, for $v\in(0,\bar v)$,
	\[
	H'(v)
	=
	-\frac{1}{c^*-C(v+\underline I(c^*))},
	\]
	so $H'(v)<0$ and $H'(v)\to-\infty$ as $v\to \bar v$.
	
	A deterministic incentive-compatible mechanism that assigns an upper
	censorship $I$ with $\Delta I(c^*)=v$ to every report induces effort
	$1$ for $\lambda<v$ and effort $0$ for $\lambda>v$. Hence, Sender's payoff is $\Gmarg (v)H(v)$. We fix $v^*\in\argmax_{v\in[0,\bar v]}\Gmarg (v)H(v)$ and note that: $v^*$ exists, and is unique because $\Gmarg$ is log-concave and $H$ is positive and strictly concave. Moreover, $v^*\in(0,\bar v)$,  because: $\Gmarg(0)=0$, $\Gmarg(\bar v)>0$, and $H'(v)\to-\infty$ as $v\to \bar v$.
	
	\emph{Part 2: $H(v^*)\Gmarg (v^*)$ is an upper bound on Sender's payoff.}
	Fix any direct, incentive-compatible coordination mechanism
	$M=(T,E,\pi)$. Let $\pi_E(\cdot | \lambda)$ denote the marginal distribution over $ E$ following report $(c^*,\lambda)$. We define the expected effort and benefit by, for all $\lambda$, $q(\lambda)\coloneqq \int _E e\diff \pi_E(e|\lambda )$, and  $b(\lambda)\coloneqq \int _{E\times \mathcal I}e\Delta I(c^*)\diff\pi(e, I|(c^*, \lambda))$.
	Extend $q$ by zero outside the support of $\Gmarg$, and let
	$\lambda_+$ denote the supremum of that support. 
	The envelope theorem and incentive compatibility imply that
	$
	U_1(\rho^*,\delta^*,M,(c^*,\lambda)).
	=
	U_1(\rho^*,\delta^*,M,(c^*,\lambda_+)) +\int_\lambda^{\lambda_+}q(z)\diff z$. 
	Because $c^*>x_0$, Receiver can choose zero effort and obtain payoff
	zero. Hence $U(\lambda_+)\ge0$, and therefore
	$
	U(\lambda)\ge\int_\lambda^{\lambda_+}q(z)\diff z$.	Let $\gmarg$ denote the density of the attention cost. Sender's expected payoff is $V_0(\rho^*, \delta^*, M) = \int_0^1 S(\lambda) \gmarg(\lambda) \diff\lambda$, in which $S(\lambda) = \int_{E\times\mathcal{I}} e(1-I'(c^{*-})) \diff\pi(e,I | (c^*, \lambda))$. By Proposition \ref{app:prop:uppercensorships}, the probability of action 1 for any information policy is bounded by that of an upper censorship with the same net informativeness, i.e., $S(\lambda) \le \int_{E\times\mathcal{I}} e H(\Delta I(c^*)) \diff\pi(e,I | (c^*,\lambda))$.

	Define
	$D =
	\{(e,B)\in[0,1]\times\mathbb R_+
	\mid 0\le B\le e\bar v\}$, and the function $L\colon D\to\mathbb R$ by
	\[
	L(e,B)\coloneqq
	\begin{cases}
		eH(B/e), & e>0,\\
		0, & (e,B)=(0,0).
	\end{cases}
	\]
	The function $-L$ is the perspective function of $-H$, so $L$ is continuous and concave on $D$ \citep[Section 2.2]{hiriart-urruty_fundamentals_2004}. Hence, if $q(\lambda)>0$,
	Jensen's inequality implies
	\[
	\begin{aligned}
		S(\lambda)
		&\le
		q(\lambda)H\left(\frac{b(\lambda)}{q(\lambda)}\right).
	\end{aligned}
	\]
	If $q(\lambda)=0$, then $e=0$, $\pi_E(\cdot | \lambda )$-almost surely, so
	$S(\lambda)=0$. Moreover, by concavity of $H$, we have, for every $H'(v^*)$ in the superdifferential of $H$ at $v^*$, for every $v\in(0,1)$, $H(v)\le H(v^*)-H'(v^*)(v^*-v)$. Using $b(\lambda)/q(\lambda)$ in the previous inequality, and the identity $b(\lambda) = U_1(\rho^*, \delta^*, M, ( c^*, \lambda )) + \lambda q(\lambda)$, these observations imply the upper bound given by
	\[
	\int _{[0,1]}S(\lambda) \gmarg(\lambda )\diff \lambda \le 	\int _{[0,1]} q(\lambda) [H(v^*) - H'(v^*)(v^* - \lambda)] + H'(v^*) U_1(\rho^*, \delta^*, M, (c^*, \lambda)) \diff \Gmarg(\lambda).
	\]
		Applying integration by parts to the utility term of the right-hand side, and using $H'(v^*)\le 0$, we obtain
	\[
	\int _{[0,1]}S(\lambda) \gmarg(\lambda )\diff \lambda \le \int_0^1 q(\lambda) \bigl[ f(\lambda)(H(v^*)-H'(v^*)(v^*-\lambda)) + H'(v^*) \Gmarg(\lambda) \bigr] \diff\lambda.
	\]
	Let $
	\alpha\coloneqq-\frac{H'(v^*)}{H(v^*)}>0$. By definition of $v^*$, we have $\alpha=\frac{f(v^*)}{\Gmarg (v^*)}$.
	Because  $\lambda\mapsto f(\lambda)/\Gmarg(\lambda)$ is nonincreasing,
	$f(\lambda)-\alpha \Gmarg(\lambda)$ is nonnegative for
	$\lambda<v^*$ and nonpositive for $\lambda>v^*$, almost everywhere. Hence, in the right-hand side of the above equation, the coefficient of $q(\lambda)$ in the integrand, $H(v^*)\bigl[	f(\lambda)-\alpha \Gmarg(\lambda)	+\alpha(v^*-\lambda)f(\lambda)	\bigr]$, is nonnegative if $\lambda<v^*$ and nonpositive if $\lambda>v^*$. We conclude that, by $q(\lambda)\in[0,1]$,
		\begin{align*}
			V_0(\rho^*,\delta^*,M) 	&\le \int_0^{v^*}H(v^*)\bigl[ f(\lambda)-\alpha \Gmarg(\lambda)		+\alpha(v^*-\lambda)f(\lambda) 	\bigr]\diff\lambda\\
			&\le H(v^*)\Gmarg (v^*),
		\end{align*}
in which the second inequality follows from integration
by parts.

	\emph{Part 3: construction of a deterministic, incentive compatible, direct coordination mechanism.} Define $I^* \coloneqq I^{C(v^*+c^*-x_0)}$ and, for all $\lambda$, $e_\lambda^* \coloneqq \mathbf 1_{\{l \in [0,1] \mid l\le v^*\}}(\lambda)$. The coordination mechanism $F$ assigning the pair $(I^*,e_\lambda^*)$ to report $(c^*,\lambda)$ is a direct activation mechanism. We proceed to show that $F$ is incentive compatible.
	Every report induces the
	same information policy $I^*$, and $
	\Delta I^*(c^*)=v^*$. For every type $(c^*,\lambda)$, the recommended effort
	$e_\lambda^*$ maximizes
	$
	a\mapsto a(v^*-\lambda)
	$
	over $a\in[0,1]$. Because every report induces the same policy, no report
	can yield a higher payoff together with any effort deviation. Thus $F$
	is incentive compatible.
	
	By Part 1, $
	V_0(\rho^*,\delta^*,F)=\Gmarg (v^*)H(v^*)$.
	Applying Part 2 yields $
	V_0(\rho^*,\delta^*,F)
	\ge
	V_0(\rho^*,\delta^*,G)$, 
	which proves the result.
	
\end{proof}

		\subsubsection{Example} \label{app:sec:examples}
		The following example provides a setting in which a coordination mechanism yields a strictly higher payoff to Sender than any persuasion mechanism. To establish this observation, we use the equivalence between persuasion mechanisms and experiments (Theorem \ref{thm:mechanisms}). 
		
		The example clarifies the source of the loss of optimality of persuasion mechanisms. The improvement arises because the coordination mechanism features a communication stage between the reporting stage and the effort choice. Receiver learns about the experiment from the effort recommendation. This observation is discussed in detail in Section \ref{app:sec:rev:discussion}.

		\begin{example}\label{app:example:knownoutside}
		In this example, Sender knows the outside option, the cost $k$ is linear, and $F_0$ is the uniform distribution. The outside option is known to be $3/4$, and there are two types, corresponding to attention costs $\lambda_L = \frac{9}{800}$ and $\lambda_H = \frac{1}{50}$, with probabilities $\frac{8}{9}$ and $ \frac{1}{9}$, respectively. As a preliminary observation, we connect an information policy's censorship point with its net informativeness at the outside option. Consider an upper censorship $I^t$, with $t\in [1/2, 3/4]$, and let $v\coloneqq I^t(3/4)$. By the definition of $C$ (see Lemma \ref{app:lem:C}), we have
		\begin{align*}
			t &= C(v) \coloneqq \frac{3}{4} - \sqrt{\frac{9}{16} - 2v},
		\end{align*}
		and the probability of a mean in excess of $c$ is
		\begin{align*}
			H(v) &\coloneqq  1-C(v) = \frac{1}{4} + \sqrt{\frac{9}{16} - 2v}.
		\end{align*}

		First, we consider the optimal product mechanism. By Proposition \ref{app:prop:revelation_product} and Theorem \ref{thm:mechanisms}, the optimal product mechanism gives the same payoff to Sender as the optimal single information policy. The optimal information policy can induce high effort from both types or only from the low-cost type. Consider an information policy that induces high effort from both types. By Proposition \ref{app:lem:boundrightderivative}, there exists a weakly better information policy that is an upper censorship, so restricting attention to upper censorships is without loss of optimality. By Lemma \ref{app:lem:C}, the optimal upper censorship has the lowest net informativeness at $3/4$: $\lambda_H$. Hence, Sender's payoff conditional on each type is:
		\begin{align*}
			H\Big (\frac{1}{50} + \frac{1}{4}\Big ) &= \frac{1}{4} + \sqrt{\frac{25 - 16}{400}} = {\frac{2}{5}}.
		\end{align*}
	
		Consider an information policy that leads only the low type to exert effort. By the preceding argument, we search for the optimal upper censorship, and the optimal censorship point induces a net informativeness at $3/4$ of $\lambda_L$. Hence, Sender's payoff conditional on the low type is:
		\begin{align*}
			H\left(\frac{9}{800} + \frac{1}{4}\right) &= \frac{1}{4} + \sqrt{\frac{50 - 18}{800}} ={\frac{9}{20}};
		\end{align*}
		and Sender gets 0 given the high type. Therefore, the highest payoff from a product mechanism and a persuasion mechanism is $2/5=336/840$. Letting $s \coloneqq C(1/50 + 1/4)$, this payoff is obtained by the direct and incentive compatible coordination mechanism $(T, E, \pi)$ with $\pi(I^s, 1|t)=1$ for all $t$ --- in the generalized principal-agent setting of Section \ref{app:sec:rev:rev} --- or by the persuasion mechanism $(I_r)$ with $I_r=I^s$ for all $r$ --- in the game described in the main text.

		We compare the above mechanism with the following activation mechanism. Type $\lambda_H$ is assigned a lottery that recommends effort 1 with probability $\frac{31}{112}$, and 0 with complementary probability. Importantly, given the high report, effort 1 is associated with the upper censorship $I^s$, whereas 0 is associated with no information. Type $\lambda_L$ is assigned effort 1 and an upper censorship $I$ with $\Delta I(3/4)=v_L=7/512$. This mechanism is not a product mechanism because of the random effort recommendation and the fact that the information policy assigned to the high type depends on the recommended effort. Formally, we study the direct coordination mechanism $(T, E, \overline \pi)$ in which, given $u=C(v_L+\underline I(3/4))$, we have $\overline \pi(I^s, 1|(3/4, \lambda_H))=\frac{31}{112}=:p$,  $\overline \pi(\underline I, 0|(3/4, \lambda_H))=1-p$, and  $\overline \pi(I^u, 1|(3/4, \lambda_L))=1$.

		This activation mechanism is incentive compatible.	In the truthful-obedient equilibrium, the low type gets $\frac{7}{512} - \frac{9}{800}=\frac{31}{12800}$, and the best payoff from a deviation is obtained by reporting a high cost and choosing effort 1 unless the information policy is uninformative, $\frac{31}{112} (v_H-\lambda_L ) =\frac{31}{12800}$. In the truthful-obedient equilibrium, the high type gets $0$, and the best payoff from a deviation is obtained by reporting a low type and choosing effort 0, because $7/512<1/50$. Conditional on the low type, Sender's payoff is
		\begin{align*}
			H\left(\frac{7}{512}+\frac{1}{4}\right) &=\frac{1}{4} + \sqrt{\frac{32 - 14}{512}} = {\frac{7}{16}}.
		\end{align*}
		Conditional on the high type, Sender's payoff is $H\left(\frac{1}{50} +\frac{1}{4} \right) = \frac{2}{5}$ with probability $p = \frac{31}{112}$. Therefore,  the expected payoff is
		\begin{align*}
				\frac{8}{9} \cdot  \frac{7}{16}  + 		\frac{1}{9} \cdot \frac{31}{112} \cdot \frac{2}{5} =	\frac{7}{18} + \frac{31}{2520} = {\frac{337}{840}}.
		\end{align*}
		
		Hence, $(T, E, \overline \pi)$ induces a strictly higher payoff than any persuasion mechanism.
		\end{example}
		
		\subsubsection{Discussion and economic applications}\label{app:sec:rev:discussion}
		In this section, we contrast persuasion mechanisms and coordination mechanisms from an economic perspective, in light of the results established in Section \ref{app:sec:rev:rev}, \ref{app:sec:disintegration}, and \ref{app:sec:examples}. We explain that the additional communication possibilities of coordination mechanisms are inessential for the applications that motivate the model.

		First, we isolate a property of coordination mechanism, the informativeness of effort recommendations, and argue that this property is responsible for the potential loss of optimality of persuasion mechanisms. Then, in light of this property, we argue that coordination mechanisms capture a more general strategic interaction than necessary for the leading interpretation of the model. Instead, persuasion mechanisms and single experiments provide accurate descriptions of the applications. Therefore, even if persuasion mechanisms are not without loss relative to coordination mechanisms, they are useful to capture environments in which Receiver chooses attention after observing an information structure but before receiving any additional recommendation. The applications that motivate this paper, in which information is gathered through platforms or newspapers, naturally fit this timing. In what follows, we adopt the term ``experiment'' to refer informally to an information policy or a signal in order to ease the connection with the applications of the model.
		
		\paragraph{The informative role of effort recommendations} 	
		In a coordination mechanism, Sender draws an experiment and an effort recommendation after Receiver's report. Then, Sender communicates the recommendation to Receiver. Receiver observes the recommendation before choosing effort, but the experiment is not observed at this stage. Hence, Receiver chooses effort knowing that the experiment follows a distribution inferred by the coordination mechanism, the report, and the observed recommendation. Therefore, the recommendation conveys information about the realized experiment that is not available at the reporting stage, which we refer to as ``interim information''. Sender uses this interim information as a means to correlate the uncontractible effort with the experiment. Deterministic mechanisms rule out this correlation device, because the assigned experiment is perfectly inferred from the agent's report. More generally, any product mechanism does not convey interim information, and this observation is the key intuition for Proposition \ref{app:prop:revelation_product}.

		By Proposition \ref{app:prop:disintegration_activation}, the stochastic nature of the effort recommendation is responsible for the loss of generality of persuasion mechanisms. Activation mechanisms are without loss (Proposition \ref{app:prop:disintegration_activation}), whereas product mechanisms are with loss (Example \ref{app:example:knownoutside}). An activation mechanism reduces to a product mechanism in two cases: if the effort recommendation is deterministic, or if it is stochastic but uninformative (i.e., the assigned information policy is constant across recommendations). Because uninformative stochasticity conveys no interim information, an activation mechanism can strictly improve upon a product mechanism only if it is stochastic and provides interim information.\footnote{A recommendation that is a deterministic function of Receiver's privately known type is random from Sender's ex-ante perspective, but provides no interim information.} Therefore, the possibility of improving upon persuasion mechanisms arises only in conjunction with the additional commitment power required to deploy coordination mechanisms that are not deterministic.

		\paragraph{The informative role of effort recommendations in  applications}
	The interim information conveyed by effort recommendations is not a natural ingredient in the leading interpretations of the model. First, consider a professional (Receiver) who subscribes to a weekly briefing on AI developments, and incorporates the reading into a morning routine, incurring a financial and cognitive cost. The professional uses what she learns to evaluate a new AI service. The publisher (Sender) determines the relative prevalence of information and advertising, with the goal of persuading the professional to purchase the AI service, thereby monetizing the advertisement.\footnote{The utility of Sender in the hypothesis of Proposition \ref{prop:unique} captures this incentive, and is inspired by \citet{gehlbach_government_2014}.} A coordination mechanism would require the publisher to: (i) recognize that a specific brief is uninformative, and (ii) provide the recommendation to Receiver of merely skimming through incoming content. This arrangement requires the publisher to adjust the reader's engagement week-by-week, and to commit to communicating the realized quality of each release. In contrast, a persuasion mechanism captures a plausible rigidity: the publisher commits to an editorial standard, and the reader chooses a level of effort against that standard. In this interpretation, a persuasion mechanism captures the subscriber's ability to select from a menu of multiple ``tracks''; e.g., a subscriber to \emph{TLDR} can choose among various free and ad-supported editions, such as \emph{TLDR AI}, \emph{TLDR Tech}, or \emph{TLDR Web Dev}.\footnote{\url{https://tldr.tech/newsletters}.}

	A similar rigidity applies to the exposure to social-media platforms and news-aggregation apps. For example, consider a platform that designs an algorithm to mix current affairs with advertising, and a user who chooses her overall exposure. It is possible that the platform attempts to use a coordination mechanism via push notifications or trending alerts, which can signal that a specific day's feed is highly informative, effectively recommending high effort. However, evaluating whether an algorithmic alert is genuinely breaking news or merely engagement clickbait consumes attention. For this reason, we interpret the cost of effort as the cognitive cost of overall attention, exerted without knowing about specific alerts. This ``precommitment'' choice can be available to the user as a daily device-level time limit, or a focus mode, that pauses notifications. This limit can be interpreted as the probability with which a randomly chosen item in the future feed is processed (tightening the limit reduces this probability), and the effort cost captures the opportunity cost of allocating time to the application. Persuasion mechanisms are relevant in this setting because the platform can allow users to select among multiple feed designs. Even if a coordination mechanism theoretically correlates a ``recommended usage limit'' with the realized quality of the feed, the rigidity of a device-level limit prevents this possibility. Mobile news aggregators such as Flipboard's app illustrate this environment: users select a thematic magazine, such as a feed design, and then enforce a daily exposure limit via device-level tools such as Android Digital Wellbeing.\footnote{See \url{https://android.p2hp.com/digital-wellbeing/index.html}.}

		In these and other examples, the exposure choice is followed by the interaction with the news source. In the interaction stage, the feed adjusts to retain or lose attention at each login, and the user adopts practices such as active reading and binge-scrolling. A dynamic model of information provision with coordination mechanisms can capture this interaction \citep{myerson_multistage_1986,sugaya_revelation_2021}. Overall, the platform-user interaction can be conceptualized in two stages. In the first stage, the platform  sets an editorial practice, and the user chooses exposure. In the second stage, conditional on exposure, a dynamic interaction occurs. The present model provides a first step toward analyzing the first stage.

		\section{Proofs} \label{app:sec:proofs}

		We endow $\mathcal F$ with the $L^1$ norm and the norm topology.	We define the functions $W_\lambda$, given by $F\mapsto \int _{[0,1]} V_\lambda (\Delta I_{ F} (c)) \frac{\partial \gcond}{\partial c}(c|\lambda)\diff c$, and $W$, given by $F\mapsto  \int _{[0,1]} W_\lambda(F) \diff \Gmarg (\lambda)$. The function $g\colon \mathbb R^2\to \mathbb R$ exhibits \emph{increasing differences} if $t\mapsto g(s', t)-g(s, t)$ is nondecreasing for all $s',\,s \in \mathbb R$ with $s< s'$.
		
		Proofs that are mainly technical or follow from known arguments are relegated to Appendix \ref{app:sec:supp}.
		
		\begin{definition}\label{app:def:equilibriumexperiment}
			The experiment $F$ is $W$-\emph{maximal} if $F$ maximizes $W$ on $\mathcal F$. The experiment $\hat F\in \mathcal F$ is an \emph{equilibrium experiment} if there exists an equilibrium $\langle I^*_\bullet, \rho, \eta, \alpha \rangle$ such that $I^*_r = I_{\hat F}$ for all $r \in T$. Receiver's \emph{value of $F\in \mathcal F$} is $V_\lambda (\Delta I_F(c)):=\max_{e \in [0,1]} V _\lambda(e, \Delta I_F(c))$. There are \emph{multiple Sender's payoffs} if there exist equilibria $\langle I^*_\bullet, \rho, \eta, \alpha \rangle$ and $\langle \tilde I^*_\bullet, \tilde \rho, \tilde \eta, \tilde \alpha \rangle$ such that $\hat W(I^*_\bullet, \rho, \eta, \alpha)\ne \hat W( \tilde I^*_\bullet, \tilde \rho, \tilde \eta, \tilde \alpha)$.
		\end{definition}

		\begin{remark}\label{app:rem:extensive}
			Let's fix an equilibrium $\langle I^*_\bullet, \rho, \eta, \alpha \rangle$. For any information policy $I \in \mathcal{I}$, let $I_\bullet$ be the mechanism such that $I_r = I$ for all $r \in T$. By Lemma \ref{lem:receiver}, we have $\eta(c, \lambda, I_\bullet) = \eta_\lambda^* \circ \Delta I(c)$ for some selection $\eta_\lambda^*$ from $\Delta J(c) \mapsto \argmax_{e \in [0,1]} V_\lambda(e, \Delta J(c))$ by Reporting and Effort Optimality. We define $\underline c_\lambda(\Delta I) = \sup \{c \in [0,x_0] \mid \eta_\lambda^* \circ \Delta I(c) = 0\}$, if $\{c \in [0,x_0] \mid \eta_\lambda^* \circ \Delta I(c) = 0\} \neq \emptyset$, and $\underline c_\lambda(\Delta I) = 0$ otherwise. We define $\overline c_\lambda(\Delta I) = \inf \{c \in [x_0, 1] \mid \eta_\lambda^* \circ \Delta I(c) = 0\}$, if $\{c \in [x_0, 1] \mid \eta_\lambda^* \circ \Delta I(c) = 0\} \neq \emptyset$, and $\overline c_\lambda(\Delta I) = 1$ otherwise.
		\end{remark}

		\subsection{Proofs for Section \ref{sec:receiver}}

			\begin{proof}[Proof of Lemma \ref{lem:value}]
				Let's fix Receiver's type $(c, \lambda)$ and $I\in \mathcal I$. By definition of $U_R$, letting $\alpha(c, x)$ be any probability measure over $\{0,1\}$ such that $\alpha (c, x)\big ( \argmax _{a \in \{0,1\}}a(x-c) \big ) =1$ for all $x\in[0,1]$, we have
			\begin{align*}
				\int _{[0,1]} U_R(x, e, c, \lambda ) \diff I'(x) +\lambda k(e) &= \int _{[c, 1]} x-c\diff I'(x)\\
				&\qquad -\big (1-\alpha(c, c)(\{1\})\big )\big (I'(c)-I'(c^-)\big)(c-c),\\
				&= \int _{[c, 1]} x-c\diff I'(x).
			\end{align*}
			Moreover, 
			\begin{align*}
					\int _{[0,1]} U_R(x, e, c, \lambda ) \diff I'(x) +\lambda k(e)  &= (1-c) - \int _{[c, 1]}I'(x)\diff x,\\
					&=	x_0-c +I(c).
			\end{align*}
			in which the first equality is due to Riemann--Stieltjes integration by parts \citep[Lemma 2]{machina_expected_1982} and the second one is due to absolute continuity of $I$. It follows that
			\begin{align*}
				\int _{[0,1]} U_R(x, e, c, \lambda ) \diff I'(x) - \int_{[0,1]} U_R(x, e, c, \lambda) \diff \overline F(x) = \Delta I(c).
			\end{align*}			
		\end{proof}

\begin{proof}[Proof of Lemma \ref{lem:receiver}]
	Let $\langle I^*_\bullet, \rho, \eta, \alpha \rangle$ be an equilibrium. Fix type $(c, \lambda)$ and mechanism $I_\bullet$. Let $r^* = \rho(c, \lambda, I_\bullet)$ and $e^* = \eta(c, \lambda, I_\bullet)$. We define $\varepsilon \odot F =  \varepsilon F + (1- \varepsilon )\overline F$, and note that $\varepsilon \odot F\in\mathcal F$. 
	
	Part 1 follows from the definition of information policies and the Action Optimality condition for $\alpha$.
	
	For Part 2, the Reporting and Effort Optimality condition implies that $(r^*, e^*)$ maximizes Receiver's expected utility:
	\begin{align*}
		(r^*, e^*) \in \argmax_{(r, e) \in T \times [0,1]} \int_{[0,1]} \max_{a \in \{0,1\}} U_R(x, a, e, c, \lambda) \diff (e \odot   I_r )'(x).
	\end{align*}
	By Lemma \ref{lem:value}, for any $r \in T$ and $e \in [0,1]$, we have:
	\begin{align*}
		\int_{[0,1]} \max_{a \in \{0,1\}} U_R(x, a, e, c, \lambda) \diff (e \odot I_r)'(x) = e \Delta I_r(c) - \lambda k(e) + I_{\overline{F}}(c) + x_0 - c.
	\end{align*}
	Therefore, $(r^*, e^*)$ must maximize $e \Delta I_r(c) - \lambda k(e) + I_{\overline{F}}(c) + x_0 - c$ over $T \times [0,1]$. This maximization implies that
	\begin{align*}
		e^* \in \argmax_{e \in [0,1]} e \Delta I_{r^*}(c) - \lambda k(e).
	\end{align*}
	Hence, the effort $\eta(c, \lambda, I_\bullet)$ depends on the mechanism $I_\bullet$ only through the net informativeness of $I_{r^*}$, and satisfies the desired condition.
\end{proof}

\begin{proof}[Proof of Lemma \ref{lem:interval}]
	Let $\langle I^*_\bullet, \rho, \eta, \alpha \rangle$ be an equilibrium, and let's fix $\lambda \in[0,1]$ and $I\in\mathcal I$. Let $I_\bullet$ be the mechanism such that $I_r = I$ for all $r \in T$. We start with three preliminary observations. First, by Lemma \ref{lem:receiver}, the equilibrium effort $\eta(c, \lambda, I_\bullet)$ maximizes $e \Delta I(c) - \lambda k(e)$ over $e \in [0,1]$. Thus, $\eta(c, \lambda, I_\bullet)$ equals $\eta^*\circ \Delta I (c)$ for some selection $\eta^*$ from $\Delta J(c) \mapsto \argmax_{e \in [0,1]} V_\lambda(e, \Delta J(c))$. Second, every selection from $\Delta J(c) \mapsto \argmax_{e \in [0,1]} V_\lambda(e, \Delta J(c))$ is nondecreasing, because $V_\lambda$ satisfies strictly increasing differences, via Corollary \ref{cor:superm} and known results \citep[Theorem 6.3]{topkis_minimizing_1978}. From these observations, it follows that $\eta^*\circ \Delta I $ is nondecreasing on $[0, x_0]$ and nonincreasing on $[x_0, 1]$ because $\Delta I$ is nondecreasing on $[0,x_0]$ and $\Delta I$ is nonincreasing on $[x_0, 1]$.
	
	If $\eta^* (\Delta I(x_0))=0$, then every cutoff $c$ has $\eta^* (\Delta I(c))=0$, by the above observations. Let's suppose that $\eta^* (\Delta I(x_0))>0$. We define $\underline c_\lambda(\Delta I) = \sup \{c\in[0,x_0] : \eta^*\circ \Delta I (c)=0\}$, if $\{c\in[0,x_0] : \eta^*\circ \Delta I (c)= 0\}\ne \emptyset$, and $\underline c_\lambda(\Delta I)=0$ otherwise. We define $\overline c_\lambda(\Delta I) = \inf \{c\in[x_0, 1] : \eta^*\circ \Delta I (c)= 0\}$, if $\{c\in[x_0, 1] : \eta^*\circ \Delta I (c)= 0\}\ne\emptyset$, and $\overline c_\lambda(\Delta I) =1$ otherwise. First, we note that $\eta^*\circ \Delta I (c)> 0$ only if: $c\in [\underline c, \overline c]$; second, $c\in (\underline c, \overline c)$ only if: $\eta^*\circ \Delta I (c)> 0$. Thus, for all $\lambda$ we have that: either no type $(c,\lambda)$ chooses positive effort or $\eta^{-1}_\lambda((0, 1])$ is an interval. The result follows from the fact that $\Delta I(x_0)\ge\Delta I(c)$ for all $c\in[0,1]$.
\end{proof}
	\subsection{Proofs for Section \ref{sec:mechanisms}} \label{app:sec:mechanisms}
		Theorem \ref{thm:mechanisms} is implied by Proposition \ref{app:prop:mechanisms}. We state and prove the result using a more general model, in which $V_\lambda(e, \Delta I(c))$ is replaced by function that satisfies strictly increasing differences in $(e, \Delta I(c))$ ($f(e, \Delta I(c))$). In this way, we show that the increasing-differences property is key, not linearity.

		For this section, we fix a function $f\colon [0,1]\times [ 0,1] \to \mathbb R$ that satisfies strictly increasing differences, and such that: $f(\cdot, a)$ is continuous for all $a\in[0,1]$, $f(e, \cdot)$ is nondecreasing for all $e\in[0,1]$, the derivative with respect to the variable $a$, $\frac{\partial f}{\partial a} (e, \cdot)$, exists, is nonnegative, is bounded for all $e\in[0,1]$, and positive for all $e\in(0,1]$. We maintain the notation and definitions in the main text, with the exceptions that the following definitions replace the corresponding ones in the main text: The \emph{value of $I\in\mathcal I$}, given type $(c, \lambda)$ and effort $e$, is $V_\lambda(e, \Delta I(c)):=f(e, \Delta I(c)) -K(e, \lambda)$, and the cost of effort $e\in[0,1]$ is $K(e, \lambda)$, for a continuous function $K(\cdot, \lambda)$. We use the shorthand $t=(c_t, \lambda_t)$, we define the set of optimal effort levels $	E_{\lambda_t} (\Delta I(c_t))\coloneqq \argmax_{e \in [0,1]}V_{\lambda_t}(e, \Delta I(c_t))$, 
		and the corresponding value function $V_{\lambda_t}(\Delta I(c_t)):=\max_{e \in [0,1]}V_{\lambda_t}(e, \Delta I(c_t))$, for $I\in \mathcal I$. A persuasion mechanism  $I_\bullet$ is \emph{incentive compatible} (IC) if, for every $t$, $			t\in \argmax_{r\in R} V_{\lambda_t}(\Delta I_r(c_t)) $.		
		
		\begin{definition}
			An IC persuasion mechanism $I_\bullet$ is \emph{equivalent to an experiment} if there exists an information policy $I$ such that, for all $t\in T$:
			\begin{align*}
				\text{1.}\  & E_{\lambda_t} (\Delta I_t(c_t)) \subseteq  		E_{\lambda_t} (\Delta I(c_t)),\\
				\text{2.}\ & \partial I_t(c_t) \subseteq \partial I(c_t) \  \text{if} \ (0, 1] \cap	E_{\lambda_t} (\Delta I_t(c_t)) \ne \emptyset.
			\end{align*}
		\end{definition}
		
		\begin{proposition}\label{app:prop:mechanisms}
			Every IC persuasion mechanism is equivalent to an experiment.
		\end{proposition}
		\begin{proof}
			Let's fix an IC persuasion mechanism $I_\bullet$. The proof has three steps: (1) we define an information policy $J$, (2) we show that $J$ induces the same effort and (3) action as $I_\bullet$.

			\emph{Step 1: Definition of information policy $J$.}             
			Let's define the function $I\colon[0,1]\to [0,1]$ as
			\begin{align*}
				I(c) := \sup_{r\in R}  I_r(c), \ c\in [0,1].
			\end{align*}
			$I(c)$ is well defined because $0\le I_r(c)\le I_{F_0}(c)\le 1-x_0$, $c\in [0, 1]$. $I$ is the pointwise supremum of a family of convex functions, so $I$ is convex. We have $I_{\overline F}(c)\le I(c) \le I_{F_0}(c)$, $c\in [0,1]$, because $I_r\in\mathcal I, r \in R$. We extend $I$ on $(1, \infty)$, so that the resulting extended function $J\colon \mathbb R_+ \to \mathbb R_+$ is an information policy, by defining $J(c)=I_{F_0}(c)$, $c\in (1, \infty)$, and $J(c)=I(c)$, $c\in[0,1]$. Thus, $J\in \mathcal I$.			
			
			\emph{Step 2: Effort distribution.} There are two cases.
			\begin{enumerate}
				\item[1.]$ E_{\lambda_t} (\Delta I_t(c_t))\cap (0, 1] \ne \emptyset$.
				\item[2.] $ E _{\lambda_t} (\Delta I_t(c_t) ) = \{0\}$.
			\end{enumerate}
			First, we consider case (1.). By the envelope theorem (Lemma \ref{app:lem:ET}), we have: 
			\begin{align*}
				V_{\lambda_t}(a) - V_{\lambda_t}(\Delta I_{t}(c_t)    ) = \int _{\Delta I_{t}(c_t) } ^a \frac{\partial f}{\partial a} (e(\tilde  a), \tilde a)\diff \tilde a,
			\end{align*}
			for a selection $e$ of $E_{\lambda_t}$. Because $f$ exhibits strictly increasing differences, $e(\tilde a) \ge e(\Delta I_{t}(c_t) )$ if $\tilde a \ge \Delta I_{t}(c_t) $.  By the assumption that $\frac{\partial f}{\partial a} (\cdot , \tilde a)>0$ on $(0, 1]$ for all  $\tilde  a$
			\begin{align*}
				V_{\lambda_t}(a ) - V_{\lambda_t}(\Delta I_{t}(c_t) )>0, \ \text{for all} \ a > \Delta I_{t}(c_t).
			\end{align*}
			Thus, in case (1.), IC implies that
			\begin{align*}
				\sup_{r\in R} \Delta I_r (c_t) = \Delta I_t(c_t).
			\end{align*}	
			Let's consider case (2.), and, towards a contradiction, let's suppose $0\notin E_{\lambda_t} (\Delta J(c_t) ) $. By Berge's Maximum Theorem \citep[Theorem 17.31]{aliprantis_infinite_2006}, $E_{\lambda_t}$ is upper hemi-continuous and has compact values. Hence, by the sequential characterization of upper hemi-continuity of compact-valued correspondences \citep[Theorem 17.16]{aliprantis_infinite_2006}, there exists $\overline a\in (\Delta I_t(c_t), \Delta J(c_t))$ and $\eta>0$ such that $ \eta \in E_{\lambda_t} (\overline a) $ (else, define $a_n:=\frac{1}{n}\Delta I_t(c_t) + \left (1-\frac{1}{n} \right )\Delta J(c_t),~n\in \mathbb N$, to get: $a_n\to \Delta J(c_t)$ as $n\to \infty$, $E_{\lambda_t}(a_n)=\{0\},~n\in \mathbb N$, and $0\notin E_{\lambda_t}(\Delta J(c_t))$, which contradicts upper hemi-continuity of $E_{\lambda_t}$). By the assumption that $\frac{\partial f}{\partial a} (\cdot , \tilde a)>0$ on $(0, 1]$ for all  $\tilde  a$
			\begin{align*}
				V_{\lambda_t} (\Delta J(c_t) ) - V_{\lambda_t} (\overline a  )>0.
			\end{align*}
			The above inequality and the envelope theorem imply that
			\begin{align*}
				V_{\lambda_t} (\Delta J(c_t) ) - V_{\lambda_t} (\Delta I_t(c_t) )>0.
			\end{align*}
			Hence, IC does not hold, which is a contradiction.  Thus, $0\in E_{\lambda_t} (\Delta J(c_t) ) $.

			\emph{Step 3:  Action distribution.} Let's suppose that $d\in \partial I_s(c_s)$ and $d\notin  \partial  J(c_s)$ for some type $s\in T$. We have that $d$ is a subgradient of $I_s$ at $c_s$, and $d$ is not a subgradient of $J$ at $c_s$.  From the fact that $J(c_s)=I_s(c_s)$ --- established above ---, there exists $x\in \mathbb R$ such that
			\begin{align*}
				I_s(x) \ge I_s(c_s) +d(x-c_s) >J(x),
			\end{align*}
			which implies $I_s(x)>J(x)$. The last inequality contradicts the definition of $J$. 
		\end{proof}
		
		\subsection{Proofs for Section \ref{sec:application}}
		 \subsubsection{Proof of Theorem \ref{thm:existence}}
		 
		In this section, we maintain the assumption that: the conditional density of the cutoff type given the attention type $\lambda$, $g(\cdot | \lambda)$, is absolutely continuous for all $\lambda$.
		
		The first part of Theorem \ref{thm:existence}, existence of an equilibrium with an upper censorship as Sender's information policy, is implied by Proposition \ref{app:prop:upper}, given that Assumption \ref{ass:1} contains the continuity requirements assumed in this section. Proposition \ref{app:prop:upper} follows mainly from two results: Lemma  \ref{app:lem:uniqueness}, and the known property of upper censorship in Lemma \ref{app:lem:knownzeta}; a version of the property is in the working paper \citet[Appendix A.5]{lipnowski_persuasion_2021}; \citet[Theorem 2]{kolotilin_persuasion_2017} and \citet[Theorem 2]{romanyuk_cream_2019} establish similar results. The second part, same utility across equilibria, is established in Lemma \ref{app:lem:uniqueness}.

		\begin{lemma}\label{app:lem:continuity}
			The function $W$	is continuous on $\mathcal F$.
		\end{lemma}
		\begin{proof}
			See Appendix \ref{app:sec:supp}.
		\end{proof}

		\begin{lemma}\label{app:lem:meas}
			There exists a measurable selection from $(c, \lambda, x)\mapsto \argmax_{b \in [0,1]} b(x-c)$. Furthermore, for every $I$, there exists a measurable selection from $(c, \lambda)\mapsto \argmax_{e \in [0,1]} e\Delta I(c) - \lambda k(e)$.
		\end{lemma}
			\begin{proof}
			The nontrivial part is the second one. The maximand is a real-valued function that is continuous in $c$, $\lambda$, and $e$. So, the result follows from the Measurable Maximum Theorem \citep[Theorem 18.19]{aliprantis_infinite_2006}.
		\end{proof}

		It is without loss of generality to consider equilibria in which the persuasion mechanism is IC. The construction that establishes this result is essentially the same as in the canonical revelation principle (e.g., \citet{myerson_optimal_1982}). However, because Receiver chooses an action after the information policy is selected, the result in its present form is technically not a direct corollary of the results in \citet{myerson_optimal_1982}. The results in Section \ref{app:sec:rev} restrict attention to Sender-optimal tie-breaking of Receiver's action choice, so those results can be stated in the setting of \citet{myerson_optimal_1982}.

		\begin{lemma}[Revelation principle for persuasion mechanisms]\label{app:lem:revelation_direct}
			For every equilibrium $\langle I^*_\bullet, \rho,\eta, \alpha\rangle$, there exists an equilibrium $\langle J^*_\bullet, \rho^*,\eta^*, \alpha^* \rangle$ with the same equilibrium Sender's expected payoff and in which $\rho^*(c, \lambda, J^*_\bullet) = (c, \lambda)$ for all $(c, \lambda) \in T$.
		\end{lemma}
		\begin{proof}
			Let $\langle I^*_\bullet, \rho, \eta, \alpha \rangle$ be an equilibrium. We construct the mechanism $J^*_\bullet \in \mathcal{I}^{|T|}$ by setting $J^*_{(c, \lambda)} = I^*_{\rho(c, \lambda, I^*_\bullet)}$ for all $(c, \lambda) \in T$. 
			
			We define Receiver's strategy $(\rho^*, \eta^*, \alpha^*)$ as follows. Let $\alpha^* = \alpha$. For the mechanism $J^*_\bullet$, let $\rho^*(c, \lambda, J^*_\bullet) = (c, \lambda)$ and $\eta^*(c, \lambda, J^*_\bullet) = \eta(c, \lambda, I^*_\bullet)$ for all $(c, \lambda) \in T$. For any other mechanism $\hat{I}_\bullet \neq J^*_\bullet$, let $\rho^*(c, \lambda, \hat{I}_\bullet) = \rho(c, \lambda, \hat{I}_\bullet)$ and $\eta^*(c, \lambda, \hat{I}_\bullet) = \eta(c, \lambda, \hat{I}_\bullet)$.
			
			We verify that $\langle J^*_\bullet, \rho^*, \eta^*, \alpha^* \rangle$ is an equilibrium. Action Optimality holds because $\alpha^* = \alpha$. 
			
			For Reporting and Effort Optimality, consider the mechanism $J^*_\bullet$ and a type $(c, \lambda) \in T$. Under $I^*_\bullet$, the pair $(\rho(c, \lambda, I^*_\bullet), \eta(c, \lambda, I^*_\bullet))$ maximizes Receiver's expected utility over all $(r, e) \in T \times [0,1]$. By construction, the set of information policies available in $J^*_\bullet$ is a subset of those available in $I^*_\bullet$, i.e., $\{J^*_r\}_{r \in T} \subseteq \{I^*_r\}_{r \in T}$. The maximum expected utility achievable under $J^*_\bullet$ cannot exceed that under $I^*_\bullet$. Since reporting $(c, \lambda)$ and choosing effort $\eta^*(c, \lambda, J^*_\bullet)$ under $J^*_\bullet$ yields the exact same information policy $I^*_{\rho(c, \lambda, I^*_\bullet)}$ and effort $\eta(c, \lambda, I^*_\bullet)$ as the optimal choices under $I^*_\bullet$, this pair is optimal under $J^*_\bullet$. For any $\hat{I}_\bullet \neq J^*_\bullet$, Reporting and Effort Optimality holds because $(\rho^*, \eta^*)$ coincides with the equilibrium strategy $(\rho, \eta)$.
			
			For Sender Rationality, observe that Sender's expected utility under $J^*_\bullet$ given $(\rho^*, \eta^*, \alpha^*)$ is identical to the expected utility under $I^*_\bullet$ given $(\rho, \eta, \alpha)$, because the induced information policies, efforts, and actions are identical for every type. Furthermore, for any deviation $\hat{I}_\bullet \neq J^*_\bullet$, Sender's expected utility is identical to that under the original equilibrium because $(\rho^*, \eta^*)$ coincides with $(\rho, \eta)$. Since $I^*_\bullet$ was optimal for Sender in the original equilibrium, $J^*_\bullet$ is optimal for Sender given $(\rho^*, \eta^*, \alpha^*)$.
		\end{proof}

	The following result establishes that Sender's expected utility given $I$ is the same in every equilibrium adopting a slightly stronger uniqueness condition than in Definition \ref{app:def:equilibriumexperiment}. The comparison holds for two reasons. First, Definition \ref{app:def:equilibriumexperiment} compares Sender's expected utility given the \emph{equilibrium information policy} across equilibria, whereas the proof compares Sender's expected utility given an arbitrary and fixed information policy across equilibria. Second, the proof studies the conditional expected utility given $\lambda$.
	
	\begin{lemma}\label{app:lem:uniqueness}
		The experiment $F$ is an equilibrium experiment if, and only if: $F$ is $W$-maximal. Moreover, there are not multiple Sender's payoffs.
	\end{lemma}
	\begin{proof}
		We first establish that Sender's expected utility in any equilibrium is achieved by a constant mechanism. Let $\langle I^*_\bullet, \rho, \eta, \alpha \rangle$ be an equilibrium. By Lemma \ref{app:lem:revelation_direct}, there exists an equilibrium $\langle J^*_\bullet, \rho^*, \eta^*, \alpha^* \rangle$ with the same Sender's expected payoff, in which $\rho^*(c, \lambda, J^*_\bullet) = (c, \lambda)$ for all $(c, \lambda) \in T$. Truthful reporting implies that $J^*_\bullet$ is an IC persuasion mechanism. By Theorem \ref{thm:mechanisms}, $J^*_\bullet$ is equivalent to an information policy $J \in \mathcal{I}$. By Lemma \ref{lem:policies}, there exists an experiment $\hat F \in \mathcal F$ such that $I_{\hat F} = J$. Let $\hat I_\bullet$ denote the constant mechanism assigning $I_{\hat F}$ to every report. By Definition \ref{def:mechanisms:eq}, for every type $(c, \lambda)$, the effort $\eta^*(c, \lambda, J^*_\bullet)$ maximizes $V_\lambda(e, \Delta I_{\hat F}(c))$, and the action probability induced by $\alpha^*$ remains an optimal selection from $\partial I_{\hat F}(c)$. Consequently, the strategies $(\eta^*, \alpha^*)$ satisfy Reporting and Effort Optimality and Action Optimality given $\hat I_\bullet$. Because the induced effort and action distributions are identical to those under $J^*_\bullet$, Sender's expected utility evaluated at the constant mechanism $\hat I_\bullet$ exactly equals the equilibrium payoff.
		
		Let $F_\bullet$ denote the constant mechanism assigning $I_F$ (for $F \in \mathcal F$) to every report. We now show that $F$ is $W$-maximal if and only if $F_\bullet$ is rational for Sender given $(\rho, \eta, \alpha)$, where $\alpha$ satisfies Action Optimality and $(\rho, \eta)$ satisfies Reporting and Effort Optimality.

		First, observe that Sender can improve upon any persuasion mechanism by choice of a persuasion mechanism that assigns the same information policy to every report, as a consequence of Lemma \ref{app:lem:revelation_direct} and Proposition \ref{app:prop:mechanisms}. Hence, it suffices to show that the function 
		\begin{align*}
			D_\lambda (\cdot, \alpha, \eta) \colon F \mapsto \int_{[0,1]} \int _{[0,1]}\alpha (c, \lambda, x) \diff   (\eta(c, \lambda, F_\bullet) \odot F)(x)\diff \Gcond(c|\lambda)-W_\lambda(F)
		\end{align*}
		is constant on $\mathcal F$ for all $\lambda$. First, let's express Sender's equilibrium-conditional-expected utility given $\lambda$ and constant mechanism $F_\bullet$ as 
		\begin{align*}
			\hat W_\lambda(F_\bullet) := \int_{[0,1]} \int _{[0,1]} \eta_\lambda ^*(\Delta I_F(c)) \left (\alpha (c, \lambda, x)- \alpha (c, \lambda, x_0) \right ) \diff F(x)\diff \Gcond(c|\lambda) \\
			+\int _{[0,1]}  \alpha (c, \lambda, x_0) \diff \Gcond(c|\lambda),
		\end{align*}
		for a selection $\eta_\lambda^*$ from $a\mapsto \argmax_{e \in [0,1]}V_\lambda (e, a)$, via Remark \ref{app:rem:extensive}. By Lemma \ref{lem:receiver}, there exists a selection $d_I^1$ from the subdifferential of $\Delta I_F $ on $[0, x_0]$ and a selection $d_I^2$ from the subdifferential of $\Delta I_F $ on $(x_0, 1]$ such that:
		\begin{align*}
			- ( \hat W_\lambda(F_\bullet) - \hat W_\lambda(\overline F_\bullet) )= \int_{[0,x_0]} \eta_\lambda ^*(\Delta I_F(c)) d_I^1(c) \diff \Gcond(c|\lambda) + \int_{(x_0, 1]} \eta_\lambda ^*(\Delta I_F(c))d_I^2(c)\diff \Gcond(c|\lambda)
		\end{align*}
		By the envelope theorem (Lemma \ref{app:lem:ET}), $\eta_\lambda ^*$ is a selection from the subdifferential of the convex and nondecreasing function $V_\lambda $. By construction, $\Delta I_F$ is: (i) convex on $[0, x_0]$, and (ii) convex on $(x_0, 1]$. Hence, by the rules of subdifferential calculus (Fact \ref{app:fact:convex}), there exists a selection $d$ from the subdifferential of $V_\lambda \circ \Delta I_F$ such that $d(c)= \eta_\lambda ^*(\Delta I_F(c)) d_I^1(c)$ for all $c\in[0, x_0]$, and $d(c) = \eta_\lambda ^*(\Delta I_F(c)) d_I^2(c)$ for all $c\in( x_0, 1]$. Hence:
		\begin{align*}
			- ( \hat W_\lambda(F_\bullet) - \hat W_\lambda(\overline F_\bullet) ) & = \int_{[0,x_0]} d(c) \diff \Gcond(c|\lambda)+ \int_{(x_0, 1]} d(c) \diff \Gcond(c|\lambda) \\
			& = \int_{[0,x_0]} d(c) \diff \Gcond(c|\lambda)+ \int_{[x_0, 1]} d(c) \diff \Gcond(c|\lambda),
		\end{align*}
		in which the second equality uses absolute continuity of $\Gcond(\cdot |\lambda)$.
		By Fact \ref{app:fact:convex}, the composition $V_\lambda \circ \Delta I_F$ is a convex function on $[0,x_0]$, so $V_\lambda \circ \Delta I_F$ is the integral of any selection from its subdifferential on $[0, x_0]$ \citep[Corollary 24.2.1.]{rockafellar_convex_1970}. Similarly, $V_\lambda \circ \Delta I_F$ is a convex function on $[x_0, 1]$. By absolute continuity of $\gcond(\cdot |\lambda)$, we integrate by parts to obtain
		\begin{align*}
			- ( \hat W_\lambda(F_\bullet) - \hat W_\lambda(\overline F_\bullet) ) = V_\lambda \circ \Delta I_F (1)\gcond(1|\lambda) - V_\lambda \circ \Delta I_F (0)\gcond(0|\lambda) \\ - \int_{[0,1]} V_\lambda \circ \Delta I_F (c) \frac{\partial \gcond}{\partial c}(c|\lambda) \diff c.
		\end{align*}
		The fact that $\Delta I_F(1)=\Delta I_F(0)=0$ implies
		\begin{align*}
			- ( \hat W_\lambda(F_\bullet) - \hat W_\lambda(\overline F_\bullet) ) =\left ( \gcond(1|\lambda) -\gcond(0|\lambda) \right ) V_\lambda (0) - \int_{[0,1]} V_\lambda \circ \Delta I_F (c) \frac{\partial \gcond}{\partial c}(c|\lambda) \diff c.
		\end{align*}
		Hence, we have
		\begin{align*}
			\hat W_\lambda(F_\bullet) = W_\lambda (F)+ \hat W_\lambda(\overline F_\bullet) - ( \gcond(1|\lambda) -\gcond(0|\lambda) ) V_\lambda (0).
		\end{align*}
		Therefore, we have
		\begin{align*}
			D_\lambda (F, \alpha, \eta) = \int _{[0,1]} \alpha (c, \lambda, x_0) \diff \Gcond(c|\lambda) - ( \gcond(1|\lambda) -\gcond(0|\lambda) ) V_\lambda (0).
		\end{align*}
		As a result, $D_\lambda (\cdot, \alpha, \eta) $ is constant on $\mathcal F$. We conclude that $F$ is $W$-maximal if and only if $F_\bullet$ is rational for Sender given $(\rho, \eta, \alpha)$, $\alpha$ satisfies Action Optimality, and $(\rho, \eta)$ satisfies Reporting and Effort Optimality.
		
		From the above equivalence, it follows that if $\langle \hat I^*_\bullet, \rho, \eta, \alpha \rangle$ is an equilibrium such that $\hat I^*_r = I_{\hat F}$ for all $r \in T$, then $\hat F$ is $W$-maximal. For the other direction, let $F$ be $W$-maximal. By Lemma \ref{app:lem:meas}, there exist $\rho, \eta$ and $\alpha$ that satisfy the equilibrium measurability conditions, Action Optimality, and Reporting and Effort Optimality, given the constant mechanism $F_\bullet$. Because $F$ is $W$-maximal, $F_\bullet$ is rational for Sender given $(\rho, \eta, \alpha)$ by the above equivalence. Thus, $\langle F_\bullet, \rho, \eta, \alpha \rangle$ is an equilibrium. As an implication, there are not multiple Sender's payoffs, because the expected utility of Sender in any equilibrium equals the expected utility of a constant mechanism, which is uniquely determined by the maximum of $W$ on $\mathcal{F}$.
	\end{proof}

		\begin{proposition} \label{app:prop:existence}
			There exists an equilibrium.
		\end{proposition}

	\begin{proof}
		The set $\mathcal F$ is compact in the $L^1$ topology
		\citep[Proposition 1]{kleiner_extreme_2021}. By Lemma
		\ref{app:lem:continuity}, $W$ is continuous on $\mathcal F$. Hence there
		exists $F^*\in\mathcal F$ such that
		\[
		W(F^*)=\max_{F\in\mathcal F}W(F).
		\]
		By Lemma \ref{app:lem:uniqueness}, $F^*$ is an equilibrium experiment.
		Therefore, the constant persuasion mechanism assigning $I_{F^*}$ to every
		report, together with measurable optimal effort and action selections,
		constitutes an equilibrium.
	\end{proof}
		\begin{lemma} \label{app:lem:knownzeta}
			Let $I\in \mathcal I$ and $\zeta \in [0,1]$. There exists $\theta \in[0,\zeta]$ such that:
			\begin{enumerate}
				\item[(1.)] $(I^\theta )'(\zeta) = I(\zeta)$;
				\item[(2.)] $I^{\theta}(\zeta^-) \leq  I'(\zeta^-)$, and
				\begin{align*}
					I^\theta(x) - I(x) \geq 0 &, \ \text{for all} \  x\in [0,\zeta],\\
					I^\theta (x) - I(x) \leq 0 &, \ \text{for all} \  x\in [\zeta,\infty).
				\end{align*}. 
			\end{enumerate}
		\end{lemma}

			\begin{proof}
				See Appendix \ref{app:sec:supp}.
			\end{proof}

		\begin{proposition}\label{app:prop:upper}
			Under Assumption \ref{ass:1}, there exists an equilibrium in which Sender's information policy is an upper censorship.
		\end{proposition}
		\begin{proof}[Proof of Proposition \ref{app:prop:upper}]
				By Lemma \ref{app:lem:uniqueness}, if $F^*\in\mathcal F$ maximizes
				\begin{align*}
					W\colon F \mapsto	\int _{[0,1]}\int _{[0,p]} V_\lambda (\Delta I_{ F} (c)) \frac{\partial \gcond}{\partial c}(c|\lambda)\diff c 
					+
					\int _{[p, 1]} V_\lambda (\Delta I_{ F} (c)) \frac{\partial \gcond}{\partial c}(c|\lambda)\diff c \diff G (\lambda),
				\end{align*}
				then there exists an equilibrium in which $F^*$ is Sender's experiment. Suppose two experiments $F, H\in \mathcal F$ such that $I_F(x) \ge I_H(x)$ for all $x\in [0,p]$ and $I_F(x) \le I_H(x)$ for all $x\in [p, 1]$. Because (i) $V_\lambda$ is nondecreasing, (ii) $\frac{\partial \gcond}{\partial c}(\cdot |\lambda)$ is nonnegative on $[0,p]$ and nonpositive on $[p,1]$, it follows that $W(F)\ge W(H)$. Hence, the result follows from Lemma \ref{app:lem:knownzeta}. In particular, by Proposition \ref{app:prop:existence}, there exists an equilibrium experiment $\hat F$, and by  Lemma \ref{app:lem:knownzeta} there exists $F^*$ such that $F^*$ weakly improves upon $\hat F$ in terms of $W$ and $I_{F^*}$ is an upper censorship.
	\end{proof}


	\subsubsection{Proof of Proposition \ref{prop:cost}}
	
	\begin{proof}[Proof of Proposition \ref{prop:cost}]
		The proof has four steps. First, we establish a single-crossing
		property of the integral that determines the sign of the derivative
		of Sender's payoff. Second, we establish a monotonicity property
		of this integral in the attention cost, given certain conditions.
		The third step verifies that these conditions hold in a neighborhood
		of $\theta_0$, and that every optimum belongs to this
		neighborhood for sufficiently small attention costs. The final step
		completes the argument. Let $g$ denote the marginal density of the outside option type, let $f_0$ denote the density of $F_0$, and $
			x_\theta
			\coloneqq \frac{\int_\theta^1 t \diff F_0(t)}{1-F_0(\theta)}$.
		
		\emph{Step 1.}
		We claim that
		\begin{align*}
			\int_\theta^\zeta (c-\theta)g'(c)\diff c\leq0
			\quad\text{only if}\quad
			\int_{\theta'}^{\zeta'}(c-\theta')g'(c)\diff c<0,
		\end{align*}
		for all $\theta\leq\theta'$ and $\zeta\leq\zeta'$, with
		$\theta<\zeta$, $\theta'<\zeta'$, and
		$(\theta',\zeta')\ne(\theta,\zeta)$.
		
		If $p\leq\theta'$, the conclusion follows from monotonicity of
		$g$ on $[p,1]$. Suppose that $\theta'<p$.	If the first integral ($	\int_\theta^\zeta (c-\theta)g'(c)\diff c$) is nonpositive, then $p<\zeta$.
		By integration by parts,
		\begin{align*}
			\int_\theta^\zeta(c-\theta)g'(c)\diff c
			=(\zeta-\theta)g(\zeta)-\int_\theta^\zeta g(c)\diff c.
		\end{align*}
		The derivative of this expression with respect to $\theta$ is
		$g(\theta)-g(\zeta)$. This derivative is increasing on
		$[0,p]$ and positive on $[p,\zeta)$. Thus, the function $\theta\mapsto 	\int_\theta^\zeta (c-\theta)g'(c)\diff c$ is increasing or first decreases and then increases to
		$0$ at $\theta=\zeta$. Consequently, if it is nonpositive at
		$\theta$, it is negative at every $\theta'\in(\theta,\zeta)$. Moreover, for $\zeta'>\zeta>p$ and $\theta'<\zeta$,
		\begin{align*}
			\int_{\theta'}^{\zeta'}(c-\theta')g'(c)\diff c
			-\int_{\theta'}^\zeta(c-\theta')g'(c)\diff c
			=\int_\zeta^{\zeta'}(c-\theta')g'(c)\diff c<0,
		\end{align*}
		by the properties of $g$.
		
		\emph{Step 2.}
		For fixed $\theta$, the integral
		\begin{align*}
			\int_\theta^{\overline c_{\lambda^*}(\Delta I_\theta)}
			(c-\theta)g'(c)\diff c
		\end{align*}
		is strictly increasing in $\lambda^*$ on any interval $\Lambda$ such that, for all $\lambda\in\Lambda$, we have
		\begin{align*}
			\overline c_{\lambda}(\Delta I_\theta)>
			\max\{\theta,x_0\},
			\qquad
			\overline c_{\lambda}(\Delta I_\theta)\geq p.
		\end{align*}
		Indeed, $\lambda\mapsto 
			\overline c_{\lambda}(\Delta I_\theta)
			=x_\theta-\frac{\lambda}{1-F_0(\theta)}$ 
		is  decreasing on $\Lambda$.
		
		\emph{Step 3.}
		We claim that there exists a closed interval $J$ containing
		$\theta_0$ in its interior such that, for all sufficiently small
		$\lambda^*>0$, every optimal upper censorship has its cutoff
		in the interior of $J$ and, for every $\theta\in J$,
		\begin{align*}
			\underline c_{\lambda^*}(\Delta I_\theta)<\theta,
			\qquad
			\overline c_{\lambda^*}(\Delta I_\theta)>
			\max\{\theta,x_0,p\}.
		\end{align*}
		The costless optimum $\theta_0$ is unique
		\citep[Lemma 7]{kolotilin_censorship_2022}, and its characterization (Step 4, or the results in \citealt{kolotilin_censorship_2022}) implies that 
		$\theta_0<p<x_{\theta_0}$. Choose a closed interval $J\subset(0,p)$ containing $\theta_0$
		in its interior such that $x_\theta>p$ on $J$. Full support of $F_0$ implies
		$\Delta I_\theta(\theta)>0$ and $x_\theta>\max\{\theta,x_0\}$. For all sufficiently small
		$\lambda^*>0$, for all $\theta\in J$,
		\begin{align*}
			\Delta I_\theta(\theta)>\lambda^*,
			\qquad
			x_\theta-\frac{\lambda^*}{1-F_0(\theta)}
			>\max\{\theta,x_0,p\}.
		\end{align*}
		The first inequality gives
		$\underline c_{\lambda^*}(\Delta I_\theta)<\theta$, and, because $\Delta I_\theta(c)=(1-F_0(\theta))(x_\theta-c)$
		on $[\max\{\theta,x_0\},x_\theta]$, the second inequality establishes the desired result. 
		
		Finally, Lemma \ref{app:lem:uniqueness} implies
		\begin{align*}
			\sup_{\theta\in[0,1]}
			\left|W_{\lambda^*}(I_\theta')-W_0(I_\theta')\right|
			\leq\lambda^*\int_0^1|g'(c)|\diff c.
		\end{align*}
		The function $\theta\mapsto W_0(I_\theta')$ has a unique maximizer
		$\theta_0$, so its value there exceeds its maximum
		outside the interior of $J$. Let
		\begin{align*}
			\delta:=W_0(I_{\theta_0}')
			-\max_{\theta\in[0,1]\setminus\operatorname{int}J}
			W_0(I_\theta')>0.
		\end{align*}
		For $\lambda^*\int_0^1|g'(c)|\diff c<\delta/2$, the above inequality implies, for every $\theta\notin\operatorname{int}J$,
		\begin{align*}
			W_{\lambda^*}(I_{\theta_0}')-W_{\lambda^*}(I_\theta')
			\geq\delta-2\lambda^*\int_0^1|g'(c)|\diff c>0.
		\end{align*}
		Hence every optimal cutoff belongs to the interior of $J$.
		
		\emph{Step 4.}
		For each $c$, the function $\theta\mapsto I_\theta(c)$ is
		absolutely continuous, with derivative $
			\frac{\partial I_\theta(c)}{\partial\theta}
			=f_0(\theta)(c-\theta)
			\mathbf{1}_{\{\theta<c<x_\theta\}}(c)$, 
		for almost every $(\theta,c)$.
		Its absolute value is bounded by $f_0(\theta)$.
		Because $g'$ is integrable, the function $\theta\mapsto W_{\lambda^*}(I_\theta')$ is
		absolutely continuous, and
		\begin{align*}
			\frac{d}{d\theta}W_{\lambda^*}(I_\theta')
			=f_0(\theta)\int_\theta^{x_\theta}
			(c-\theta)\mathbf{1}_{\{\Delta I_\theta(c)>\lambda^*\}}
			g'(c)\diff c
		\end{align*}
		for almost every $\theta$. For $\lambda^*=0$, this formula gives a characterization of $\theta_0$ as the unique value of $\theta$ such that $	\frac{d}{d\theta}W_{\lambda^*}(I_\theta')=0$.
		For positive costs sufficiently small that Step 3 applies,
		it gives, for almost every $\theta\in J$,
		\begin{align*}
			\frac{d}{d\theta}W_{\lambda^*}(I_\theta')
			=f_0(\theta)
			\int_\theta^{\overline c_{\lambda^*}(\Delta I_\theta)}
			(c-\theta)g'(c)\diff c.
		\end{align*}
		
		On $J$, the function
		$\theta\mapsto\overline c_{\lambda^*}(\Delta I_\theta)$ is
		absolutely continuous, with
		\begin{align*}
			\frac{d}{d\theta}
			\overline c_{\lambda^*}(\Delta I_\theta)
			=\frac{f_0(\theta)}{1-F_0(\theta)}
			\left(\overline c_{\lambda^*}(\Delta I_\theta)-\theta\right)
			\geq0
		\end{align*}
		almost everywhere. Hence it is nondecreasing.
		Step 1 therefore implies that the integral determining the
		sign of Sender's derivative crosses zero at most once and
		from above on $J$.
		
		The objective is continuous on $[0,1]$, so a maximizer exists,
		and every maximizer belongs to the interior of $J$ by Step 3.
		Step 1 implies uniqueness of the maximizing cutoff,
		which we denote by $\theta_{\lambda^*}$. At $\theta_0$, Step 3 and positivity of $\lambda^*$ give
		\begin{align*}
			p<\overline c_{\lambda^*}(\Delta I_{\theta_0})
			<x_{\theta_0}.
		\end{align*}
		Hence, Step 2 implies
		\begin{align*}
			\int_{\theta_0}^{\overline c_{\lambda^*}
				(\Delta I_{\theta_0})}
			(c-\theta_0)g'(c)\diff c>0.
		\end{align*}
		The single-crossing property and optimality of  
		$\theta_{\lambda^*}$ imply $\theta_{\lambda^*}>\theta_0$.
		By Theorem \ref{thm:existence}, the maximizing upper censorship
		is optimal among all information policies.
	\end{proof}
	
		\subsubsection{Proof of Proposition \ref{prop:reverse}}
		
		Proposition \ref{prop:reverse} is implied by Proposition \ref{app:prop:RHR}. We assume that $k$ is linear. We assume that $F_0$ and $\Gmarg$ are continuously differentiable, the reverse hazard rate is nonincreasing.
		
		\begin{lemma}\label{app:lem:RHR_lemma}
			Let the outside option put full mass at $\zeta\in (x_0, 1)$, the distributions of $\lambda$ and $\theta$ admit a continuous PDF, the support of $\lambda $ include $[0, \Delta I_1(\zeta)]$, and the reverse hazard rate $\gmarg (\lambda)/\Gmarg(\lambda)$ of the attention cost be nonincreasing in $\lambda$. There exists a unique $\theta$ such that $\theta<\zeta<\overline x_\theta$ and
			\begin{align*}
				(1-F_0( \theta)) (\zeta - \theta) = \frac{\Gmarg (\Delta I^{\theta}(\zeta ))}{\gmarg(\Delta I^{\theta}(\zeta))}.
			\end{align*}
		\end{lemma}
		
		\begin{proof}
			Let the outside option put full mass at $\zeta\in (x_0, 1)$. We note that the state $\theta$ is such that $\theta\le \zeta\le \overline x_\theta$  if and only if $\overline x_\theta\in [\zeta, \overline x_\zeta]$. In addition, because $\theta \mapsto \overline x_\theta$ is increasing, there exists $a\in [0, \zeta]$ such that:  $\theta\le \zeta\le \overline x_\theta$ if and only if $\theta\in[a, \zeta]$. In the rest of this proof, we show that the result holds by the intermediate-value theorem.
			
			We consider the function $v\colon \theta\mapsto 	(1-F_0( \theta)) (\zeta - \theta) - \frac{\Gmarg (\Delta I^{\theta}(\zeta ))}{\gmarg(\Delta I^{\theta}(\zeta))}$. It holds that $v$ is decreasing and continuous on $[a, \zeta]$, because: if $\theta<\zeta<\overline x_\theta$, then $\Delta I^\theta(\zeta )$ is increasing in $\theta$. Moreover, $v(\zeta)=- \frac{\Gmarg (\Delta I^{\theta}(\zeta ))}{\gmarg(\Delta I^{\theta}(\zeta))}<0$, because $\theta>\zeta>x_0>0$.
			
			It suffices to show that $v(a)> 0$. We note that $\theta=a$ if and only if $\overline x_\theta = \zeta $. Hence, $\Delta I_a(\zeta)=0$. We conclude that $v(a) =	(1-F_0( a)) (\zeta - a) >0 $, because $a>0$ by $\zeta > x_0$.
		\end{proof}
		\begin{proposition}\label{app:prop:RHR}
			Let the outside option put full mass at $\zeta\in (x_0, 1)$, the distributions of $\lambda$ and $\theta$ admit a continuous PDF, the support of $\lambda $ include $[0, \Delta I_1(\zeta)]$, the reverse hazard rate $\gmarg (\lambda)/\Gmarg(\lambda)$ of the attention cost be nonincreasing in $\lambda$, and let $\overline \theta $ be the unique state $\theta$ such that: $\theta<\zeta<\overline x_\theta$, and
			\begin{align*}
				(1-F_0( \theta)) (\zeta - \theta) = \frac{\Gmarg (\Delta I^{\theta}(\zeta ))}{\gmarg(\Delta I^{\theta}(\zeta))}.
			\end{align*}
			The following conditions hold:
			\begin{enumerate}
				\item There exists an equilibrium in which the information policy is a $\overline \theta $ upper censorship. 
				\item If there exist $\theta\in\Theta$ and an equilibrium in which the information policy is the $\theta $ upper censorship, then $\theta=\overline \theta$.
			\end{enumerate}
		\end{proposition}
		\begin{proof}
			The first part is implied by the following two steps.
			
			\emph{Step 1.} By Lemma \ref{app:lem:RHR_lemma}, there exists a state $ \theta^\star  \in\Theta $ such that: there exists an equilibrium in which the information policy is a $\theta^\star $ upper censorship. By Lemma \ref{app:lem:uniqueness}, the state $\theta^\star $ solves $	\max _{\theta\in\Theta } v(\theta)$, letting $v\colon \theta\mapsto \left ( 1- I'_\theta (\zeta^-)\right) \Gmarg (\Delta I^\theta (\zeta))$. Note that every state $\theta$ such that $\theta\ge \zeta$ induces $v(\theta)=v(\zeta) = v(1)=(1-F_0(\zeta))\Gmarg(\Delta I_1(\zeta))$. Moreover every state $\theta$ such that $\zeta\ge \overline x_\theta$ induces $v(\theta)=v(0)=0$. Hence, we have that, for all $\theta\in [0, a]\cup [\zeta, 1]$, $v(\zeta )\ge v(\theta)$.
			
			\emph{Step 2.} As shown in the proof of Lemma \ref{app:lem:RHR_lemma}, there exists $a\in (0, \zeta)$ such that: the state $\theta$ is such that $\theta\le \zeta\le \overline x_\theta$  if and only if  $\theta\in[a, \zeta]$. Moreover, if $\theta< \zeta< \overline x_\theta$, then $I^\theta(\zeta) = I_1(\theta) + F_0(\theta)(\zeta  -\theta)$. Hence,  if $\theta< \zeta< \overline x_\theta$, then by the implicit function theorem $v$ is continuously differentiable at $\theta$, and
			\begin{align*}
				v'(\theta) = f_0(\theta) ((1-F_0(\theta))(\zeta-\theta)\gmarg(\Delta I^\theta(\zeta))-\Gmarg(\Delta I^\theta(\zeta))).
			\end{align*}
			Hence, we have that: for all $\theta \in (a, \zeta) $, $v(\overline \theta) > v(\theta)$.

			\emph{Step 3.} By the observations in Step 1, we have that $v(\zeta) = \lim _{\varepsilon\to 0^+}v(\zeta+\varepsilon)$. For $\varepsilon \in (0,  \theta -a )$, it holds that:
			\begin{align*}
				v(\theta-\varepsilon) =(1-F_0(\theta-\varepsilon))\Gmarg(I_1(\theta-\varepsilon) + F_0(\theta-\varepsilon)(\zeta - \theta+\varepsilon) - (\zeta-x_0)_+).
			\end{align*}
			Hence, $v$ is continuous at $\theta$. Therefore, $v(\overline \theta)> v(\zeta)$.	We conclude that $\theta^\star=\overline \theta$.
		\end{proof}

		\subsubsection{Proof of Proposition \ref{prop:unique}} 
		The proof of Proposition \ref{prop:unique} has two steps. The first and main step has the same structure as that of Proposition \ref{app:prop:upper}. In particular, Lemma \ref{app:lem:twosidedpart1} generalizes the construction of Lemma \ref{app:lem:knownzeta} to construct: an information policy $I^\star$ that preserves the extensive margin and improves upon an arbitrary information policy $I$, for large $p$. $I^\star$ induces two censorship regions, separated by a full-revelation region. The second step of the proof: (1) adds a second censorship region at the top to include the general case of $p>x_0$, and (2) verifies that eliminating the bottom censorship region improves upon Sender's payoff. 	For the rest of this proof, we omit reference to $\lambda$ and fix an equilibrium $\langle I^*_\bullet, \rho, \eta, \alpha \rangle$.
		\begin{lemma}\label{app:lem:twosidedpart1}
			Let $I \in \mathcal I$ and define $c^*:= \overline c(\Delta I)$. There exists an information policy $I^\star$ that satisfies the following properties:
			\begin{enumerate}
				\item (FEAS) $I^\star$ is feasible, i.e., $I^\star \in \mathcal I$;
				\item (EM) $I^\star$ produces the same extensive margin as $I$, i.e., $\overline c(\Delta I^\star)= c^*$ and $\underline c(\Delta I^\star)= \underline c(\Delta I) $;
				\item (IMPR) $	\Delta I^\star (x)  \geq 0$, for all $ x\in [\underline c(\Delta I),c^*]$;
				\item (CENS)  There exist $x_\ell, \theta_\ell, \theta_m, x_m$ such that $0\le x_\ell \le  \theta_\ell \le  \theta_m \le  x_m \le 1$ and
				\begin{align*}
					I^\star (x) = \begin{cases}
						I_{\overline F}(x ) &, x \in [0, x_\ell] \\
						I_{F_0} (\theta_\ell) + F_0(\theta_\ell)(x-\theta_\ell) &, x \in ( x_\ell, \theta_\ell] \\
						I_{F_0}(x) &, x \in ( \theta_\ell, \theta_m]  \\
						I_{F_0} (\theta_m) + F_0(\theta_m)(x-\theta_m) &, x \in (\theta_m, x_m] \\
						I_{\overline F} (x)  &, x \in (x_m, \infty ).
					\end{cases}
				\end{align*}
			\end{enumerate}
		\end{lemma}
		\begin{proof}
			We use the notation: $\overline c(\Delta I)=: \overline c $, $\underline c(\Delta I)=: \underline c$.
			In the first step, we prove the result for the case in which there is a feasible information policy that is a straight line between the points $\underline p := (\underline c, I(\underline c))$ and $\overline p:= (\overline c, I(\overline c))$. In the second step we analyze the other case.
			
			\emph{First Step}. Let's define the line $i$ such that $x\mapsto I(\underline c) + \lambda^\ast (x -\underline c)$, with slope $\lambda^* := \frac{I(\overline c) - I(\underline c)}{\overline c - \underline c}$. We claim that $i^\star (x):= \max \{i(x), I_{\overline F}(x)\}$ satisfies all properties. $i^\star$ is FEAS by hypothesis. $i^\star$ is EM because $i(\underline c) = I(\underline c)$ and $i(\overline c) = I(\overline c)$. $i^\star$ is IMPR because $I$ is convex and $i^\star$ is EM. $i^\star$ is CENS with $\theta_\ell=\theta_m=x_m$, because: (i) EM of $i^\star$ and convexity of $I$ imply that $i^\star$ is affine on  $[\underline  c, \overline c]$, (ii) $\lambda ^* \in [0,1]$ and EM imply, with $I\in \mathcal I$, that there are intersection points $\widetilde x_1, \widetilde x_2$, with $\widetilde x_1\leq \underline c \leq \overline c\leq \widetilde x_2$, such that: $i^\star(x)=I(x)$ if $x\in [0, \widetilde x_1]\cup [\widetilde x_2, 1]$.
			
			\emph{Second Step}. In this case, $i^\star$ is not FEAS. Because $i^\star$ satisfies FEAS at $x$ if $x\leq \underline c$ and if $x\geq\overline c$, there exists a point $x^*\in (\underline c, \overline c)$ such that: $i(x^*) > I_{F_0}(x^*)$. Let's define:
			\begin{align*}
				L &: = \{ \lambda \in [ I'(\underline c), 1] : I(\underline c) + \lambda (x -\underline c)
				\leq I_{F_0}(x) \text{ for all } x \in [\underline c, \infty) \}, \\
				M &:=  \{ \lambda \in [0, I'(\overline c)] : I(\overline c) + \lambda (x -\overline  c) \leq  I_{F_0}(x) \ \text{for all } x \in [0, \bar c]\},
			\end{align*}
			$\ell := \max L$, $m := \min M$, and the lines
			\begin{align*}
				y_\ell : x &\mapsto I (\underline c) + \ell (x - \underline c), \\
				y_m : x& \mapsto I ( \bar c)+ m (x - \bar c). 	
			\end{align*}
			
			As part of the rest of the proof, we establish some lemmata.
			\begin{lemma} It holds that $\ell$ and $m$ are well-defined.
			\end{lemma}
			\begin{proof}
				$L$ is nonempty because $I'(\underline c)\in L$, which follows from: (i) $I_{F_0}(x)\geq  I(x)$ for all $x$ and (ii) $I'(\underline c) \in \partial I(\underline c)$. $M$ is nonempty because $I'(\overline c)\in M$, which follows from: (i) $I_{F_0}(x)\geq   I(x)$ for all $x$ and (ii) $I'(\overline c) \in \partial I(\overline  c)$. $L, M$ are closed because $I_{F_0}$ is continuous. $L, M$ are bounded.
			\end{proof}
			\begin{lemma}
				There exists a unique pair of numbers $( \theta_\ell,  \theta_m) \in [\underline c, 1]\times [0, \bar c]$ such that: $y_\ell (\theta_\ell) = I_{F_0} (\theta_\ell ) $ and $ y_m (\theta_m) = I_{F_0} (\theta_m )$.
			\end{lemma}
			\begin{proof}
				Suppose there does not exist such $\theta_\ell$. There exists a sufficiently small $\varepsilon > 0$ such that: (i) $\ell + \varepsilon \in L$ and (ii) $I(\underline c ) + (\ell + \varepsilon) (x-\underline c) < I_{F_0} (x) $ for all $x\in [\underline c, \infty )$; we note that $\theta_\ell =1$ contradicts $\ell \in L$ because $I_{F_0}'(x) <1 $ if $x<1$. Uniqueness of $\theta_\ell $ follows from convexity of $I_{F_0}$.
				
				Suppose there does not exist such $\theta_m$. There exists a sufficiently small $\varepsilon > 0$ such that: (i) $m - \varepsilon \in M$ and (ii) $I(\bar c ) + (m - \varepsilon) (x-\bar c) < I_{F_0} (x) $ for all $x\in [0, \bar c)$; we note that $\theta_m =0$ contradicts $I\ne I_{\overline F}$. Uniqueness of $\theta_m $ follows from convexity of $I_{F_0}$.
			\end{proof}
			\begin{lemma}
				It holds that $\theta_\ell \leq \theta_m$.
			\end{lemma}
			\begin{proof}
				Let's first prove that: it suffices to show that $\ell \leq m$. Suppose $\ell \leq m$, then, from $\ell \in \partial  I_{F_0}(\theta_\ell)$, $m \in \partial I_{F_0}(\theta_m)$, and $I_{F_0}$ being strictly convex, we have: $\theta_\ell \leq \theta_m$.
				
				Next, we show that $\ell \leq \lambda^*$. Suppose that: $\ell > \lambda^*$.   Then: $I(x)  + \ell (x - \underline c) > I(\underline c) + \lambda^* (x -\underline c)$ for all $x > \underline c$. Therefore, because $\ell >0$, we get:
				\begin{align*}
					I_{F_0}(x^*) \geq I(\underline c) + \lambda^* (x^*- \underline c).
				\end{align*}
				We reach a contradiction with the definition of $x^*$, so: $\ell \leq \lambda^*$.
				
				Let's prove that $m \geq \lambda^*$. Suppose $m < \lambda^*$.  Then: $I(x)  + m (x - \overline c) > I(\overline  c) + \lambda^* (x -\overline  c)$ for all $x < \overline c$. Therefore, because $m >0$, we get:
				\begin{align*}
					I_{F_0}(x^*) \geq I(\underline c) + \lambda^* (x^* - \underline c).
				\end{align*}
				We reach a contradiction with the definition of $x^*$, so: $m \geq \lambda^*$. Therefore, we have $m \geq \lambda ^* \geq \ell$, which implies $\theta_m\geq \theta_\ell$.
			\end{proof} 
			We define a candidate $I^\star$ and verify that $I^\star$ has the desired properties.
			\begin{align*}
				I^\star (x):= \begin{cases}
					\max \{ I_{\overline F}(x ), I(\underline c) + \ell (x -\underline c) \} &, x \in [0, \theta_\ell]\\
					I_{F_0}(x) &, x \in [\theta_\ell, \theta_m]\\
					\max \{ I_{\overline F} (x), I(\overline c) + m (x -\overline c) \} &, x \in [\theta_m, \infty)\\
				\end{cases}
			\end{align*}
			Let's first verify that $I^\star$ is well-defined. Recall that $\ell \in \partial I_{F_0}(\theta_\ell)$ and $m \in \partial I_{F_0}(\theta_m)$. Because $I(\underline c) + \ell (0 -\underline c)<I_{\overline F} (0)$ and $I(\underline c)\geq I_{\overline F} (\underline c)$, $\max \{ I_{\overline F} (x ), I(\underline c) + \ell (x -\underline c) \} =I_{\overline F} (x )$ if $x < y_0$; and $ \max \{ I_{\overline F} (x ), I(\underline c) + \ell (x -\underline c) \} = I(\underline c) + \ell (x -\underline c)$ if $x > y_0$; for some $y_0\in [0, \theta_\ell]$. Similarly,  there exists $y_2\in [\theta_m, 1]$ such that: $\max \{ I_{\overline F}(x ), I(\overline c) + m (x -\overline c) \} = I_{\overline F} (x )$ if $x > y_2$, and $ \max \{ I_{\overline F} (x ), I(\overline c) + m (x -\overline c) \} = I(\overline c) + m (x -\overline c)$ if $x < y_2$.
			\begin{enumerate}
				\item CENS follows from the definition of $I^\star$ and the conclusion of the above paragraph.
				\item IMPR on $[\underline c, \theta_\ell]$  and $[\theta_m, \overline c]$ follows from convexity of $I$, and on $[\theta_\ell, \theta_m]$ follows from FEAS of $I$ in that region.
				\item EM follows by construction of $I^\star$.
				\item FEAS is established as in the last step of the proof of Lemma \ref{app:lem:knownzeta}.
			\end{enumerate}
		\end{proof}

		\begin{proof}[Proof of Proposition \ref{prop:unique}]	
			Let's define information policy $J$ by: letting $J$ equal $I^\star$, constructed as in Lemma \ref{app:lem:twosidedpart1} by replacing $ c^*$ with $p$, for $x\in[0,x_m^\circ]$, defining the point $x_m^\circ$ in which $I^\star$ intersects the line $j\colon x\mapsto I(\overline c) + I'(\overline c) (x-\overline c) $; and letting $J$ equal $x\mapsto \max \{ I_ {\overline F}(x), j(x)\}$ on $[x_m^\circ, \infty)$.
			
			It suffices to show that: if the resulting information policy $J$ induces a censorship region at the bottom, then there is an improvement over $J$ that is a bi-upper censorship. Suppose that $I^\star$ is affine on $[x_\ell, \theta_\ell]$ and $I^\star$ equals $I_{\overline F}$ on $[0, x_\ell]$, for $0<x_\ell< \theta_\ell$ (for notation, see Lemma \ref{app:lem:twosidedpart1}). By construction, $I^\star (\theta_\ell)= I_{F_0}(\theta_\ell)$. Let's define information policy $K$ by 
			\begin{align*}
				K(x) = \begin{cases}
					I_{F_0}(x)&, 0\le x\le \theta_\ell,\\
					J(x) &, x\ge \theta_\ell.
				\end{cases}
			\end{align*}
			We have $K\ge J$, so $K$ induces a weakly lower $\underline c_\lambda$ than $J$. Hence, by $\gamma\ge 0$, it suffices to verify that the expected Receiver's action is weakly higher under $K$ than under $J$. Because $p\ge x_0$, the argument in the proof of Proposition \ref{app:prop:upper} suffices. Specifically, by Lemma \ref{app:lem:uniqueness}, we have
			\begin{align*}
				W(K')-W(J') =&\int _{[0,\theta_\ell]} \big ( V_\lambda (\Delta K (c)) -V_\lambda (\Delta J (c)) \big ) \frac{\partial \gcond}{\partial c}(c|\lambda)\diff c \\ 
				&\ge 0,
			\end{align*}
			in which the inequality follows from the definition of $I^\star$, which includes $p\ge \theta_\ell$. Hence $K$ is a bi-upper censorship that improves upon $I$, for arbitrary $I$, in terms of $U_G$.
		\end{proof}

		\section{Supplementary material} \label{app:sec:supp}
		\subsection{Numerical example for Theorem \ref{thm:mechanisms}}

		In this section, we illustrate Theorem \ref{thm:mechanisms} using an example with a persuasion mechanism with two information policies. We assume that the state, the outside option, and the attention cost are binary, and that $k(e)=e$.

		We consider a particular case of the model in Section \ref{sec:model}, in terms of $k$ and the support of the type and state: $\theta\in \{0,1\}$, $c\in\{0.6, 0.75\}$, $\lambda \in\{0,0.1\}$, and $k(e)=e$. We assume that $\theta=1$ with probability $1/2$, and define the information policies $I_A$ and $I_B$ as follows. $I_A$ is the information policy of the experiment $A$, corresponding to the unique Bayes-plausible mean distribution supported in $\{0, 5/6\}$, and $I_B$ is the information policy of the experiment $B$, corresponding to the distribution supported in $\{2/7, 1\}$; see Figure \ref{fig:supex}. The experiments can be visualized with the following matrices, which record the conditional probability of a (column) message given a (row) state, and are not Blackwell ranked.
		
			 \begin{figure}[t!]
			\centering
			\subfloat[The upper envelope $I_C$ of the information policies $I_A$ and $I_B$.\protect\label{fig:supex:A}]{
				\begin{tikzpicture}
					\begin{axis}[
						xlabel={posterior mean},
						xticklabels={$0$, $2/7$, $x_0$, $2/3$, $0.8$, $1$}, 
						xtick={0, 2/7, 0.5, 0.6667, 0.8,  1},
						ytick = \empty,
						yticklabels = \empty,
						xmax=1.03,
						ymax=0.5
						]
						\addplot[domain=0:1, thick, green] {max(0.0006, x-0.5, (0.4)*x)}; \addlegendentry{$I_A$};
						\addplot[domain=0:1, ultra thick, densely dotted, red] {max(0.0012, x-0.5, 0.7*(x-2/7))}; \addlegendentry{$I_B$};
						\addplot[domain=0:1, ultra thick, dashed, blue] {max(0.0012, x-0.5, 0.4*x, 0.7*(x-2/7))}; \addlegendentry{$I_C$};
						\addplot[domain=0:1, name path = F, very thick, dashed, opacity=0] {0.5*x}; 
						\addplot[domain=0:1, name path = N, very thick, densely dotted, opacity=0] {max(0.0012,x-0.5)}; 
						\addplot[color=black, fill=black, fill opacity=0.15] fill between [of=F and N,];
					\end{axis}
				\end{tikzpicture}
			}\quad 
			\subfloat[The net informativeness of the information policies $I_A,\,I_B,\,I_C$. For all $\lambda$, type $(0.6, \lambda)$ chooses effort 1, type $(0.75, 0)$ chooses effort 1, and type $(0.75, 0.1)$ chooses effort 0.\protect\label{fig:supex:B}]{
				\begin{tikzpicture}
					\begin{axis}[
						xlabel={posterior mean},
						xticklabels={$0$, $x_0$, $0.6$, $0.75$, $1$}, 
						xtick={0, 0.5, 0.6, 0.75,  1},
						ytick = {0.1},
						yticklabels = {0.1},
						xmax=1.03,
						ymax=0.5
						]
						\addplot[domain=0:1, thick, green] {max(0.0012, x-0.5, (0.4)*x)-max(0.0012, x-0.5)}; \addlegendentry{$\Delta I_A$};
						\addplot[domain=0:1, ultra thick, densely dotted, red] {max(0.0012, x-0.5, 0.7*(x-2/7))-max(0.0012, x-0.5)}; \addlegendentry{$\Delta I_B$};
						\addplot[domain=0:1, ultra thick, dashed, blue] {max(0.0012, x-0.5, 0.4*x, 0.7*(x-2/7)) - max(0.0012, x-0.5)}; \addlegendentry{$\Delta I_C$};
						\addplot[domain=0:1, dotted] {0.1};
						\addplot[dotted, forget plot] coordinates {(0.6, 0) (0.6, 0.3)};
						\addplot[dotted, forget plot] coordinates {(0.75, 0) (0.75, 0.3)};
					\end{axis}
				\end{tikzpicture}
			}
			\caption{Panel (a) illustrates the upper envelope $I_C$ of the information policies in the IC persuasion mechanism; panel (b) illustrates the choices of effort. The shaded region illustrates the set of information policies (the state is uniformly distributed in $\{0,1\}$). The information policies $I_A$ and $I_B$ cross at $x=2/3$.}
			\label{fig:supex}
		\end{figure}
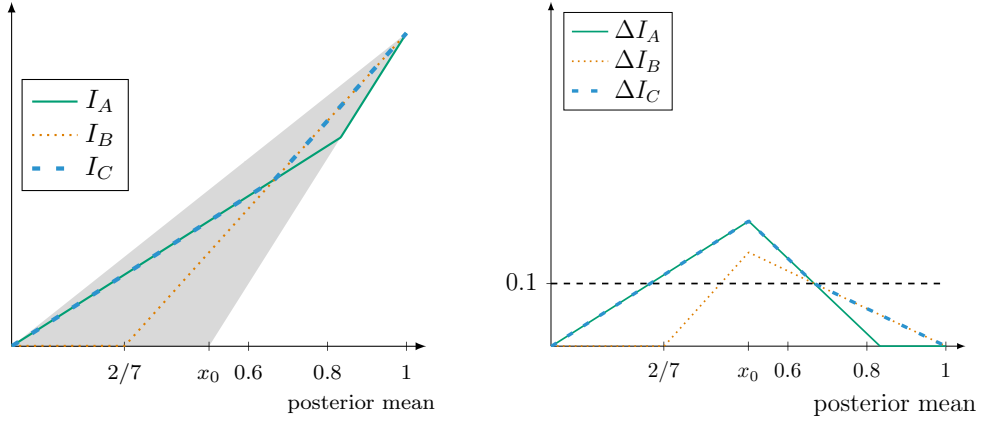

		\[
		A: \quad \begin{array}{c|cc}
			& L_A & H_A \\
			\hline
			0 & 0.8 & 0.2 \\
			1 & 0 & 1 
		\end{array}, \quad 
		\quad B: \quad \begin{array}{c|cc}
			& L_B & H_B \\
			\hline
			0 & 1 & 0 \\
			1 & 0.4 & 0.6
		\end{array}.
		\]
		
		Because only messages $H_A$ and $H_B$ lead to action 1, the expected material payoff $\overline U_i(c)=\int_{[0,1]}(x-c)_+\diff I_i'(x)$ given outside option $c$ and experiment $i$  is given by
		\begin{align*}
			\overline U_A(0.6)=7/50, \quad \overline U_B(0.6)=6/50,\quad \overline U_A(0.75)=0.05, \quad  \overline U_B(0.75)=0.075.
		\end{align*}
		
		The following persuasion mechanism is IC:
		\begin{align*}
			I_{(c, \lambda)} = \begin{cases}
				I_A, &\text{if} \ c=0.6\ \text{and} \ \lambda \in\{0, 0.1\},\\
				I_B, &\text{if} \ c=0.75\ \text{and} \ \lambda \in\{0, 0.1\},
			\end{cases}
		\end{align*}
		and induces the following action and effort choices: for all $\lambda$, type $(0.6, \lambda)$ chooses action 1 with probability 0.6, that is, the probability of $H_A$, and effort 1; type $(0.75, 0)$ chooses action 1 with probability 0.3, that is, the probability of $H_B$, and effort 1; type $(0.75, 0.1)$ chooses 1 with probability 0, which is her optimal action without information, and effort 0.
		
		We claim that the experiment $C$ induces the same effort and action choice of every type.
		\[
		C: \quad \begin{array}{c|ccc}
			& L_C & M_C & H_C \\
			\hline
			0 & 0.8& 0.2 & 0 \\
			1 & 0 & 0.4 & 0.6
		\end{array}
		\]
		We focus on experiment $C$ because it induces the information policy
		\begin{align*}
			I_C \colon x \mapsto \max \{I_A(x), I_B(x)\},
		\end{align*}
		 as in the construction for Theorem \ref{thm:mechanisms}. We start by observing that $C$ induces the mean distribution given by: $0$ with probability $0.4$, $2/3$ and 1 with probability $0.3$ each. Hence, $M_C$ leads to action 1 only for outside option $0.6$, whereas $H_C$ leads to action 1 for both outside options. As an implication, if type $(0.75, 0.1)$ chooses effort 0 and all other types choose effort 1, the mechanism is equivalent to $I_C$. The expected material payoff $\overline U_C(c)=\int_{[0,1]}(x-c)_+\diff I_C'(x)$ is given by
		\begin{align*}
			\overline U_C(0.6)=7/50, \quad \overline U_C(0.75)=0.075.
		\end{align*}
		Therefore, the mechanism is equivalent to $I_C$ in the sense of Definition \ref{def:mechanisms:eq}.

		\subsection{Symmetric-information benchmark} \label{app:sec:KG}
		
		For this section, the type distribution puts full mass at $(\zeta, \kappa)$, $k$ is linear, $\kappa>0$, and $F_0$ admits a density. Sender's \emph{problem} is:
		\begin{align*}
			\max _{I\in \mathcal I}  ( 1- I'(\zeta^-))  \mathbf 1_{
				\{I\in \mathcal I \, \mid \, \Delta I (\zeta) \geq \kappa\}
			}(I),
		\end{align*}
		because an experiment $F$ is an equilibrium experiment iff $I_F$ solves the above maximization, due to a generalization of the argument of \citet{gentzkow_rothschild-stiglitz_2016}. If $\zeta < x_0$, then $I_{\overline F}$ is optimal. Let $1\geq \zeta \geq x_0$.
		\begin{lemma}\label{app:lem:KG}
			There exists $\theta \in [0,\zeta]$ such that: $I^\theta $ solves Sender's problem and $\Delta I^{\theta} \leq \kappa$, with equality if $\theta>0$. 
		\end{lemma}
		\begin{proof}
			Let $\mathcal I^u \coloneqq \{I\in \mathcal I \mid \text{there exists} \ \theta\in[0,\zeta] \ \text{such that} \  I = I^{\theta} \}$. Without loss of optimality by Lemma \ref{app:lem:knownzeta}, we consider solutions in $\mathcal I^u$. Suppose there exists a solution $I \in \mathcal I^u$, such that $I=I_{\theta^\star}$, for some $\theta^\star \in (0,1)$. We distinguish three cases. 
			
			(1) If $\Delta I (\zeta) < \kappa$, then Sender is indifferent between $I$ and $I_{\overline F}$, so the result holds. (2) If $\Delta I (\zeta) = \kappa$, the result holds. (3) If $\Delta I (\zeta) > \kappa$, then, by definition of $I$ at $y = I (\zeta)$,
			\begin{align*}
				I_{F_0} (\theta^\star) + F_0(\theta^\star) (\zeta - \theta^\star) - y=0.
			\end{align*}
			By the implicit function theorem, there exists a differentiable function $t$ such that $t \colon y \mapsto \theta^\star$ and, for $ t(y) < \zeta $ 
			\begin{align*}
				t'(y) = \frac{1}{(\zeta-t(y))\frac{\partial F_0}{\partial \theta} (t(y))}.
			\end{align*}
			Let's define the value of Sender's expected utility given $I^{\theta}$ as $v \colon (0,1) \to [0,1]$ such that $ v\colon \theta \mapsto 1- (I^{\theta})'(\zeta^-)$.
			Because $ I_{\theta^\star}'(\zeta^-) = F_0(\theta^\star)$, $v$ is differentiable in $\theta$ at $\theta^\star$. The derivative of $v$ with respect to $I(\zeta)$ is:
			\begin{align*}
				-\frac{\partial F_0}{\partial \theta} (t(I(\zeta))) \frac{1}{(\zeta-t(I(\zeta)))\frac{\partial F_0}{\partial \theta} (t(I(\zeta)))},
			\end{align*} 
			if $\zeta > t(I(\zeta))$. It follows that we can consider without loss solutions $I\in \mathcal I^u$ that satisfy: $\Delta I^{\theta}(\zeta) = \kappa$ and $I=I^{\theta}$, or $\Delta I(\zeta) < \kappa$.
		\end{proof}

		\subsection{Auxiliary results and omitted proofs}
		\begin{lemma}[Subdifferential of convex functions] \label{app:fact:convex}
			Let $S\subseteq \mathbb R$, $f\colon S \to \mathbb R$ be convex and $\varphi \colon \mathbb R \to \mathbb R$ be a nondecreasing convex function on the range of $f$. It holds that:
			\begin{enumerate}
				\item The function $\varphi \circ f$ is convex on $S$;
				\item For all $y\in S$, letting $t=f(y)$, we have
				\begin{align*}
					\{\alpha u \mid (\alpha, u) \in \partial \varphi (t)\times \partial f(y)\} =	\partial \varphi \circ f(y).
				\end{align*}
			\end{enumerate}
		\end{lemma}
		\begin{proof}
			See \citet[Proposition 8.21 and Corollary 16.72]{bauschke_convex_2017}.
		\end{proof}
%
%
%
	
		\begin{lemma}[Implications of the Envelope Theorem]\label{app:lem:ET}
			Let $f\colon [0,1]^2 \to \mathbb R$ exhibit increasing differences and be such that: $f(\cdot, a)$ is continuous for all $a\in[0,1]$, $f(e, \cdot)$ is nondecreasing for all $e\in[0,1]$, the derivative with respect to the variable $a$, $\frac{\partial f}{\partial a} (e, \cdot)$, exists, and is bounded for all $e\in[0,1]$. The following hold.
			\begin{enumerate}
				\item We have $\argmax_{e \in [0,1]} f(e, a)\ne \emptyset$ for all $a\in[0,1]$.
				\item The function $a\mapsto \max _{e \in [0,1]} f(e, a)$ is nondecreasing and absolutely continuous.
				\item If $a\mapsto \frac{\partial f}{\partial a}(e, a)$ is nondecreasing for all $e\in[0,1]$, then $a\mapsto \max _{e \in [0,1]} f(e, a)$ is convex.
				\item If $f$ exhibits strictly increasing differences, $a\mapsto \frac{\partial f}{\partial a}(e, a)$ is nondecreasing, $f(e, \cdot)$ is increasing for all $e\in(0,1]$, $ \argmax_{e \in [0,1]} f(e, a) \cap (0,1] \ne \emptyset$, and $1\ge a' >a\ge 0$, then
				\begin{align*}
					\max _{e \in [0,1]} f(e, a')>\max _{e \in [0,1]} f(e, a).
				\end{align*}
			\end{enumerate}
		\end{lemma}
		\begin{proof}
			By upper semi-continuity of $f$, $\argmax_{e \in [0,1]} f(e, a)\ne \emptyset$, so 1.\ holds. Then, by the increasing-differences property of $f$, there exists a nondecreasing selection $e^\star \colon a\mapsto \argmax_{e \in [0,1]} f(e, a)$ on $[0,1]$ \citep{topkis_minimizing_1978}. By our hypotheses, we apply the envelope theorem \citep{milgrom_envelope_2002}, letting $V(a) := \max _{e \in [0,1]} f(e, a)$, to establish that $V$ is absolutely continuous and
			\begin{align*}
				V(a) = V(0) + \int _{[0,a]} \frac{\partial f}{\partial a} (e^\star(\tilde a), \tilde a) \diff \tilde a.
			\end{align*}
			$V$ is nondecreasing because $\frac{\partial f}{\partial a}\ge0 $. Hence, 2.~holds.
			
			Let's establish that $V$ is convex if $a\mapsto \frac{\partial f}{\partial a}(e, a)$ is nondecreasing. By the increasing-differences property of $f$: (i) $e\mapsto  \frac{\partial f}{\partial a} (e,  a) $ is nondecreasing, and (ii) there exists a nondecreasing $e^\star \colon a\mapsto \argmax_{e \in [0,1]} f(e, a)$. As a result, $a\mapsto \frac{\partial f}{\partial a}( e^\star (a), a)$ is nondecreasing. Thus, $V$ is convex \citep[Theorem 24.8]{rockafellar_convex_1970}. Hence, 3.~holds.
			
			Let $a'>a$, for $a', a\in[0,1]$, and $e' \in \argmax_{e \in [0,1]} f(e, a) \cap (0,1] $. Then: $V(a') - V(a) = \int  _{[a, a']} \frac{\partial f}{\partial a} (e^*(\tilde a), \tilde a) \diff \tilde a$ for every selection $e^*$ of $\argmax_{e \in [0,1]} f(e, a) \cap (0,1]$. We have the following chain of inequalities under the additional hypotheses stated in part 4.:
			\begin{align*}
				V(a') - V(a) &\ge \int_{[a, a']} \frac{\partial f}{\partial a} (e', \tilde a) \diff \tilde a \\
				& \ge  \int_{[a, a']} \frac{\partial f}{\partial a} (e', a) \diff \tilde a,
			\end{align*}
		in which the first inequality follows from the strict increasing-differences property of $f$ and the definition of $e'$, the second inequality holds because $a\mapsto \frac{\partial f}{\partial a}(e, a)$ is nondecreasing (for the first inequality, in particular, we note that: (i) every selection $e^*$ of $\argmax_{e \in [0,1]} f(e, a) \cap (0,1]$ is nondecreasing, (ii) there exists a selection $e^*$ of $\argmax_{e \in [0,1]} f(e, a) \cap (0,1]$ such that $e^*(a)=e'$). Part 4 holds because $\int _{[a, a']} \frac{\partial f}{\partial a} (e', a) \diff \tilde a = (a'-a) \frac{\partial f}{\partial a} (e', a) $.
		\end{proof}
		
		\paragraph{Proof of Lemma \ref{app:lem:continuity}}
		
		\begin{proof}
			Fix $F,H\in\mathcal F$. The value function $V_\lambda^0$ is $1$-Lipschitz. Hence, for every
			$c\in[0,1]$,
			\begin{align*}
				\left|
				V_\lambda^0(\Delta I_H(c))
				-
				V_\lambda^0(\Delta I_F(c))
				\right|
				&\leq
				\left|
				\Delta I_H(c)-\Delta I_F(c)
				\right|\\
				&\leq
				\int_{[0,c]}|H(x)-F(x)|\diff x\\
				&\leq
				\|H-F\|_1.
			\end{align*}
			Therefore,
			\begin{align*}
				|W_\lambda(H)-W_\lambda(F)|
				\leq
				\|H-F\|_1
				\int_{[0,1]}
				\left|
				\frac{\partial \gcond}{\partial c}(c\mid\lambda)
				\right|\diff c.
			\end{align*}
			Integrating with respect to $G$ gives
			\begin{align*}
				|W(H)-W(F)|
				\leq
				\|H-F\|_1
				\int_{[0,1]}\int_{[0,1]}
				\left|
				\frac{\partial \gcond}{\partial c}(c\mid\lambda)
				\right|
				\diff c\diff \Gmarg(\lambda).
			\end{align*}
			Thus $W$ is continuous on $\mathcal F$.
		\end{proof}
	\paragraph{Proof of Lemma \ref{app:lem:knownzeta}}
		\begin{proof}[Proof of Lemma \ref{app:lem:knownzeta}]
				Let $\zeta \in [0,1]$. Let $M:= \{m \in [0, I'(\zeta^-)] : I(\zeta) + m(x-\zeta) \le I_{F_0} (x) \ \text{for all} \ x \in [0, \zeta]\}$, and $m:=\min M$. We construct an information policy starting from the line $x\mapsto I(\zeta) + m (x-\zeta)$, via the next claims.
				
				\emph{(1) $m$ is well-defined.} (i) $M$ is nonempty, because $0\le I'(\zeta^-)\le 1$ (which follows from $I\in \mathcal I$), $I'(\zeta^-) \in \partial I(\zeta^-)$ and $I(x) \leq I_{F_0}(x)$ for all $x$; (ii) $M$ is closed, because the mapping $m\mapsto I(\zeta) + m (x-\zeta)$ is a continuous function on $[0, I'(\zeta^-)]$; (iii) $M$ is bounded because $I'(\zeta^-)\le 1$, from $I\in \mathcal I$.
				
				\emph{(2) There exists $\theta \in [0,\zeta]$ such that $I_{F_0}(\theta) = I(\zeta) + m (\theta-\zeta )$.} If $m=0$, then $0=I_{F_0}(0)\ge I(\zeta) \ge 0$. Hence, taking $\theta =0$ verifies our claim. Let $m>0$, and suppose there does not exist  $\theta \in [0, \zeta]$ such that $I_{F_0}(\theta) = I(\zeta) + m (\theta-\zeta )$. There exists $\overline \varepsilon>0$ such that: $I(\zeta) + (m-\varepsilon)(x-\zeta)< I_{F_0}(x)$ for all $x\in[0,\zeta]$ and $0<\varepsilon\le \overline \varepsilon$. Moreover, for a sufficiently small $\varepsilon>0$, we have $m-\varepsilon \in M$. Thus, we have a contradiction with the definition of $m$. 
				
				\emph{(3) $m \in \partial I_{F_0}(\theta)$ and  $I(\zeta) + m(x-\zeta)  = I_{F_0}(\theta) + (x-\theta) {F_0}(\theta) $ for all $x$.} First, we argue that $m \in \partial I_{F_0}(\theta)$. By convexity of $I_{F_0}$ and definition of $\theta$,  $x\mapsto I(\zeta) + m(x-\zeta) $ is tangent to $I_{F_0}$ at $\theta$. Thus, $m$ is a subgradient of $I_{F_0}$ at $\theta$. Now, we argue that  $I(\zeta) + m(x-\zeta)  = I_{F_0}(\theta) + (x-\theta) {F_0}(\theta) $ for all $x$. $m=F_0(\theta)$ because $I_{F_0}$ is differentiable (by the fact that $F_0(x^-)=F_0(x), x\in\mathbb R$). The equality follows because $x\mapsto I(\zeta) + m(x-\zeta) $ is equal to $I_{F_0}$ at $x=\theta$.

				We define the following function.
				\begin{align*}
					I^u \colon 
					& x \mapsto \begin{cases}
						I_{F_0} (x) &, x \in [0, \theta] \\
						I(\zeta) + m(x-\zeta) &, x \in (\theta, \zeta] \\
						\max \{ I(\zeta) + m(x-\zeta), I_{\overline F}(x)  \} &, x\in (\zeta, \infty).
					\end{cases}
				\end{align*}
				
				Now, we claim that $I^u=I^\theta$. It suffices to show that: (i) for some $x_{u} \in [0,1]$
				\begin{align*}
					I^u(x) = \begin{cases}
						I_{F_0}(x) &, x \in [0, \theta] \\
						I_{F_0}(\theta) + (x-\theta) F_0(\theta) &, x \in (\theta, x_u] \\
						I_{\overline F} (x) &, x \in (x_u, \infty),
					\end{cases}
				\end{align*}
				and (ii) $I^u\in \mathcal I$. We claim that (i) holds by means of the next three claims.

				\emph{There exists $x_u \in [\zeta, 1]$ such that: }
				\begin{align*} 
					I(\zeta) + m(x-\zeta) \geq I_{\overline F}(x) &\quad , x \in [0, x_u]\\
					I(\zeta) + m(x-\zeta) \le I_{\overline F}(x) &\quad, x \in (x_u, 1].
				\end{align*}
				Let's note that: (a) $I(\zeta)  \geq I_{\overline F}(\zeta)$; (b) by $m \in \partial I_{F_0}(\theta)$ and $I_{F_0}(1) = I_{\overline F}(1)$, we have that $I_{\overline F}(1) \geq I(\zeta) + m(1-\zeta)$, and (c) the two functions, $x\mapsto I(\zeta) + m(x-\zeta)$ and $I_{\overline F}$, are affine with slopes, respectively, $m$ and $1$, such that: $m\leq 1$.
				
				We proceed to verify that (ii) holds, i.e.\ $I^u\in \mathcal I$, via the next two claims.

				\emph{(1) $I_{\overline F}(x) \le I^u(x)\le I_{F_0}(x)$ for all $x\in \mathbb R_+$, and $I^u$ is locally convex at all $x\notin \{\theta, x_u\}$.}
				If $x \in[0, \theta)$, then $I^u$ is locally convex and $I_{\overline F} (x) \le I^u(x) \le I_{F_0}(x)$. If $x \in (\theta, \zeta)$, then $I^u$ is affine, $I_{\overline F}(x) \leq I(x) \le I^u(x)$ by construction of $I^u$ and definition of $I$, and $I^u(x)\le I_{F_0}(x)$ by $m\in \partial I_{F_0}(\theta)$. If $x\in [\zeta, \infty)$, then $I^u$ is locally convex (because it is the maximum of affine functions), $I_{\overline F} (x) \le I^u(x)$ by construction of $I^u$, $I^u(x)\le I_{F_0}(x)$ because: (i) $m\in \partial I_{F_0}(\theta )$ and (ii) $I_{\overline F}(x) \le I_{F_0}(x)$. To verify global convexity, it suffices to verify the next claim.

				\emph{(2) The subdifferential of $I^u$ at $x \in \{\theta, x_u\}$ is nonempty.} First, we argue that $m$ is a subgradient of $I^u$ at $\theta$. This claim follows from the fact that the slope of $I^u$ at $\theta$ is a subgradient of $ I_{F_0}$ at $\theta$, and $I^u(\theta) = I_{F_0}(\theta)$. On $[0, \theta]$, $I^u=I_{F_0}$, and on $[x_u, \infty)$ $I^u $ is above the line $x\mapsto I(\zeta) + m(x-\zeta)$. Thus, $m\in \partial I^u(\theta)$. Second, the fact that $m$ is a subgradient of $I^u$ at $x_u$ follows from the definition of $x_u$.

				We established that $I^u(x)=I^\theta(x)$ for all $x\in[0,1]$. (1.)\ and (2.)\ hold by construction.
			\end{proof}

	\newpage
	\bibliographystyle{te}
	\phantomsection
	\addcontentsline{toc}{section}{\refname}
	\bibliography{extbib.bib}

\end{document}